\documentclass[preprint,12pt]{elsarticle}

\usepackage{CJK}
\usepackage[top=2cm, bottom=2cm, left=2cm, right=2cm]{geometry}
\usepackage{algorithm}
\usepackage{algorithmicx}
\usepackage{algpseudocode}
\usepackage{amsmath}
\usepackage{indentfirst}
\usepackage{threeparttable} 
\usepackage{amssymb}
\usepackage{amsthm}
\usepackage{graphicx}
\usepackage{subfigure} 
\usepackage{caption}
\usepackage{wrapfig} 

\usepackage{colortbl} 
\usepackage{xcolor}
\usepackage{array} 
\usepackage{multirow}
\usepackage{float}
\usepackage{threeparttable}

\usepackage{appendix}

\usepackage{chemformula}

\newtheorem{theorem}{Theorem}

\journal{Knowledge-Based Systems}

\begin{document}

\begin{frontmatter}



\title{Consecutive Knowledge Meta-Adaptation Learning for Unsupervised Medical Diagnosis}


\author[inst1]{Yumin Zhang}

\affiliation[inst1]{organization={Department of Statistics, Jinan University},
            city={Guangzhou},
            postcode={510632},
            country={China}}
\author[inst1]{Yawen Hou}
\affiliation[inst2]{organization={Baidu Inc.},
            city={Beijing},
            postcode={100080}, 
            country={China}}
\author[inst2]{Xiuyi Chen\corref{cor1}}
\cortext[cor1]{Corresponding author.}
\ead{chenxiuyi01@baidu.com}
\author[inst3]{Hongyuan Yu}
\affiliation[inst3]{Multimedia Department
Xiaomi Inc.}
\author[inst2]{Long Xia}

\begin{abstract}
Deep learning-based Computer-Aided Diagnosis (CAD) has attracted appealing attention in academic researches and clinical applications. 
Nevertheless, the Convolutional Neural Networks (CNNs) diagnosis system heavily relies on the well-labeled lesion dataset, and the sensitivity to the variation of data distribution also restricts the potential application of CNNs in CAD.
Unsupervised Domain Adaptation (UDA) methods are developed to solve the expensive annotation and domain gaps problem and have achieved remarkable success in medical image analysis.
Yet existing UDA approaches only adapt knowledge learned from the source lesion domain to a single target lesion domain, which is against the clinical scenario: the new unlabeled target domains to be diagnosed always arrive in an online and continual manner. 
Moreover, the performance of existing approaches degrades dramatically on previously learned target lesion domains, due to the newly learned knowledge overwriting the previously learned knowledge  (\emph{i.e.}, catastrophic forgetting).
To deal with the above issues, we develop a meta-adaptation framework named \underline{C}onsecutive \underline{L}esion \underline{K}nowledge \underline{M}eta-Adaptation (CLKM), which mainly consists of Semantic Adaptation Phase (SAP) and Representation Adaptation Phase (RAP) 
to learn the diagnosis model in an online and continual manner.
In the SAP, the semantic knowledge learned from the source lesion domain is transferred to consecutive target lesion domains.
In the RAP, the feature-extractor is optimized to align the transferable representation knowledge across the source and multiple target lesion domains.
Furthermore, to alleviate catastrophic forgetting during the continuous unsupervised domain adaptation, a domain-quantizer is designed for reserving utility knowledge from previously learned target lesion domains.
Besides, since the semantic knowledge is continuously changing and artificial-designed kernels can't model such complex distribution, a self-adaptive kernel is designed to help measure the shift across the source and consecutive target lesion domains flexibly in the SAP.
Experimental results show the effectiveness of CLKM on medical images analysis tasks.

\end{abstract}



\begin{keyword}
Unsupervised Domain Adaptation \sep Medical Lesion Diagnosis
\sep Continual Learning
\end{keyword}

\end{frontmatter}


\section{Introduction}

In medical imaging field, Computer-Aided Diagnosis (CAD) \cite{ahmadi2019computer} is the computer-based system that helps doctors to take decisions swiftly.
With the development of artificial intelligence, deep learning-based CAD has attracted appealing attention in academic researches, where deep neural networks (e.g. CNNs) are trained on dense annotated medical data for medical lesion diagnosis~\cite{zhao2018lung,liu2018segmentation,lakhani2017deep}.
Recently, there are numerous clinical applications of deep learning-based CAD, such as diabetic retinopathy diagnosis~\cite{ting2017development}, epithelium-stroma classification~\cite{huang2017epithelium}, holistic classification of CT patterns for interstitial lung disease and liver segmentation~\cite{dou20163d}.

However, above approaches of deep learning-based CAD  are data-hungry while there usually does not exist sufficient medical lesion data due to the time-consuming label-intensive annotation and specialized biomedical knowledge.
Furthermore, in clinical applications, due to imaging protocol, device vendors, personal physical, and so on, there is a significant domain distribution discrepancy between medical training data (a.k.a. source lesion domain) and medical testing data (a.k.a. target lesion domain). 
The gap across domains causes the lesion diagnosis model trained on source medical lesion (a.k.a. source lesion domain) to suffer from dramatic performance degradation when performing on the target medical lesion (a.k.a. target lesion domain).
Thus, Unsupervised Domain Adaptation (UDA) methods \cite{perone2019unsupervised,xie2022unsupervised} are proposed to address the above challenges by adapting a trained model on well-labeled source lesion domain to a different but relevant unlabeled target lesion domain.

Generally speaking, current UDA-based medical lesion diagnosis methods can be categorized into three main streams.
The first stream aims to build domain-invariant feature spaces through adversarial learning~\cite{ganin2016domain}.
For example, Han et al. \cite{han2018gan}~utilized Generative Adversarial Network (GAN) to generate synthetic brain Magnetic Resonance (MR) samples for data augmentation.  
The second stream aims to mitigate the distribution discrepancy across domains via a reasonable metric.  
Specifically, Tzeng et al.~\cite{tzeng2014deep} constructed domain confusion loss to make domains indistinguishable.
A similar strategy is used in the COVID-19 diagnosis task~\cite{niu2021distant}.
The third stream is pseudo-labeling~\cite{lee2013pseudo} which 
predicts the labels of unsupervised target data and improves the performance  of the lesion diagnosis model via iterative training.
In the whole heart segmentation task~\cite{wang2021few}, the shape-constrained network was designed to evaluate the quality of the predictions and then to select select high-quality pseudo labels, which is the key to success in the self-training stage.
Above UDA-based medical diagnosis methods~\cite{han2018gan,niu2021distant,wang2021few} have achieved great success on discrete target lesion domains adaptation.
However, the key to adversarial learning is to distinguish discrete domains and thus can't apply to consecutive target domains directly. Besides, due to complex target distributions, reasonable metrics and reliable pseudo-labels are hardly accessible in the real-clinical environment.
Therefore, 
we here focus on the second stream, and deal with more complex distribution during the continuous unsupervised domain adaptation.

\begin{figure}[t!]
 \center{\includegraphics[width=17cm]{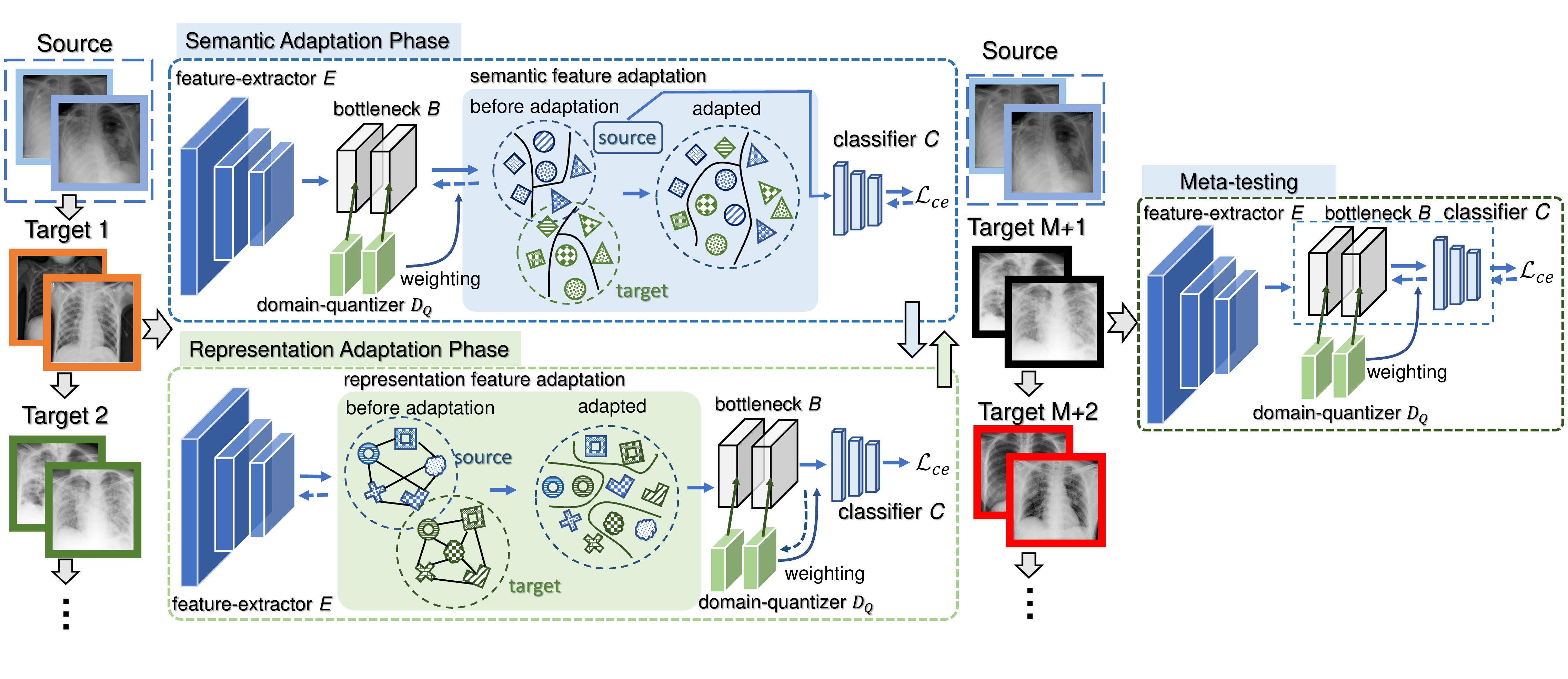}}
    \caption{
    Demonstration of Consecutive Lesion-Domain Adaptation (CLDA). 
    The source lesion data is well-labeled and stationary, while the unlabeled target lesion data to be diagnosed arrive in an online and continual manner.
    }
    \label{fig:fig1}
\end{figure}

To better illustrate the clinical problem that consecutive unlabeled target lesion domains to be diagnosed, we here leverage pneumonia caused by various factors as a possible example and provided in Figure~\ref{fig:fig1}. 
Suppose we have a diagnosis system trained in a fixed scenario with well-labeled datasets (bacterial pneumonia) and apply the system in the clinical environment (COVID-19 and COVID-19 variants to be diagnosed).
Since the strong infectivity of COVID-19 and its variants, it's hard to obtain reliable annotated lesion data in a short time.
Besides, due to the mutation of COVID-19 quickly, the diagnosis system is in a dynamic clinical environment, facing consecutive target lesion domains to be diagnosed.
Although such a diagnosis system performs well on the source lesion data, correctly diagnosing the target lesion data collected from a real-clinical environment is difficult.
Because the characteristics of images are continually changing over consecutive target lesion domains,
there is domain shift and the diagnosis system would have dramatic performance degeneration.
Moreover, the inevitable medical data privacy is another inescapable challenge.
The above issues make the UDA-based works in medical lesion diagnosis~\cite{ouyang2019data, ahn2020unsupervised} has drawn surging attention recently.
Meanwhile, as the target lesion data arrive online, the currently learned knowledge will overwrite the previously learned knowledge.
Such methods~\cite{han2018gan,niu2021distant,wang2021few,ouyang2019data, ahn2020unsupervised} face another challenge: catastrophic forgetting~\cite{mccloskey1989catastrophic}.
To cope with the above challenges, in our work, a realistic medical lesions diagnosis setting is considered: it focuses on the medical diagnosis model trained on a well-labeled source lesion domain that can adapt to the continuously arriving unlabeled target lesion domains online while maintain the diagnostic ability on previous learned domains.
We formulate this setting as Continuous Lesion-Domain Adaptation (CLDA) which is discribed as follows:

$\bullet$ \textbf{Continuous Diagnosis Ability.} 
Generally, the medical diagnosis system is trained under a fixed scenario yet applied in a dynamic clinical environment. 
Thus, the system should have the continuous diagnosis ability, which is crucial to its deployment.
Moreover, appropriate metric across the source and target lesion domains is key to continual adaptation. 
The effectiveness of Maximum Mean Discrepancy (MMD)~\cite{gretton2012kernel}, widely used in domain adaptation, depends on the prior designed kernel. 
However, as the semantic knowledge of consecutive target lesion
domains changing continuously, setting appropriate artificial-designed kernels to model such complex distributions is difficult~\cite{liu2020learning1}.
To deal with this, we propose a self-adaptive kernel to measure distribution divergence across the source and target lesion domains adaptively.

$\bullet$ \textbf{Catastrophic Forgetting.} 
The catastrophic forgetting~\cite{mccloskey1989catastrophic} in neural networks during adaptation makes the adapted model perform poorly on the previously learned lesion domains.
Generally, there have been many efforts to overcome this problem.
For example, Elastic Weight Consolidation (EWC)~\cite{kirkpatrick2017overcoming} as a regularization method has been successfully applied in two common medical tasks~\cite{baweja2018towards} : segmentation of normal structures and segmentation of white matter lesion in brain MRI.
Besides, replay methods are proposed in the task of segmenting humerus and scapula bones on MR images~\cite{ozdemir2018learn} and chest CT classification~\cite{hofmanninger2020dynamic} to prevent catastrophic forgetting.
These methods either rely on supervised information to determine the importance of knowledge to preserve  \cite{baweja2018towards} or require storing previously learned data ~\cite{ozdemir2018learn,hofmanninger2020dynamic}. 
However, in the real-clinical environment, prior knowledge of the target lesion domain is usually absent, and it is unrealistic to store previously learned data due to the limited resources of the online diagnosis devices.

Therefore, aiming to tackle the above challenges, we propose Consecutive Lesion Knowledge Meta-Adaptation (CLKM), a meta-adaptation framework that purpose to adapt the diagnosis model to consecutive unlabeled target lesion data that coming online meanwhile to maintain performance on all the previous and the new learned target lesion domain.
Concretely, our framework consists of two phases: Semantic Adaptation Phase (SAP) and Representation Adaptation Phase (RAP), 
which is alternatively optimized via meta-training mechanism.
In the SAP, we continuously transfer learned semantic features from the source lesion domain to consecutive target lesion domains. 
In the RAP, we optimize the feature-extractor to further align the learned transferable representation across the source and target lesion domains, which makes the subsequent semantic adaptation phase more effective.
To prevent the diagnosis model from catastrophic forgetting during semantic adaptation, we introduce a domain-quantizer to learn how to select the knowledge that needs to be reserved via an anti-forgetting regularizer. 
As the unlabeled target lesion data are collected continuously, the semantic features change over time, which makes setting appropriate artificial-designed kernels to model such complex distributions impossible. 
Hence, the self-Adaptive Kernel-based Maximum Mean Discrepancy (AK-MMD) is designed to help measure the shift across the source and consecutive target lesion domains in the SAP.
In the meta-testing phase, the parameters of the feature-extractor and domain-quantizer are fixed, and we only fine-tune the last few layers (\emph{i.e.}, classifier and bottleneck) to minimize the consumption of computing resources.
Compared with other domain adaptation methods, the experimental results validate the effectiveness of CLKM in adapting knowledge and alleviating catastrophic forgetting.
Our main contributions in this work are summarized as follows.

\begin{enumerate}[(1)]
\item We introduce the continuous lesion diagnosis concept into UDA medical lesion diagnosis field and formulate a new but more realistic setting: CLDA. 
Besides, to address this issue, a meta-adaptation framework CLKM is proposed that mainly consists of two phases: SAP and RAP, which aims to adapt the diagnosis model to consecutive target lesion domains and alleviate catastrophic forgetting.

\item 
The SAP and RAP cooperate in the meta-adaptation framework via innerloop and outerloop to learn the fast adaptation meta-test features.
In the SAP, the semantic knowledge is adapted to consecutive target lesion domains preliminarily. 
In the RAP, we optimize the feature-extractor to align representation features across the source lesion domain and multiple target lesion domains.

\item 
We propose the domain-quantizer against catastrophic forgetting via reserving previously learned utility diagnostic knowledge. Moreover, the self-adaptive kernel is designed to help measure the discrepancy across the source and consecutive target domains, whereas the artificial-designed kernel is not suitable due to the complexity of distributions.




\item 
Extensive experiments on six medical images analysis tasks show that the proposed CLKM enjoys the superiority over other adaptation methods according the diagnosis accuracy,
and CLKM manages to adapt the diagnosis model to consecutive target lesion domains while maintain
the diagnosis knowledge of previous learned domains.

\end{enumerate}

\section{Related Works}

\subsection{Classic Unsupervised Domain Adaptation (UDA)}
Numerous UDA methods are developed recently, and the mainstream of UDA can be come down to as adversarial learning, metric learning, and pseudo-labeling.

\textbf{(1) Adversarial Domain Adaptation.}
The main strategy of adversarial learning~\cite{goodfellow2014generative} in UDA is to align the features of different domains via a game of generator and discriminator. 
For example, Genin et al.~\cite{ganin2016domain} utilized a domain-classifier to distinguish the different domains, and leveraged the generator to confuse the domain-clasifier via generating an aligned distribution across different domains.
Similarly, the Adversarial Discriminative Domain Adaptation was proposed in~\cite{tzeng2017adversarial}, they first pre-trained source encoder (generator) and then a shared domain classifier was used to align the output of feature encoders belong different domains. 
Pinheiro et al.~\cite{pinheiro2018unsupervised} proposed SimNet, leveraging a similarity-based classifier and adapting to the unlabeled domain via optimizing the adversarial loss.
Moreover, Zhang et al. \cite{zhang2018collaborative}~proposed the Collaborative and Adversarial Network (CAN), introducing the collaborative adversarial concept into UDA, and CAN further be extended as Incremental CAN by a pseudo-labeling method.
Classic methods \cite{tzeng2017adversarial,pinheiro2018unsupervised,zhang2018collaborative} distinguish the features of different domains by training domain classifier, which ignores the existence of each domain’s special characteristics. 
Without the domain classifier, Lee et al. \cite{lee2019drop} proposed to Drop to Adapt, which learns strongly differentiated features by enforcing the cluster assumption and adversarial dropout. 
Similarly, further considering the boundaries between classes,
Saito et al. \cite{saito2017adversarial}~leveraged the adversarial process to align domain distribution.
Besides, several methods also aim to align conditional or leverage joint distributions. 
For example, Long et al.~\cite{long2017conditional} focused on harnessing the cross-variance between representations in feature space and classifier prediction, and thus they designed two conditioning strategies and proposed a Conditional Domain Adversarial Network.
Cicek and Soatto~\cite{cicek2019unsupervised} trained a shared embedding to align the joint distributions over inputs and outputs simultaneously.
Recently, Dong et al.~\cite{dong2021and} proposed Knowledge Aggregation-induced Transferability Perception Adaptation Network to distinguish transferable or untransferable knowledge across domains to reduce the the negative transfer influence.

\textbf{(2) Metric Learning for Domain Adaptation.}
Metric learning methods aim to map inputs to embedding space, and then the discrepancy between different samples can be measured. 
From this concept, researchers constructed diverse metric-based loss functions, to mitigate the domain shift.
Among these methods, Maximum Mean Discrepancy (MMD)~\cite{tzeng2014deep} is widely used to measure the shift across different domains. 
Moreover, Li et al.~\cite{li2020maximum} proposed an adversarial tight match training strategy, leveraging maximum density divergence to qualify the distribution divergence.
In~\cite{sun2016return}, the domain discrepancy is mitigated by aligning the covariance of the source and target distributions. 
Besides, Luo et al.~\cite{luo2020unsupervised} adopted the affine Grassmann distance and the Log-Euclidean metric in unsupervised domain adaptation.

\textbf{(3) Pseudo-Labeling for Domain Adaptation.}
In such methods, the pre-trained model assigns pseudo labels to the unlabeled target data in advance and then is trained together with labeled source data.
As the model is updated iteratively, the pseudo labels of the target data will also be iterated and become more accurate.
In~\cite{zhang2017joint}, the pseudo labels for unsupervised target data were matched by a pre-trained classifier.
Another method to produce pseudo-label is clustering.
For instance, Sener et al.~\cite{sener2016learning} leveraged the K-Nearest Neighbor to obtain pseudo labels.
Notably, the accuracy of pseudo labels will affect model performance significantly. 
Therefore, the above methods are not applicable when data distributions embody complex multimodal structures. 
Some works~\cite{kang2019contrastive,tang2020unsupervised} rely on complex noise reduction to address such a limitation.
For example, Chen et al.~\cite{chen2019progressive} utilized an Easy-to-Hard Transfer Strategy to select reliable pseudo-labeled target data, then introduced an Adaptive Prototype Alignment to minimize the negative influence of falsely-labeled samples.

However, we cannot apply the above methods in CLDA directly. 
The reason is that the domain discriminator successfully distinguish two discrete domains is the key to adversarial learning, while in our setting, the target lesion domains are arriving continuously.
Besides, metric learning relies on appropriate artificial-designed kernels, which is impossible facing complex distribution. 
For pseudo-labeling methods, we have not pre-trained models to generate reliable pseudo-labels.

\subsection{Medical Imaging Analysis}
Medical imaging classification plays a critical role in many clinical environments. 
Since the variety of imaging protocols, device vendors, and patient populations, domain shift in the clinical environment is a frequent problem.
Numerous approaches are proposed to tackle this problem, such as, Zhang et al.~\cite{zhang2018task} transfered the knowledge learned from CT images to X-ray images via CycleGAN. 
Chen et al.~\cite{chen2019synergistic} conducted domain translation from MR to CT domain for heart segmentation problems.
Dou et al.~\cite{dou2019pnp} proposed two parallel domain-specific encoders and decoders where the weights are shared between domains to boost the model performance in both single domain and cross-domain scenarios.
More recently, under the privacy-preserving condition, Tian et al.~\cite{tian2021privacy} proposed a novel gradient aggregation method, which can better extract the shareable information among the data from multiple local servers. 
Besides, self-supervised learning also has been widely applied in medical image tasks.
For example, Xie et al.~\cite{xie2020instance} introduced a triplet loss for self-supervised learning in nuclei segmentation task. 
By appending a classification branch to discriminate the high-level feature, Haghighi et al.~\cite{haghighi2020learning} improved Model Genesis~\cite{zhou2019models}.
Moreover, Dong et al.~\cite{dong2020can} leveraged the dependencies between the target lesion datasets to propose a new weakly-supervised semantic transfer model in an unsupervised situation, which alternatively explores transferable domain-invariant knowledge between the source and target domains.

Nevertheless, the medical imaging methods mentioned above, neglect the real-clinical environment where the target lesion data to be diagnosed arrives continuously, and they are not capable of continual learning ability.

\section{Methodology}
In this section, we first formulate the problem of adaptation on consecutive unlabeled target lesion domains and present the overall architecture of the model framework. 
Then, we introduce the details of the implementation, including the comprehensive theoretical derivation.

\subsection{Problem Setting and Overview}
\textbf{Problem Setting:}
Consider a source lesion domain $\mathcal{D}^{s}$
and $M$ consecutive target lesion domains $\{\mathcal{D}^{t}_1, \mathcal{D}^{t}_2, \cdots, \mathcal{D}^{t}_M\}$, where the $\mathcal{D}^{s}=\{\mathbf{x}_{i}^s, \mathbf{y}_{i}^s\}^{N^s}_{i=1} \subset\mathcal{P}^s$ consists of $N^s$ pairs of samples $\mathbf{x}_{i}^s$ and their one-hot encoding labels $\mathbf{y}_{i}^s$, and $\mathcal{P}^s$ denotes the source lesion distribution. 
Define $\mathcal{D}^{t}_m=  \{\mathbf{x}_{m,j}^{t}\}_{j=1}^{N_m^{t}}\subset\mathcal{P}^{t}_m$ ($m=1, 2, \cdots, M$) as the $m$-th unlabeled target lesion domain with $N^{t}_m$ target samples $\mathbf{x}_{m,j}^{t}$, where the $\mathcal{P}^{t}_m$ is the corresponding distribution of the $m$-th target lesion domain.
Different from traditional domain adaptation methods~\cite{tzeng2017adversarial,dong2020can,long2015learning,sankaranarayanan2018generate,shu2018dirt,bousmalis2016domain} that only transfer knowledge learned from a source domain to a single target domain, our model aims to learn transferable knowledge from well-labeled source lesion domain $\mathcal{D}^s$, and achieve knowledge adaptation to 
consecutive multiple unlabeled target lesion domains $\{\mathcal{D}^{t}_1, \mathcal{D}^{t}_2, \cdots, \mathcal{D}^{t}_M\}$. 
Note that the target lesion domain $\{\mathcal{D}^{t}_m\}_{m=1}^M$ arrives continuously in real-clinical environment, and we have no any priors about total number of target domains and each target distribution. 
In summary, our method focuses on learning a diagnosis model to recognize medical lesions from novel target domains continuously via exploring transferable knowledge from the source lesion domain, while alleviating catastrophic forgetting during adaptation.

\textbf{Overview:}
The overview of our proposed Consecutive Lesion Knowledge Meta-Adaptation (CLKM) is depicted in Figure~\ref{fig:fig2}.
Specifically, CLKM mainly consists of two phases: Semantic Adaptation Phase (SAP) and Representation Adaptation Phase (RAP). 
Given the $m$-th target lesion domain $\mathcal{D}_m^t$, we construct the support set $\mathcal{S}_m=\{\mathbf{x}_{m,i}^t\}_{i=1}^{N_{sup}}\in\mathcal{D}_m^t$ and the query set $\mathcal{Q}_m=\{\mathbf{x}_{m,j}^t\}_{j=1}^{N_{que}}\in\mathcal{D}_m^t$, where $N_{sup}$ and $N_{que}$ represent the number of sampled images $\mathbf{x}_{m,i}^t$ and $\mathbf{x}_{m,j}^t$ respectively. 
In the SAP, we first forward source lesion data $\mathcal{D}^s$ and support set $\mathcal{S}_m$ into the feature-extractor $E$ to obtain the mid-level features $\mathrm{F}^s$ 
and 
$\mathrm{F}^{t,sup}_m$, which are then fed into a bottleneck $B$ to extract high-level features $\mathrm{G}^s$
and $\mathrm{G}^{t,sup}_m$
for source lesion domain and the $m$-th target lesion domain. 
For exploring transferable knowledge from the source lesion domain $\mathcal{D}^s$ to the $m$-th unsupervised target lesion domain, we propose a self-Adaptive Kernel-based Maximum Mean Discrepancy (AK-MMD) loss $\mathcal{L}_{ak}$, which can measure distribution divergence across different domains adaptively via a self-adaptive kernel $\delta_{ak}$. 
Moreover, we introduce a domain-knowledge quantizer $D_Q$ to quantify the contributions of knowledge from consecutive target lesion domains, which alleviates the catastrophic forgetting by incorporating it with anti-forgetting regularizer $\mathcal{L}_w$. 
Besides, $\mathrm{G}^s$ is passed into classifier $C$ to compute cross-entropy loss $\mathcal{L}_{ce}$.

Furthermore, in the RAP, the source lesion data $\mathcal{D}^s$ and query set $\mathcal{Q}_m$ are passed into the feature-extractor $E$ to get mid-level features $\mathrm{F}^s$ and $\mathrm{F}^{t,que}_m$, which are then fed into the bottleneck $B$ to obtain high-level features $\mathrm{G}^s$ and $\mathrm{G}^{t,que}_m$ for source lesion domain and the $m$-th target lesion domain.
Similarly in SAP, the $\mathrm{G}^s$ will be passed into classifier $C$ to calculate the cross-entropy loss $\mathcal{L}_{ce}$, and the domain discrepancy between $\mathrm{G}^s$ and $\mathrm{G}^{t,que}_{m}$ will be measured by Maximum Mean Discrepancy (MMD)~\cite{tzeng2014deep}.
To make the SAP more effective, the feature-extractor $E$ is updated to align the transferable representation features across the source and multiple target lesion domains.
Meanwhile, to retain the important knowledge learned from previous target data, the domain-quantizer $D_Q$ is also be optimized. 
Since lack of the supervised information on target lesion domains, we introduce the upper bound loss $\mathcal{L}_u$ to optimize the feature-extractor $E$ and the domain-quantizer $D_Q$ simultaneously.
Inspired by~\cite{finn2017model}, we utilize meta-learning mechanism: SAP and RAP cooperate to learn fast adaptation meta-test features that perform diagnostic knowledge over consecutive target lesion domains.

\begin{figure}[t!] 
 \center{\includegraphics[width=17.5cm]{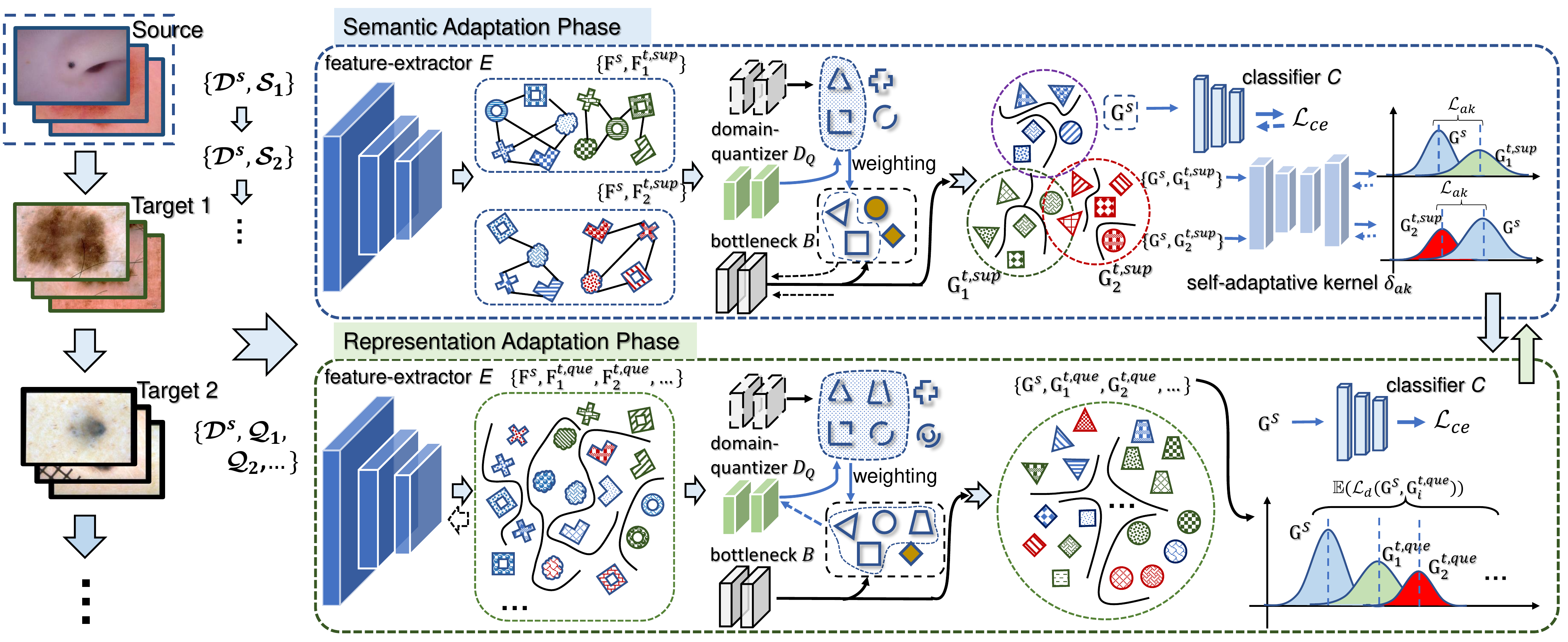}}
    \caption{Overview of Consecutive Lesion Knowledge Meta-Adaptation (CLKM), which mainly consists of the Semantic Adaptation Phase (SAP) and Representation Adaptation Phase (RAP).
    In the SAP, the semantic features are transferred preliminarily, with domain-quantizer to alleviate catastrophic forgetting, and the self-adaptive kernel to flexibly measure the gap across the source and consecutive target lesion domains;
    Further, in the RAP, the feature-extractor and domain-quantizer are optimized simultaneously for making SAP more effective.
    }
    
    \label{fig:fig2}
\end{figure}

\subsection{Semantic Adaptation Phase (SAP)}

In this subsection, we introduce how to transfer learned semantic knowledge from a source domain $\mathcal{D}^s$ to consecutive target domains, and show the details of adaptation to the $m$-th target domain $\mathcal{D}^t_m$.
To flexibly measure the gap across the source lesion domain and consecutive target lesion domains, a self-adaptive kernel $\delta_{ak}$ is designed.
Besides, we propose domain-quantizer $D_Q$ to alleviate catastrophic forgetting.
In the SAP, based on current feature-extractor $E$ and domain-quantizer $D_Q$, the classifier $C$ and bottleneck $B$ will be optimized for semantic knowledge adaptation.

\subsubsection{Classification loss $\mathcal{L}_{\rm{ce}}$ of source lesion domain}
In the process of adapting learned knowledge to consecutive target lesion domains, we can leverage the supervised information on source lesion domain to optimize the classifier $C$ via cross-entropy loss $\mathcal{L}_{ce}$.
Specifically, for the $i$-th element of source sample $\textbf{x}_{i}^{s}$ and its corresponding one-hot label $\textbf{y}_{i}^{s}$, the $\mathcal{L}_{ce}$ is formulated as follows:

\begin{equation}
     \mathcal{L}_{ce} = \mathbb{E}_{(\textbf{x}_{i}^s, \textbf{y}_{i}^s)\in\mathcal{D}^s}[-\textbf{y}_{i}^{s} \log(C(B(E(\textbf{x}_{i}^{s};\Theta_{E});\Theta_{B});\Theta_{C})],
\end{equation}
where $\Theta_{E}$, $\Theta_{B}$ and $\Theta_{C}$ are the parameters of feature-extractor $E$, bottleneck $B$ and classifier $C$.
The $E(\textbf{x}_i^s; \Theta_{E})$, $B(E(\textbf{x}_i^s; \Theta_{E});\Theta_{B})$, and
$C(B(E(\textbf{x}_i^s; \Theta_{E});\Theta_{B});\Theta_{C})$ denote the mid-level features $\mathrm{F}^s$ extracted by $E$ for source samples $\textbf{x}_{i}^{s}$, the high-level features $\mathrm{G}^s$ extracted by $B$, and the probability output of $C$, respectively.

For unlabeled consecutive target lesion domains, we cannot utilize supervised information like the source lesion domain.
Most existing methods align the source and unlabeled target domains in the feature space via reducing MMD~\cite{tzeng2014deep} based on traditional kernel methods.
Such methods \cite{tzeng2014deep,gretton2012kernel,long2017deep} work well for samples from simple distributions when using appropriate artificial-designed kernels.
However, in our settings, the unlabeled target lesion domains are collected continuously.
As a result, the semantic features (distributions) of target lesion data are dynamically changing, and a simple kernel cannot properly model such complex distributions.
To address this problem, we design a self-adaptive kernel to help measure the discrepancy over consecutive unlabeled target lesion domains flexibly.

\subsubsection{Self-Adaptive Kernel-based Maximum Mean Discrepancy (AK-MMD) loss $\mathcal{L}_{ak}$}
As mentioned above, the artificial-designed kernel $\delta$ is not flexible for consecutive target lesion domains. 
To measure the distribution discrepancy across lesion domains adaptively, we adopt a deep network approach to build a self-adaptive kernel.
Most previous works \cite{wilson2016deep,jean2018semi,wenliang2019learning} on deep kernels use a kernel directly measure the output of feature extraction network $F$:
\begin{equation}
    \delta(\textbf{x}^s,\textbf{x}^{t}_{m};\Theta_F)=\mathcal{K}_{\rho}(F(\textbf{x}^s),F(\textbf{x}^{t}_{m})),
\end{equation}
where the $\textbf{x}^s$ and $\textbf{x}^t_m$ are the samples from the source lesion domain $\mathcal{D}^s$ and $m$-th target lesion domain $\mathcal{D}^t_m$, respectively.
$\Theta_F$ represents the parameters of feature extraction network $F$.
However, measuring directly may lead to the learned kernel regarding extremely far-away inputs as similar.
Thus, according to \cite{liu2020learning1}, we add a safeguard in our kernel architecture:

\begin{equation}
    \delta_{ak}(\textbf{x}^s, \textbf{x}^t_m;\Theta_F) = [(1-\epsilon)\mathcal{K}_{\rho}(F(\textbf{x}^s),F(\textbf{x}^t_m)) + \epsilon]\mathcal{K}_{\gamma}(F(\textbf{x}^s), F(\textbf{x}^t_m)),
\end{equation}
where $\mathcal{K}_\rho(F(\textbf{x}^s), F(\textbf{x}^t_m)) = \exp(-\frac{1}{2\sigma^2_{\rho}}||F(\textbf{x}^s) - F(\textbf{x}^t_m)||^2)$ denotes a Gaussian kernel with lengthscale $\sigma_{\rho}$ on the output of $F$, and the $\mathcal{K}_{\gamma}$ denotes a Gaussian kernel with lengthscale $\sigma_{\gamma}$ on the extracted features.
We choose $0 < \epsilon < 1$.

\textbf{Learning the self-adaptive kernel.}
An appropriate metric should be able to distinguish between the distribution of the source lesion domain and the distribution of the $m$-th target lesion domain correctly (\emph{i.e.}, $\mathcal{P}^s \neq \mathcal{P}^t_m$). 
We use MMD~\cite{tzeng2014deep} as metric $d$, and use the $U$-statistic $\hat{d}^2_u$ as an unbiased estimator for $d^2$, which has nearly minimal variance among unbiased estimators \cite{gretton2012kernel}.
Then, the distance between the source lesion distribution and the $m$-th target lesion distribution can be represented as follows:

\begin{equation}
    \hat{d}^2_u(\mathcal{D}^s, \mathcal{D}^t_m, \delta_{ak})= \mathbb{E}_{i\neq j} \mathcal{M}_{ak}(u_i,u_j;\Theta_{F}),
\end{equation}
where $u_i$ denotes the pair $(\textbf{x}_i, \textbf{x}_{i,m}^t)$, and we define the function:
\begin{equation}
\begin{split}
    \mathcal{M}_{ak}(u_i,u_j;\Theta_F) = \delta_{ak}(\textbf{x}_i^s, \textbf{x}_j^s;\Theta_F)  &+ \delta_{ak}(\textbf{x}_{i,m}^{t},\textbf{x}_{j,m}^{t}; \Theta_F) \\ &- \delta_{ak}(\textbf{x}_{i}^{s}, \textbf{x}_{j,m}^{t};\Theta_F) - \delta_{ak}(\textbf{x}_{i,m}^{t},\textbf{x}_j^{s};\Theta_F).
\end{split}
\end{equation}
In hypothesis theory, when $\mathcal{P}^s \neq \mathcal{P}^t_m$, a standard central limit theorem holds \cite{serfling2009approximation}:
\begin{equation}
    \sqrt{n}(\hat{d}^2_{u} - d^2)\stackrel{\mathbf{D}}{\rightarrow}\mathcal{N}(0, \sigma^2_{\mathrm{H}_{1}}).
\end{equation}
Here, the $\mathrm{H}_{1}$ denotes the alternative  hypothesis, $\stackrel{\mathbf{D}}{\rightarrow}$ denotes convergence in distribution, $\sigma_{{\rm{H}_1}}^2 = 4(\mathbb{E}_{u_1}[(\mathbb{E}_{u_2}\mathcal{M}(u_1,u_2;\Theta_F))^2] - [\mathbb{E}_{u_1,u_2}(\mathcal{M}(u_1,u_2;\Theta_F))]^2)$, and $\mathcal{N}$ denotes the normal distribution.
We adopt $\hat{\sigma}_{\mathrm{H}_{1},\lambda}$ as a regularized estimator of $\sigma_{\mathrm{H}_1}$:
\begin{equation}
    \frac{4}{n^3}\sum_{i=1}^{n}\left(\sum_{j=1}^{n}\mathcal{M}(u_i,u_j;\Theta_F)\right)^2 - \frac{4}{n^4}\left(\sum_{i = 1}^{n} \sum_{j=1}^{n}\mathcal{M}(u_i,u_j;\Theta_F)\right)^2 + \lambda.
\end{equation}
For a reasonably large sample size $n$, the test power is dominated by:
\begin{equation}
    \hat{J}_\lambda(\mathcal{D}^s, \mathcal{D}^{t}_m; \delta_{ak}) = \frac{\hat{d}^2_{u}(\mathcal{D}^s, \mathcal{D}^{t}_m; \delta_{ak})}{\hat{\sigma}_{\mathrm{H}_{1},\lambda}(\mathcal{D}^s, \mathcal{D}^{t}_m; \delta_{ak})}.
\end{equation}
Thus, we can update the parameters $\Theta_{F}$ of the feature extraction network $F$ by maximizing $\hat{J}_\lambda$:

\begin{equation}
    \Theta_{F} \gets \Theta_{F} + \eta_{ker} \nabla_{\Theta_{F}} \hat{J}_\lambda,
\end{equation}
where $\eta_{ker}$ denotes the learning rate of training the self-adaptive kernel.
A version with more details is given in~\ref{appendix}.

\textbf{Constructing the AK-MMD loss $\mathcal{L}_{ak}$}.
We can obtain a trained self-adaptive kernel $\delta_{ak}$ through the above training strategy.
Then, the MMD based on the artificial-designed kernel is replaced by AK-MMD, to flexibly measure the gap across the source lesion domain and consecutive target lesion domains.
The following formula is the definition of $\mathcal{L}_{ak}$, when the learned knowledge is being transferred to the $m$-th target lesion domain:
\begin{equation}
    \mathcal{L}_{ak} = \hat{d}^2_{u}(\mathcal{D}^s, \mathcal{D}^t_m, \delta_{ak}).
\end{equation}
As shown in Figure~\ref{fig:fig2}, bottleneck $B$ is optimized via AK-MMD loss $\mathcal{L}_{ak}$ in SAP.

\begin{figure}[H]
    \centering
    \includegraphics[width=15cm]{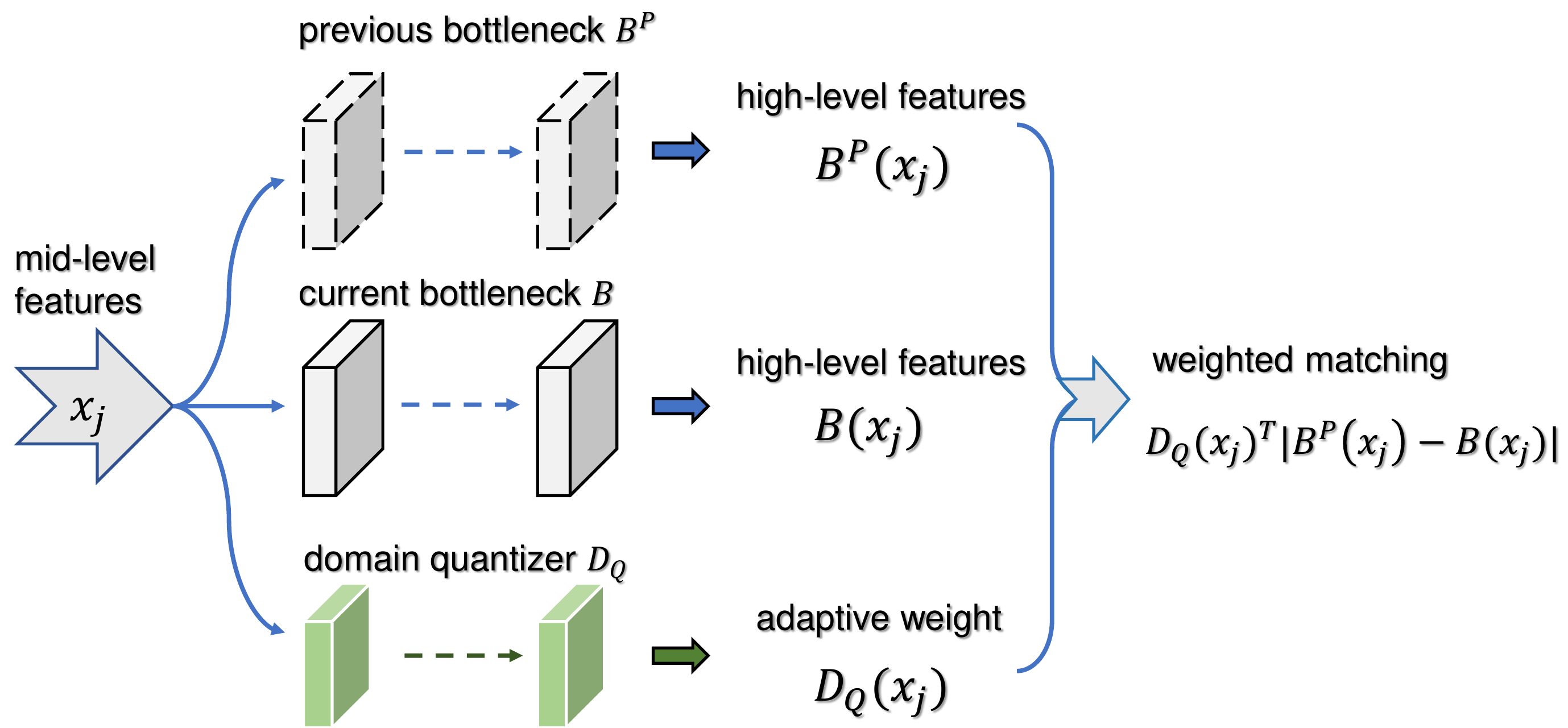}
    \caption{The weighted matching procedure.}
    \label{fig3}
\end{figure}

\subsubsection{Anti-forgetting regularizer $\mathcal{L}_w$ of matching weights}
The current learned knowledge will overwrite the previously learned knowledge as the target lesion domains are continuously arriving, which makes the adapted model degrades dramatically on previously domains (\emph{i.e.}, catastrophic forgetting~\cite{mccloskey1989catastrophic}).
To make predictions for both new target lesion data and past learned target lesion data, we should prevent catastrophic forgetting during continuous adaptation.
However, existing solutions~\cite{baweja2018towards,lee2017overcoming,zenke2017continual,shin2017continual} need to store previous samples or require prior professional knowledge, which doesn’t meet our online update settings.

For a bottleneck that is well optimized on previous target lesion domains, its intermediate features contain the utility knowledge for those domains, and thus reserving this knowledge during the model adaptation could help against forgetting.
Nevertheless, not all knowledge needs to be maintained, 
a more feasible approach is to select the important knowledge to be preserved by matching different weights.
Therefore, we introduce a domain-quantizer $D_Q$ into our model to alleviate catastrophic forgetting, and the weighed matching procedure is depicted in Figure~\ref{fig3}.

\textbf{Retain previous important knowledge}.
We use $B_{{l}}: \mathcal{R}^{d} \rightarrow \mathcal{R}^m$ to denote the $l$-th intermediate layer of the current bottleneck $B$, and $B_{l}^{p}$ represents the $l$-th intermediate layer of the previous bottleneck $B^p$.
For a mini-batch features $x \in \mathcal{R}^{n \times d}$, the weighted feature matching loss of $l$-th layer is $\sum_{j=1}^n w_j^{T}||B_{l}(x_j)-B_{l}^{p}(x_j)||$, where $w_j$ denotes the weight of feature $x_j$.
We design a domain-quantizer $D_Q$ to match appropriate weights for different layers relay on different features $w_{j} = D_{Q_l}(x_j)$, where the $D_{Q_l}$ is the $l$-th intermediate layer of domain-quantizer.
Thus, the total matching weighted feature loss is formulated as follows:
\begin{equation}
    \mathcal{L}_w = \sum_l \sum_j D_{Q_l}(x_j)^{T}||B_{l}(x_j)-B_{l}^{p}(x_j)||.
\end{equation}
As shown in Figure~\ref{fig:fig2}, in the optimize bottleneck $B$ process, we add $\mathcal{L}_{w}$ to retain previous important knowledge, thereby alleviating catastrophic forgetting.

\subsection{Representation Adaptation Phase (RAP)}
In the SAP, we fix the parameters of feature-extractor $E$, and transfer the lesion semantic knowledge to consecutive target domains preliminarily.
Since the semantic features are extracted from representation features, thus well-aligned representation features could make the semantic features adaptation more effective.
Therefore, we optimize the feature-extractor $E$ in the RAP to explore the transferrable representation features across the source and multiple target lesion domains.
Moreover, aiming to reserve more utility knowledge to alleviate forgetting, the domain-quantizer $D_Q$ is also updated in this phase.
However, we can’t calculate the loss on the target lesion data immediately due to the absence of supervised information. 
Inspired by the generalization bound of classic domain adaptation~\cite{long2017deep}, we extend it into CLDA and introduce the upper bound loss $\mathcal{L}_u$ on the consecutive target lesion domains, which is leveraged to optimize feature-extractor $E$ and domain-quantizer $D_Q$ simultaneously.

\subsubsection{Upper bound loss $\mathcal{L}_{u}$ of the target lesion domains}
Our purpose is to learn a diagnosis model $R$ that can maintain performance over the consecutive target lesion domains $\mathcal{D}^{t}_{m}(m = 1, 2, \cdots, M)$:

\begin{equation}
    \min_{\Theta_R} \mathcal{L}_{t}
    = \mathbb{E}_{m}\mathbb{E}_{(\textbf{x}_i^t, \textbf{y}_i^t) \in \mathcal{D}^t_m}[-\textbf{y}_i^t\log(R(\textbf{x}_i^t;\Theta_R))],
    \label{eq10}
\end{equation}
where model $R = C \circ B \circ E$, and $\Theta_R$ is the parameters of $R$.
Yet we cannot access the label $\textbf{y}_i^t$ of the $i$-th element sample $\textbf{x}_i^t$ from the target lesion domain, which is necessary for calculating the multiple target domains loss $\mathcal{L}_t$.
Reviewing the generalization bound analysis of discrete domain adaptation~\cite{long2017deep}, 
the upper bound of the target domain error $\mathcal{E}_{\mathcal{P}^t_m}$ can be represented as follows:

\begin{equation}
    \mathcal{E}_{\mathcal{P}^t_m}(R) \leq \mathcal{E}_{\mathcal{P}^s}(R) + d_{\mathcal{H}\Delta\mathcal{H}}(\mathcal{P}^s,\mathcal{P}^t_m) + \tau_m,
\end{equation}
where $\mathcal{E}_{\mathcal{P}^s}$ represents the error of $R$ on the source data, $d_{\mathcal{H}\Delta\mathcal{H}}(\mathcal{P}^s, \mathcal{P}^t_m)$ is the $\mathcal{H}\Delta\mathcal{H}$-divergence, and $\tau_m = \min_{\Theta_R}[\mathcal{E}_{\mathcal{P}^s}(R)+\mathcal{E}_{\mathcal{P}^t_m}(R)]$.
We further extend this analysis to our setting.

\begin{theorem}
For consecutive target lesion distributions $\mathcal{P}^t_m(m = 1, 2, \cdots, M)$, we assume the $d_{\mathcal{H}\Delta\mathcal{H}}(\mathcal{P}^t_m,\mathcal{P}^t_{m+1}) \leq \alpha|t_{m} - t_{m+1}|$, where $\alpha$ is a constant, and $t_m$ is the $\mathcal{P}^t_m$ arrival time. 
Then for any $R$, with probability at least $1 - \beta$ over the arrival time of target distributions $t_1, t_2, \cdots, t_M$,
\end{theorem}
\begin{equation}
    \mathcal{L}_{t} \leq \mathcal{E}_{\mathcal{P}^s}(R)+\frac{1}{M}\sum_{m=1}^{M}[d_{\mathcal{H}\Delta\mathcal{H}}(\mathcal{P}^s,\mathcal{P}^t_m)] + \mathbb{E}_{m}\tau_m + O(\frac{\alpha}{\beta n}).
    \label{eq13}
\end{equation}
Thus, we can obtain the upper bound of Eq.(\ref{eq10}), and update the parameters $\Theta_{R}$ to control the upper bound loss $\mathcal{L}_{u}$. 
To learn the transferable representation tailored to CLDA, we sample $\textbf{x}^t_m$ from $\mathcal{Q}_m$ and use $\max_m d(R(\textbf{x}^t_m;\Theta_R),R(\textbf{x}^t_{m+1};\Theta_R))$ as an approximation of the constant in~Eq.(\ref{eq13}):
\begin{equation}
    \begin{split}
    \mathcal{L}_{u} = \mathbb{E}_{(\textbf{x}_{i}^{s},\textbf{y}_{i}^{s})\in\mathcal{D}^s}[-\textbf{y}_{i}^{s} \log(R(\textbf{x}_{i}^{s};\Theta_{R}))]
    & + \frac{1}{M}\sum_{m=1}^M d(R(\textbf{x}^s;\Theta_{R}),R(\textbf{x}^t_m;\Theta_{R}))\\
    & +\max_{m}d(R(\textbf{x}^t_{m};\Theta_R),R(\textbf{x}^t_{m+1};\Theta_R)).
    \end{split}
\end{equation}
In the RAP, the feature-extractor $E$ and domain-quantizer $D_Q$ are optimized via the upper bound loss $\mathcal{L}_u$, meanwhile the parameters of bottleneck $B$ and classifier $C$ are fixed.

\subsection{Training and Testing}
\textbf{Meta-training phase}. The meta-training phase consists of SAP and RAP.
In the SAP, feature-extractor $E$ extracts mid-level features $\mathrm{F}^s$ and $\mathrm{F}^{t,p}_m$ from source lesion data $\mathcal{D}^s$ and support set $\mathcal{S}_m$, and then they are forwarded into bottleneck $B$ to get high-level features $\mathrm{G}^s$ and $\mathrm{G}^{t,p}_m$ for calculating AK-MMD loss $\mathcal{L}_{ak}$.
To optimize classification loss $\mathcal{L}_{ce}$, the $\mathrm{G}^s$ is passed into classifier $C$.
Additionally, to alleviate catastrophic forgetting, penalty term $\mathcal{L}_\omega$ from the domain-quantizer $D_Q$ is added.
For $m = 0, 1, \cdots, M-1$, the object of training in the SAP is formulated as:
\begin{equation}
    (\Theta_B, \Theta_C)_{m+1} \gets (\Theta_B, \Theta_C)_{m} - \eta_{sap}\nabla_{(\Theta_B,\Theta_C)}(\mathcal{L}_{ce}+\mathcal{L}_{ad}+\lambda \mathcal{L}_\omega),
    \label{eq16}
\end{equation}
where $\eta_{sap}$ is the learning rate in the SAP, and $\lambda$ denotes the hyperparameter of the anti-forgetting regularization.
$(\Theta_B, \Theta_C)_{m}$ are the parameters of bottleneck $B$ and classifier $C$ when adapting to $m$-th target lesion domain.
In the SAP, based on current feature-extractor $E$, the semantic knowledge learned from the source lesion domain is preliminarily adapted to consecutive target lesion domains. 
In the RAP, we optimize feature-extractor $E$ aiming to make the SAP more effective.
Besides, the domain-quantizer $D_Q$ is updated for learning to preserve more critical knowledge.
The $\mathcal{L}_u$ is calculated from the source data $\mathcal{D}^s$ and query set $\mathcal{Q}_m$, and then used to optimize $E$ and $D_Q$:

\begin{equation}
    (\Theta_E, \Theta_{D_Q})\gets (\Theta_E, \Theta_{D_Q}) - \eta_{rap} \nabla_{(\Theta_E, \Theta_{D_Q})}\mathcal{L}_{u},
    \label{eq17}
\end{equation}
where $\eta_{rap}$ is the learning rate in the RAP.
The overall optimization procedure in the meta-training phase is summarized in Algorithm~\ref{alg1}.

\textbf{Meta-testing phase.}
In the meta-testing stage, we will fix the parameters of feature-extractor $E$ and domain-quantizer $D_Q$, and only fine-tune classifier $C$ and bottleneck $B$ a few epochs via small batch data from new target lesion domains.

\begin{algorithm}
    \caption{Optimization Framework of CLKM Model}
    \begin{algorithmic}[1] 
        \Require source lesion distribution $\mathcal{P}^{s}$ and consecutive target lesion distribution $\mathcal{P}^{t}_m$.
        \Ensure learned feature-extractor $E$ and domain-quantizer $D_Q$.
        \For{$t=0$ \textbf{to}  MaxIter}
        \State{Initialize the bottleneck $B$ and  classifier $C$;} 
        \State{Sample source lesion data $\mathcal{D}^{s} = \{\mathbf{x}_{i}^{s},\mathbf{y}_{i}^{s}\}$ from $\mathcal{P}^{s}$};
        \State{Sample support set $\mathcal{S}_{m}$ and query set
        $\mathcal{Q}_{m}$ from $\mathcal{P}^{t}_m$ $(m = 1, 2, \cdots, M)$};
       
        \For{$m=0$ \textbf{to} $M-1$}
        \State{In the SAP, update the $\Theta_C$ and $\Theta_B$ via Eq.(\ref{eq16}});
        \EndFor
        \State{In the RAP, update $\Theta_E$ and $\Theta_{D_Q}$ via Eq.(\ref{eq17});}
        \EndFor
    \end{algorithmic}
    \label{alg1}
\end{algorithm}

\section{Experiments}

In this section, we evaluate the performance of our method based on a skin lesion classification task. 
The details of the architectures and experimental settings are as follows.

\subsection{Datasets}

We adopt six publicly accessible skin lesion datasets:
HAM10000 (HAM)~\cite{tschandl2018ham10000}, MSK~\cite{codella2018skin}, UDA~\cite{codella2018skin}, SONIC (SON)~\cite{codella2018skin}, Derm7pt (D7P)~\cite{kawahara2018seven}, and PH2~\cite{mendoncca2013ph}.
The above datasets contain the following seven
diagnostic categories: melanocytic nevus (NV), melanoma (MEL), basal cell carcinoma (BCC), dermatofibroma (DF), benign keratosis (BKL), vascular lesion (VASC), and actinic keratosis (AKIEC).
We randomly divided each dataset into 50$\%$ training set, 20$\%$ validation set, and 30$\%$ testing set. 
HAM is selected as the only source lesion domain, and 
we remove the lesion data beyond the seven diagnostic categories.
\begin{figure}[H]
    \centering
    \includegraphics[width=7in]{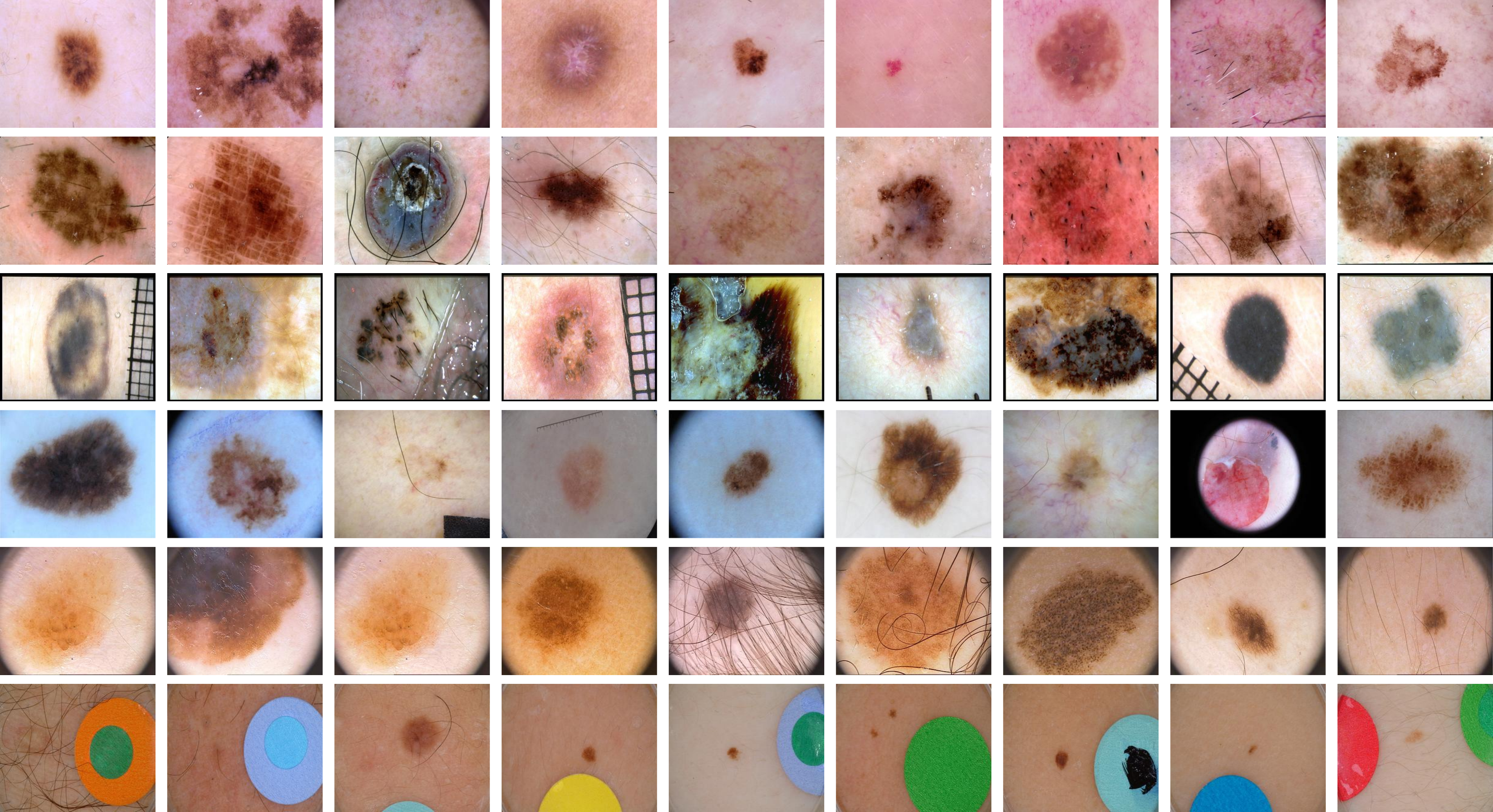}
    \caption{\textbf{
    The images in six accessible medical lesion datasets.
    } 
    The images from the first to the sixth rows are sampled from the HAM, MSK, D7P, UDA, PH2, and SON datasets respectively.}
    \label{fig:fig4}
\end{figure}

\textbf{Diagnostic categories}.
\textbf{(1) NV:} Melanocytic nevus is a benign tumor of melanocytes with many variations. 
In appearance, they are usually symmetrical in the distribution of color and structure, unlike melanomas.
\textbf{(2) MEL:} Melanoma is a malignant tumor that comes from melanocytes and can take on different mutations.
Melanomas can be invasive or noninvasive.
Melanoma in its early stages can be cured with simple surgical removal.
\textbf{(3) BCC:} Basal cell carcinoma is a common variant of epithelial skin cancer that rarely metastases but still requires early treatment. 
It exists in different morphologies, such as flat, nodular, pigmentation, and cystic.
\textbf{(4) DF:} Dermatofibroma is a benign skin lesion that is regarded as benign hyperplasia or minimal trauma to an inflammatory response.
The most common dermatoscopic presentation is reticular lines at the periphery with a central white patch denoting fibrosis.
\textbf{(5) BKL:} Benign keratosis is a common type including seborrheic keratoses and solar lentigo. 
Because of their biological similarities, the three are placed together, although they may look different under a dermoscopy.
\textbf{(6)VASC:} Vascular skin lesions include cherry angiomas, angiokeratomas, pyogenic granulomas, and hemorrhage.
\textbf{(7) AKIEC:} Actinic Keratoses (Solar Keratoses) and Intraepithelial Carcinoma (Bowen’s disease) are common non-invasive squamous cell carcinomas that can be treated locally without surgery.
However, these lesions may progress to invasive squamous cell carcinoma.
Actinic Keratoses are more common on the face, while Intrapithelia Carcinoma is more common to appear on other parts of the body.

\begin{figure}[!t]
    \centering
    \includegraphics[width=7in]{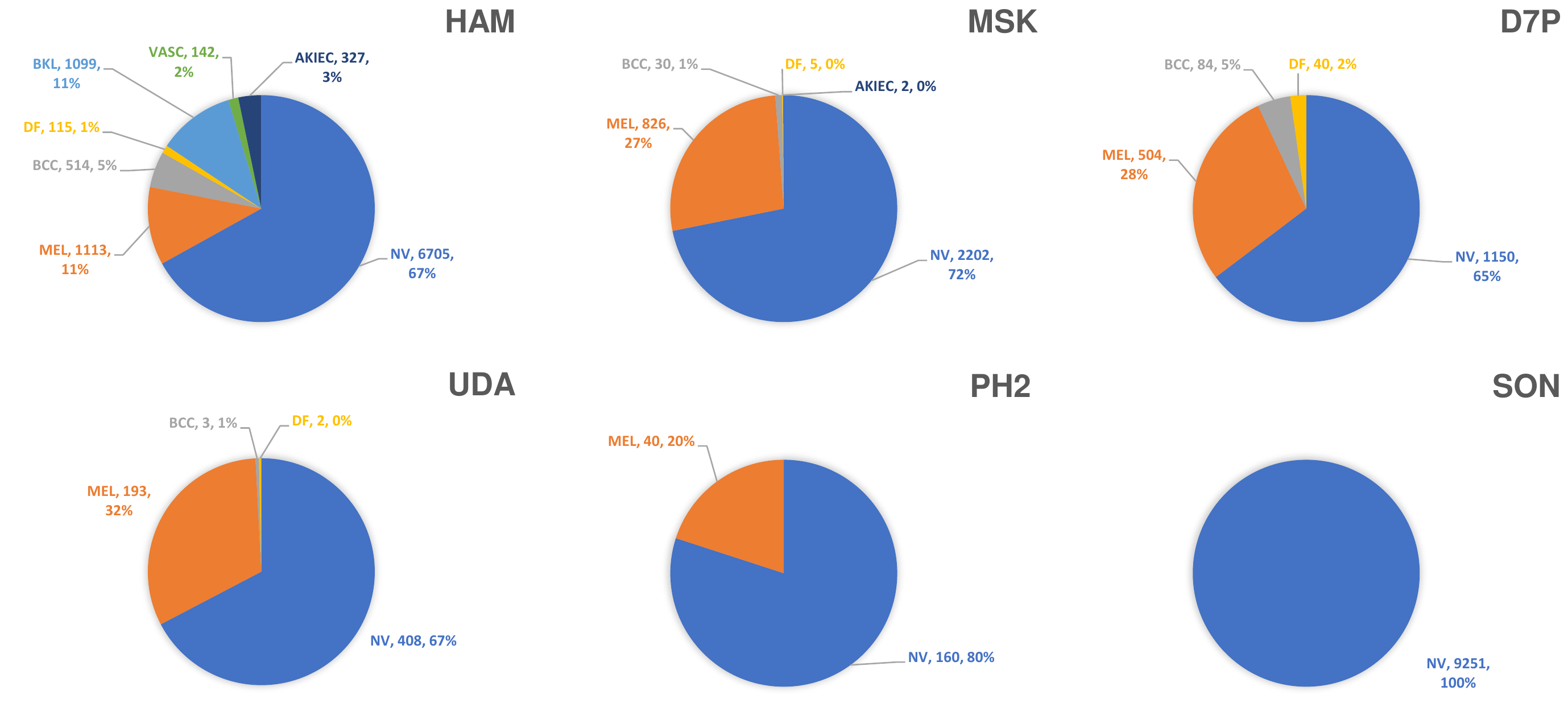}
    \caption{
    \textbf{
    The data distribution in six public skin lesion domain.}
    In each lesion domain, the proportion of different categories is extremely far away.
    Besides, the number of lesion contained in each domain is also different.
    HAM contains all seven categories of skin lesion, but only NV in SON.}
    \label{fig:fig5}
\end{figure}

\subsection{Implementation}

The original skin lesion images are shown in Figure~\ref{fig:fig4}. 
Except for the domain shift across different domains that should be resolved, we still need to address the unbalanced sample problem, which requires our model to capture the light information as much as possible.

\begin{table}[H]
    \setlength{\abovecaptionskip}{0cm}
    \setlength{\belowcaptionskip}{0.3cm}
    \vspace{-0.5cm}
    
    \captionsetup{font=normal}
    \caption{Classification Accuracy(\%) on the skin lesion task. }

    \footnotesize
    \centering

    \setlength{\tabcolsep}{1.5mm}
    \begin{tabular}
    { |l |c c c c c c >{\columncolor{lightgray!30}}c |>{\columncolor{lightgray!30}}c|}
    \hline
    
    \hline
        Methods & UDA, PH2
          & UDA, D7P & MSK, D7P & MSK, PH2 & SON, PH2 &SON, D7P & Avg. & $\Delta$ \\
    \hline
    SourceOnly & 60.4 & 66.8 & 63.6
            & 65.0 & 67.4 & 57.1 & 63.4 & $\Uparrow$ 10.2\\
    JAN~\cite{long2017deep} & 59.8 & 69.1 &
            68.5 & 69.2 & 70.1 & 71.5 & 68.0 & $\Uparrow$ 5.6\\
    DAN~\cite{ganin2016domain} & 69.9 & 62.7 & 65.8
        & 64.9 & 72.5 & 66.9 & 67.1 & $\Uparrow$ 6.5\\
    CDA(mer) \cite{long2018conditional} &70.5 &65.3 &65.4 &65.4 & 67.2& 61.5& 65.9 & $\Uparrow$ 7.7\\
    DAN(ewc) \cite{kirkpatrick2017overcoming} & 71.2 & 66.4 & 66.6 & 66.7 & 66.5 & 59.8 &  66.2 &$\Uparrow$ 7.4\\
    CaCo~\cite{huang2022category}&70.3 &67.5 &68.2 & 67.5&67.8 &64.1 &67.6 & $\Uparrow$ 6.0 \\
    SSRT~\cite{sun2022safe}&70.4 &71.1 &69.6 & 71.3& 69.1& 68.5& 70.0 & $\Uparrow$ 3.6\\
    RIPU~\cite{xie2022towards}&71.8 &72.5 &71.2 &70.5 &71.3 &70.7 & 71.3 & $\Uparrow$ 2.3\\
    \hline
    \hline
    CLKM(fe) & 71.3 & 65.6 & 71.5 & 70.7 & 64.7 & 65.3 & 68.2 &$\Uparrow$ 5.4\\
    CLKM(dq) & 71.3 & 66.5 & 71.1 & 71.6 & 71.6 & 71.6 & 70.6 &$\Uparrow$ 3.0\\
    CLKM(f\&d) & 71.9 & 71.9 & 72.2 & 72.4 & 74.3 & 70.1 & 72.1 &$\Uparrow$ 1.5\\
    CLKM(full) & \textbf{72.3} & \textbf{73.4} & 
    \textbf{74.0} & \textbf{73.4} & \textbf{74.4} & \textbf{74.3} & \textbf{73.6} & -\\
    \hline
    
    \hline
    \end{tabular}
    \label{tab1}
\end{table}
\textbf{Setting.} In the meta-training phase, 
HAM is selected as the only source lesion domain.
For the remaining data, we first select three domains as meta-training datasets and the other two as meta-testing datasets.
It is noted that, in the meta-testing phase, only a small batch sample from the meta-testing datasets is used to optimize the bottleneck and classifier.
Besides, to perform a comprehensive analysis, we also consider a scenario in which two domains are meta-training datasets and the other three domains are meta-testing datasets.
To simulate realistic training approaches, we augmented the data by resizing all images to 
$650\times650$, cropping the center $300\times300$, randomly resizing and cropping $224\times224$, and randomly horizontal flipping. 
Finally, we use the all test sets from the six skin lesion domains to evaluate each method.
Our method is implemented based on PyTorch, and we adopt ResNet18 as the backbone for CLKM as well as the other baselines.
For the bottleneck, we use a three-layer fully connected network and adopt ReLU as its activations.
We use stochastic gradient descent with 0.9 momentum and $5\times10^{-4}$ weight decay.
The initial learning rate of the SAP and the RAP will be adjusted according to the specific meta-training datasets, generally around at $1\times 10^{-6}$. 
For the self-adaptive kernel, we use a five-layer fully connected network and adopt softplus as the activation function.

\textbf{Baselines.}
\textbf{(1) SourceOnly:} 
The diagnosis model is only trained on the source lesion domain without adaptation.
\textbf{(2) DAN~\cite{ganin2016domain} and JAN~\cite{long2017deep}:} 
For these two classic discrete domain adaptations, we apply them to CLDA as follows.
In the meta-training phase, we train the model with labeled source lesion data and adapt it to accessible target domains. 
In meta-testing, the two methods are adapted to the continuous new target lesion data and then test their performance.
In meta-testing, we adapt both methods to the continuous new target data and test their performance.
\textbf{(3) DAN+EWC~\cite{kirkpatrick2017overcoming}:} EWC is a common method to overcome catastrophic forgetting.
We separate the source lesion domain into two task so that get EWC regularization, and then each task both adapt targets with DAN in the meta-training phase.
\textbf{(4) CaCo~\cite{huang2022category}:} CaCo is a novel technique that tackles UDA from the perspective of instance contrastive learning. 
We apply it in CLDA via designing a category-aware and domain-mixed dictionary, then construct the category contrast loss to pull closer to the same lesion category and push apart different lesion categories.
\textbf{(5) SSRT~\cite{sun2022safe}:} SSRT is a powerful discrete domain adaptation baseline, we transfer the knowledge learned from the source lesion domain to the accessible target lesion domains in meta-training. 
In meta-testing, we adapt it to the new continuous target lesion domains and then test its performance. 
\textbf{(6) RIPU~\cite{xie2022towards}:} As an active learning method for domain adaptation in the semantic segmentation task, we modify the pixel-level cross-entropy loss function in RIPU and then apply it to the CLDA. 

\textbf{Results.}
Table~\ref{tab1} shows the overall performance of each method.
The HAM is selected as the only source lesion domain, and randomly select two lesion domains as the meta-testing datasets to get an overall evaluation of each method.
The distribution of these datasets varies greatly, as shown in Figure~\ref{fig:fig5}, so we further provided fine-grained evaluation results in Table~\ref{tab2}-Table~\ref{tab7}.
As we can see, our method achieve the optimal performance.
Generally, the domain-quantizer help to retain previously learned knowledge, and thus the accuracy of CLKM(f\&d) are 3.9\% and 1.5\% higher than CLKM(fe) and CLKM(dq) respectively.
The self-adaptive learned kernel is further help to measure the discrepancy across different domains, and the accuracy of CLKM(full) is 2.3\% higher than state-of-the-art RIPU.

\begin{table}[!t]
    \setlength{\abovecaptionskip}{0cm}
    \setlength{\belowcaptionskip}{-0.5cm}
    \captionsetup{font=normal}
    \caption{Classification Accuracy(\%) when meta-testing datasets are: \textbf{SON},\textbf{PH2}}
    \small
    \centering

    \setlength{\tabcolsep}{1.5mm}
    \begin{tabular}
    { |l| c c c c c c >{\columncolor{lightgray!30}}c |>{\columncolor{lightgray!30}}c|}
    \hline
    
    \hline
        Methods & HAM
          & D7P & SON & MSK & UDA & PH2 & Avg. & $\Delta$ \\
    \hline
    SourceOnly & \textbf{72.3} & 51.7 & 79.5
            & 60.1 & 59.3 & 81.7 & 67.4 & $\Uparrow  7.0$ \\
    JAN \cite{long2017deep} & 64.1 & 55.7 & 93.2
             & 64.1 & 59.9 & 83.3 & 70.1 & $\Uparrow 4.3$ \\
    DAN~\cite{ganin2016domain} & 68.5 & 57.5 & 96.0
            & 67.3 & 62.1 & 83.3 & 72.5 & $\Uparrow 1.9$ \\
    CDA(mer) \cite{long2018conditional} &67.2 & 56.5& 85.4& 59.4&60.1 &79.1 & 68.0 & $\Uparrow$ 6.4\\
    DAN(ewc) \cite{kirkpatrick2017overcoming} & 68.9 & 53.2 & 82.3 & 58.8 & 57.7 & 78.3 & 66.5 &$\Uparrow 7.9$ \\
    CaCo~\cite{huang2022category}& 69.2&54.7 &83.6 &60.1 &59.2 &79.4 & 67.7 & $\Uparrow$ 6.7\\
    SSRT~\cite{sun2022safe}& 70.5&55.2 &80.7 &62.3 &58.4 &80.3 & 67.9 & $\Uparrow$ 6.5\\
    RIPU~\cite{xie2022towards}& 68.3& 58.2& 84.9&63.7 & 60.1& 82.6& 69.6 & $\Uparrow$ 4.8\\
    \hline
    \hline
    CLKM(fe) & 67.0 & 54.8 & 87.6 & 55.3 & 50.0 & 73.3 & 64.7 &$\Uparrow 9.7$ \\
    CLKM(dq) & 68.1 & 59.7 & 96.5 & 63.2 & 58.8 & 83.3 & 71.6 &$\Uparrow 2.8$ \\
    CLKM(f\&d) & 68.3 & 58.1 & 94.4 & 66.9 & 59.3 & 83.3 & 71.5 &$\Uparrow 2.9$ \\
    CLKM(full) & 67.1& \textbf{61.3} & \textbf{98.9} & \textbf{69.4} & \textbf{63.2} & \textbf{86.7} & \textbf{74.4} & -\\
    \hline
    
    \hline
    \end{tabular}
    \label{tab2}
\end{table}

\begin{table}[H]
    \setlength{\abovecaptionskip}{0cm}
    \setlength{\belowcaptionskip}{-0.5cm}
    \captionsetup{font=normal}
    \caption{Classification Accuracy(\%) when meta-testing datasets are: \textbf{D7P},\textbf{SON}}
    \small
    \centering

    \setlength{\tabcolsep}{1.5mm}
    \begin{tabular}
    { |l| c c c c c c >{\columncolor{lightgray!30}}c |>{\columncolor{lightgray!30}}c|}
    \hline
    
    \hline
        Methods & HAM
          & D7P & SON & MSK & UDA & PH2 & Avg. & $\Delta$ \\
    \hline
    SourceOnly & 65.9 & 43.4 & 63.2
            & 47.6 & 47.3 & 75.0 & 57.1 & $\Uparrow  17.0$ \\
    JAN \cite{long2017deep} & 66.4 & 60.1 & 93.2
             & 64.5 & 61.5 & 83.3 & 71.5 & $\Uparrow 2.6$ \\
    DAN \cite{ganin2016domain} & \textbf{68.9} & 54.1 & 70.2 & 63.5 & \textbf{64.8} & 80.0 & 66.9 & $\Uparrow 7.2$ \\
    CDA(mer) \cite{long2018conditional} &67.2 &50.3 &69.2 &54.6 &55.7 & 75.8& 62.1 & $\Uparrow$ 12.0\\
    DAN(ewc) \cite{kirkpatrick2017overcoming} & 66.0 & 48.1 & 68.7
            & 51.7 & 51.1 & 73.3 & 59.8 &$\Uparrow 14.3$ \\
    CaCo~\cite{huang2022category}&66.9 &49.7 &71.3 &53.4 &52.8 & 74.2& 61.4 & $\Uparrow$ 12.7 \\
    SSRT~\cite{sun2022safe}&66.3 & 50.8& 73.2&55.9 &56.2 &76.1 & 63.1 & $\Uparrow$ 11.0\\
    RIPU~\cite{xie2022towards}& 67.1&52.6 &75.4 & 58.2& 59.1& 76.9 & 64.9 & $\Uparrow$ 9.2\\
    \hline
    \hline
    CLKM(fe) & 68.3 & 60.1 & 97.4 & 64.8 & 61.0 & 80.0 & 71.9 &$\Uparrow 2.2$ \\
    CLKM(dq) & 68.8 & \textbf{61.5} & 92.3 & 64.2 & 59.3 & 83.3 & 71.6 &$\Uparrow 2.5$ \\
    CLKM(f\&d) & 68.2 & 61.3 & 99.2 & 68.2 & 61.0 & 81.7 & 73.2 &$\Uparrow 0.9$ \\
    CLKM(full) & 67.5 & 61.3 & \textbf{99.8} & \textbf{68.6} & 62.6 & \textbf{85.0} & \textbf{74.1} & -\\
    \hline
    
    \hline
    \end{tabular}
    \label{tab3}
\end{table}

\begin{table}[H]
    \setlength{\abovecaptionskip}{0cm}
    \setlength{\belowcaptionskip}{-0.5cm}
    \captionsetup{font=normal}
    \caption{Classification Accuracy(\%) when meta-testing datasets are: \textbf{D7P},\textbf{MSK}}
    \small
    \centering

    \setlength{\tabcolsep}{1.5mm}
    \begin{tabular}
    { |l| c c c c c c >{\columncolor{lightgray!30}}c |>{\columncolor{lightgray!30}}c|}
    \hline
    
    \hline
        Methods & HAM
          & D7P & SON & MSK & UDA & PH2 & Avg. & $\Delta$ \\
    \hline
    SourceOnly & 70.3 & 47.0 & 76.8
            & 54.5 & 56.0 & 76.7 & 63.6 & $\Uparrow 10.4 $ \\
    JAN \cite{long2017deep} & 68.5 & 54.6 & 
            88.7 & 60.1 & 57.1 & 81.7 & 68.5 & $\Uparrow 5.5 $ \\
    DAN~\cite{ganin2016domain} & \textbf{74.7} & 45.7 & 87.0
            & 55.0 & 52.2 & 80.0 & 65.8 & $\Uparrow 8.2$ \\
    CDA(mer) \cite{long2018conditional} &72.3 &46.5 &87.6 &57.2 &53.8 & 82.7& 66.7& $\Uparrow$ 7.3\\
    DAN(ewc) \cite{kirkpatrick2017overcoming} & 70.6 & 47.0 & 86.3 & 55.8 & 54.9 & 85.0 & 66.6 &$\Uparrow 7.4$ \\
    CaCo~\cite{huang2022category}&71.6 & 49.8& 88.3& 57.2& 57.4& 83.8& 68.0 & $\Uparrow$ 6.0\\
    SSRT~\cite{sun2022safe}&69.3 & 51.6&90.3 & 59.4& 58.1& 83.1&68.6 & $\Uparrow$ 5.4 \\
    RIPU~\cite{xie2022towards}&70.4 & 53.5& 92.6&62.5 & 60.4& 82.7& 70.4& $\Uparrow$ 3.6 \\

    \hline
    \hline
    CLKM(fe) & 67.4 & 60.3 & 97.5 & 65.1 & 58.8 & 80.0 & 71.5 &$\Uparrow 2.5$ \\
    CLKM(dq) & 68.4 & 60.1 & 96.0 & 61.9 & 58.8 & 81.7 & 71.1 &$\Uparrow 2.9$ \\
    CLKM(f\&d) & 67.7 & 59.2 & 98.7 & 64.9 & 59.3 & 83.3 & 72.2 &$\Uparrow 1.8$ \\
    CLKM(full) & 67.8 & \textbf{60.4} & \textbf{99.6} & \textbf{67.3} & \textbf{62.1} & \textbf{86.7} & \textbf{74.0} & -\\
    \hline
    
    \hline
    \end{tabular}
    \label{tab4}
\end{table}

\begin{table}[H]
    \setlength{\abovecaptionskip}{0cm}
    \setlength{\belowcaptionskip}{-0.5cm}
    \captionsetup{font=normal}
    \caption{Classification Accuracy(\%) when meta-testing datasets are: \textbf{D7P},\textbf{UDA}}
    \small
    \centering

    \setlength{\tabcolsep}{1.5mm}
    \begin{tabular}
    { |l| c c c c c c >{\columncolor{lightgray!30}}c |>{\columncolor{lightgray!30}}c|}
    \hline
    
    \hline
        Methods & HAM
          & D7P & SON & MSK & UDA & PH2 & Avg. & $\Delta$ \\
    \hline
    SourceOnly & 70.0 & 53.2 & 76.6 
            & 61.7 & 57.7 & 81.7 & 66.8 & $\Uparrow  6.6$ \\
    JAN~\cite{long2017deep} & 69.9 & 55.5 & \textbf{97.7}
             & 58.3 & 54.9 & 78.3 & 69.1 & $\Uparrow 4.3$ \\
    DAN~\cite{ganin2016domain} & 70.5 & 41.7 & 93.1
        & 47.3 & 53.3 & 70.0 & 62.7 & $\Uparrow 10.7$ \\
    CDA(mer) \cite{long2018conditional} &70.8 & 45.3& 88.2&49.7 &52.8 &73.5 & 63.4 & $\Uparrow$ 10.0 \\
    DAN(ewc) \cite{kirkpatrick2017overcoming} & 71.2 & 49.2 & 86.1
            & 55.0 & 54.9 & 81.7 & 66.4 &$\Uparrow 7.0$ \\
    CaCo~\cite{huang2022category}& 71.6& 51.4&88.7 &57.2 &56.5 &78.6 & 67.3 & $\Uparrow$ 6.1 \\
    SSRT~\cite{sun2022safe}&72.0 &52.3 &90.1 &59.0& 58.3& 79.2&68.5 & $\Uparrow$ 4.9\\
    RIPU~\cite{xie2022towards}&70.4 &53.9 &89.2 &61.7 &60.4 &82.5 & 69.7& $\Uparrow$ 3.7\\
    \hline
    \hline
    CLKM(fe) & \textbf{72.6} & 54.4 & 77.7 & 60.1 & 56.0 & 78.3 & 66.5 &$\Uparrow 6.9$ \\
    CLKM(dq) & 66.5 & 56.1 & 88.8 & 52.6 & 49.5 & 80.0 & 65.6 &$\Uparrow 7.8$ \\
    CLKM(f\&d) & 71.0 & 56.6 & 82.1 & 64.7 & 59.3 & 80.0 & 69.0 &$\Uparrow$ 4.4\\
    CLKM(full) & 67.9 & \textbf{60.6} & 93.7 & \textbf{67.8} & \textbf{62.1} & \textbf{88.3} & \textbf{73.4} & -\\
    \hline
    
    \hline
    \end{tabular}
    \label{tab5}
\end{table}

\begin{table}[H]
    \setlength{\abovecaptionskip}{0cm}
    \setlength{\belowcaptionskip}{-0.5cm}
    \captionsetup{font=normal}
    \caption{Classification Accuracy(\%) when meta-testing datasets are: \textbf{MSK},\textbf{PH2}}
    \small
    \centering

    \setlength{\tabcolsep}{1.5mm}
    \begin{tabular}
    { |l| c c c c c c >{\columncolor{lightgray!30}}c |>{\columncolor{lightgray!30}}c|}
    \hline
    
    \hline
        Methods & HAM
          & D7P & SON & MSK & UDA & PH2 & Avg. & $\Delta$ \\
    \hline
    SourceOnly & 68.7 & 50.5 & 77.9
            & 56.6 & 56.0 & 80.0 & 64.9 & $\Uparrow 8.5$ \\
    JAN~\cite{long2017deep} & 68.5 & 55.9 & 89.3
             & 63.9 & 57.7 & 80.0 & 69.2 & $\Uparrow 4.2$ \\
    DAN~\cite{ganin2016domain} & \textbf{72.9} & 45.0 & 91.3
            & 47.9 & 55.5 & 76.7 & 64.9 & $\Uparrow 8.5$ \\
    CDA(mer) \cite{long2018conditional} &71.3 &46.8 & 95.2& 50.4&55.8 &77.4 & 66.2 & $\Uparrow$ 7.2\\
    DAN(ewc) \cite{kirkpatrick2017overcoming} & 69.4 & 47.4 & \textbf{99.9} & 55.2 & 54.9 & 73.3 & 66.7 &$\Uparrow 6.7$ \\
    CaCo~\cite{huang2022category}& 70.3&50.1 &96.5 &57.9 &56.8 &76.4 & 68.0 & $\Uparrow$ 5.4\\
    SSRT~\cite{sun2022safe}& 68.5&52.6 &97.1 &58.5 &57.3 &77.9 & 68.7& $\Uparrow$ 4.7\\
    RIPU~\cite{xie2022towards}& 69.7&54.2 &98.2 &61.4 & 56.2 & 81.2& 70.2& $\Uparrow$ 3.2\\

    \hline
    \hline
    CLKM(fe) & 68.6 & 59.4 & 96.3 & \textbf{69.1} & 56.6 & 80.0 & 71.7 &$\Uparrow 1.7$ \\
    CLKM(dq) & 63.6 & \textbf{61.5} & 86.9 & 59.5 & \textbf{61.0} & \textbf{91.7} & 70.7 &$\Uparrow 2.7$ \\
    CLKM(f\&d) & 68.1 & 60.6 & 99.0 & 65.8 & 59.3 & 81.7 & 72.4 &$\Uparrow 1.0$ \\
    CLKM(full) & 67.7 & 61.3 & 99.2& 66.2 & 59.3 & 86.7 & \textbf{73.4} & -\\
    \hline
    
    \hline
    \end{tabular}
    \label{tab6}
\end{table}

\begin{table}[H]
    \setlength{\abovecaptionskip}{0cm}
    \setlength{\belowcaptionskip}{-0.5cm}
    \captionsetup{font=normal}
    \caption{Classification Accuracy(\%) when meta-testing datasets are: \textbf{UDA},\textbf{PH2}}
    \small
    \centering

    \setlength{\tabcolsep}{1.5mm}
    \begin{tabular}
    { |l| c c c c c c >{\columncolor{lightgray!30}}c |>{\columncolor{lightgray!30}}c|}
    \hline
    
    \hline
        Methods & HAM
          & D7P & SON & MSK & UDA & PH2 & Avg. & $\Delta$ \\
    \hline
    SourceOnly & 68.4 & 47.2 & 61.0
            & 52.0 & 53.9 & 80.0 & 60.4 & $\Uparrow 11.9$ \\
    JAN \cite{long2017deep} & 65.6 & 46.6 &
            81.5 & 48.8 & 49.5 & 66.7 & 59.8 & $\Uparrow 12.5$ \\
    DAN~\cite{ganin2016domain} & \textbf{73.8} & 52.6 & 93.3& 61.2 & 56.6 & 81.7 & 69.9 & $\Uparrow 2.4$ \\
    CDA(mer) \cite{long2018conditional} &72.2 &54.1 &95.8 &59.5 &58.3 &80.4 & 70.1& $\Uparrow$ 2.2\\
    DAN(ewc) \cite{kirkpatrick2017overcoming} & 71.1 & 54.8 & \textbf{99.0} & 58.2 & 61.0 & 83.3 &  71.2 &$\Uparrow 1.1$ \\
    CaCo~\cite{huang2022category}&70.5 & 57.6& 96.2& 63.4& 59.7 & 82.6 & 71.7 & $\Uparrow$ 0.6\\
    SSRT~\cite{sun2022safe}&72.8 &58.1 &90.4 & 62.9& 58.3& 81.7& 70.7& $\Uparrow$ 1.6\\
    RIPU~\cite{xie2022towards}& 69.4& 56.8& 89.2& 63.6&59.1 & 82.3& 70.1 & $\Uparrow$ 2.2\\
    \hline
    \hline
    CLKM(fe) & 70.2 & 55.9 & 94.0 & 67.0 & 59.3 & 81.7 & 71.3 &$\Uparrow 1.0$ \\
    CLKM(dq) & 67.5 & 60.1 & 94.5 & 64.0 & 59.9 & 81.7 & 71.3 &$\Uparrow 1.0$ \\
    CLKM(f\&d) & 68.7 & 58.6 & 84.0 & \textbf{68.3} & 61.5 & \textbf{90.0} & 71.9 &$\Uparrow 0.4$ \\
    CLKM(full) & 67.3 & \textbf{60.3} & 
    92.5 & 67.4 & \textbf{61.5} & 85.0 & \textbf{72.3} & -\\
    \hline
    
    \hline
    \end{tabular}
    \label{tab7}
\end{table}

\subsection{Analysis}
\textbf{Data-distribution.} 
In Figure~\ref{fig:fig5}, we have counted the proportions of the seven diagnostic categories in the six skin lesion domains.
As we can see, the number of different lesions is extremely unbalanced, which makes the model to classify medical images correctly more difficult.
Even so, CLKM still possesses effective results.
On the whole, the CLKM(f\&d) performed better than other classical methods, and the addition of a self-adaptive kernel further improved the performance of CLKM.
Further, we analyze the experimental results from the perspective of data distribution.

\textbf{(a) Meta-testing datasets: SON, PH2.}
When we select SONIC and PH2 as the meta-test sets, the NV is the majority in the test set, and there is only a small amount of MEL.
Comparing the other transfer methods, including SourceOnly, 
the CLKM with only feature-extractor (\emph{i.e.}, CLKM(fe)) isn't effective.
This is because the gap across training domains and testing domains is not that large, and even fine-tuning the model's classifier can achieve a not bad performance.
The result of SourceOnly illustrates this conclusion.
When we add a domain-quantizer to the model (\emph{i.e.}, CLKM(f\&d)), the performance of CLKM immediately improved.

\textbf{(b) Meta-testing datasets: D7P, SON.}
When we select D7P and SONIC as the meta-testing datasets, the performance of the regularized constrained DAN (\emph{i.e.}, DAN(ewc)) is inferior to that of normal DAN.
This degradation is also present when we select SON and PH2 as the meta-testing datasets.
In contrast, adding a domain-quantizer to CLKM always improves the model performance.

\begin{figure}[H]
\centering
\subfigcapskip = -13pt
\subfigure[meta-testing: UDA, PH2]{
\begin{minipage}[t]{0.32\linewidth}
\centering
\includegraphics[width=2.35in]{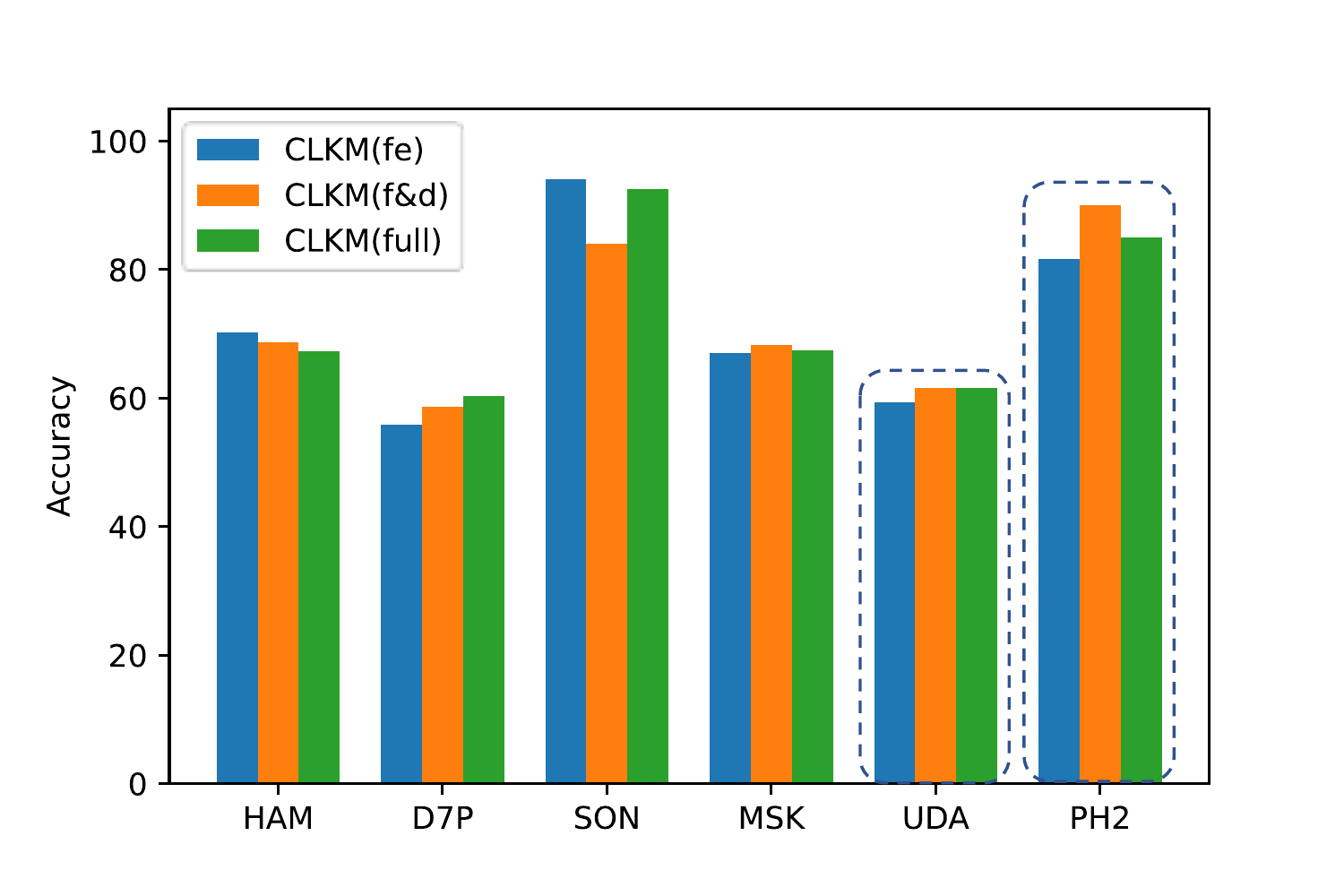}
\end{minipage}
}%
\subfigure[meta-testing: MSK, PH2]{
\begin{minipage}[t]{0.32\linewidth}
\centering
\includegraphics[width=2.35in]{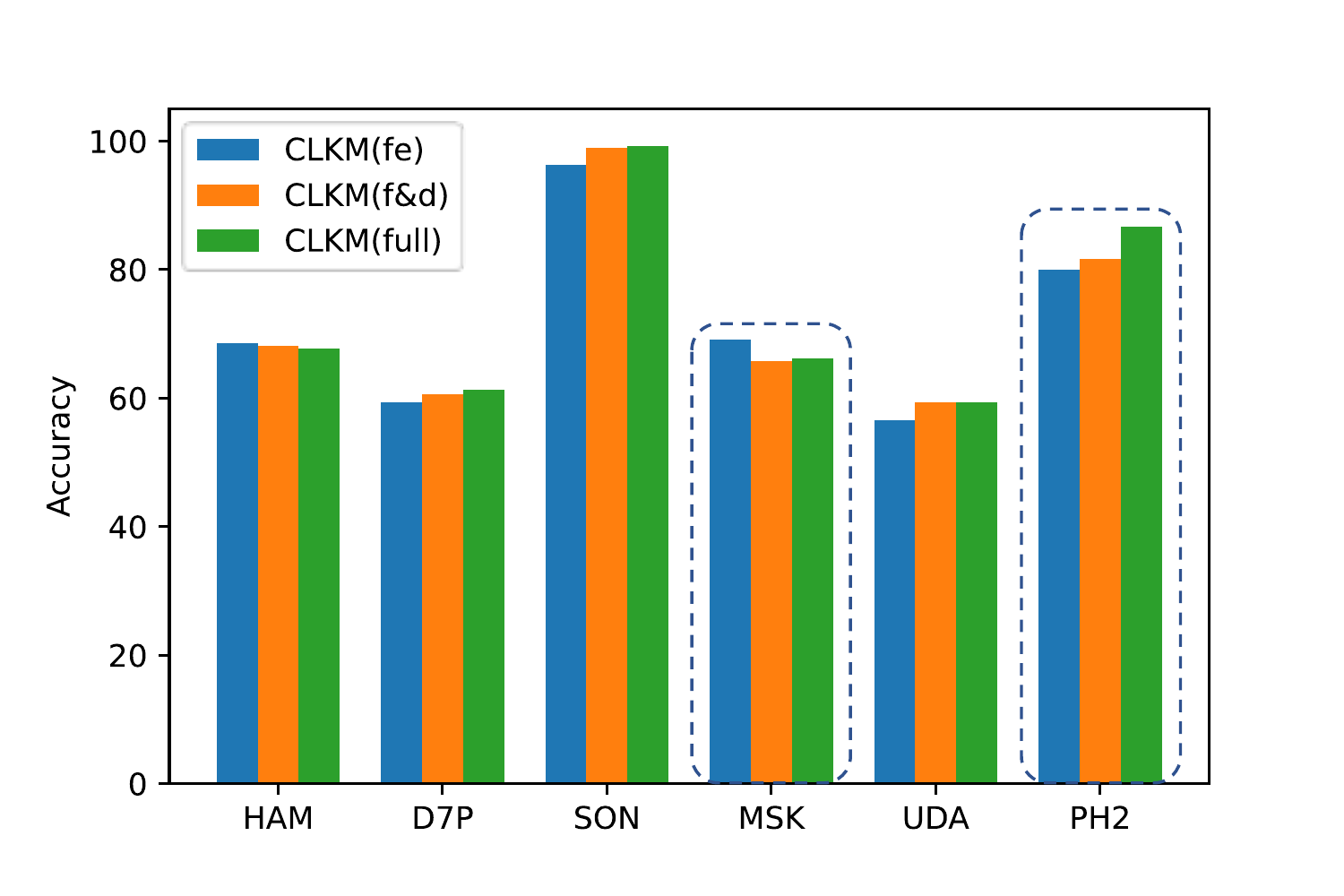}
\end{minipage}
}%
\subfigure[meta-testing: SON, PH2]{
\begin{minipage}[t]{0.32\linewidth}
\centering
\includegraphics[width=2.35in]{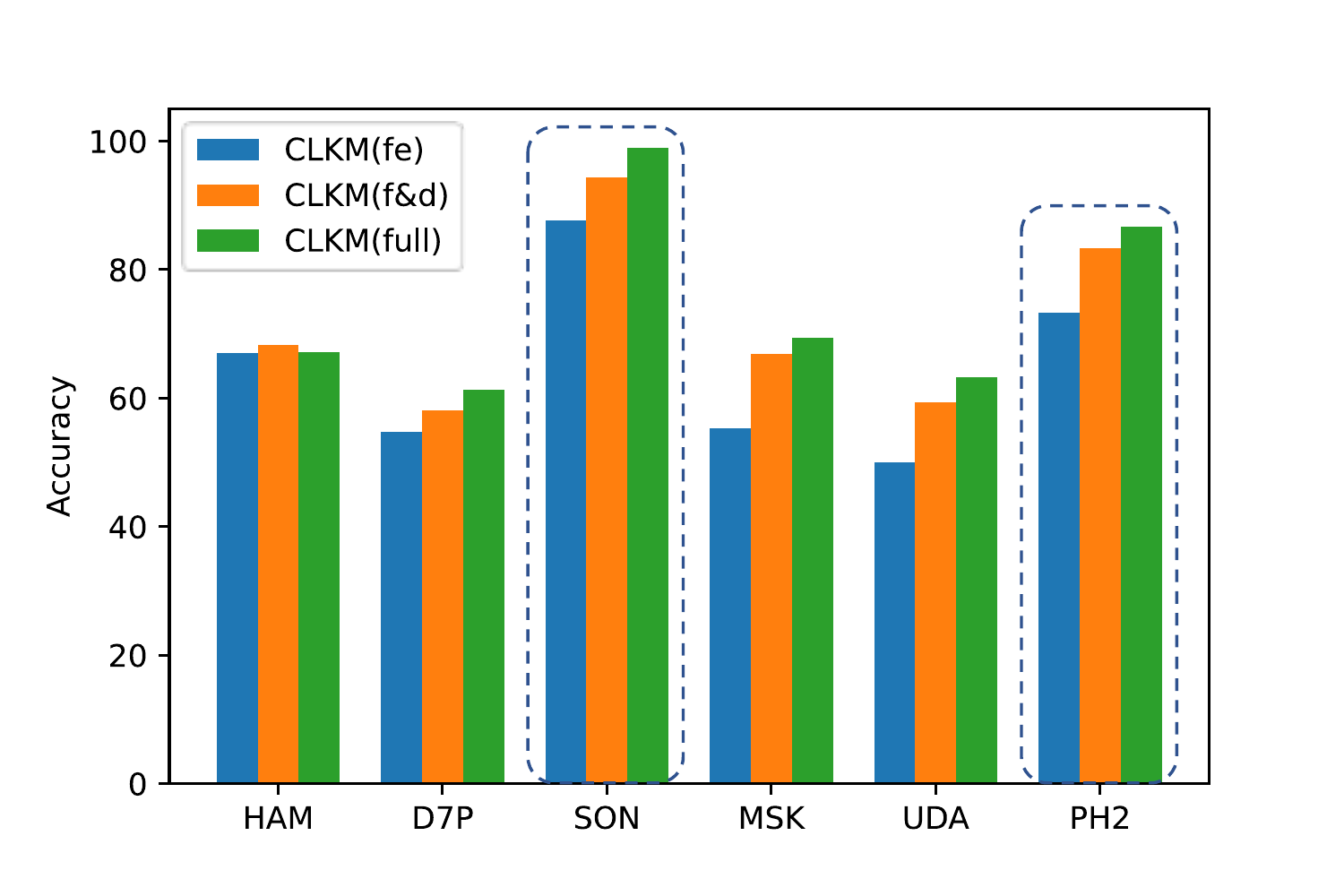}
\end{minipage}
}%
\vspace{-0.4cm}
\subfigure[meta-testing: D7P, SON]{
\begin{minipage}[t]{0.32\linewidth}
\centering
\includegraphics[width=2.35in]{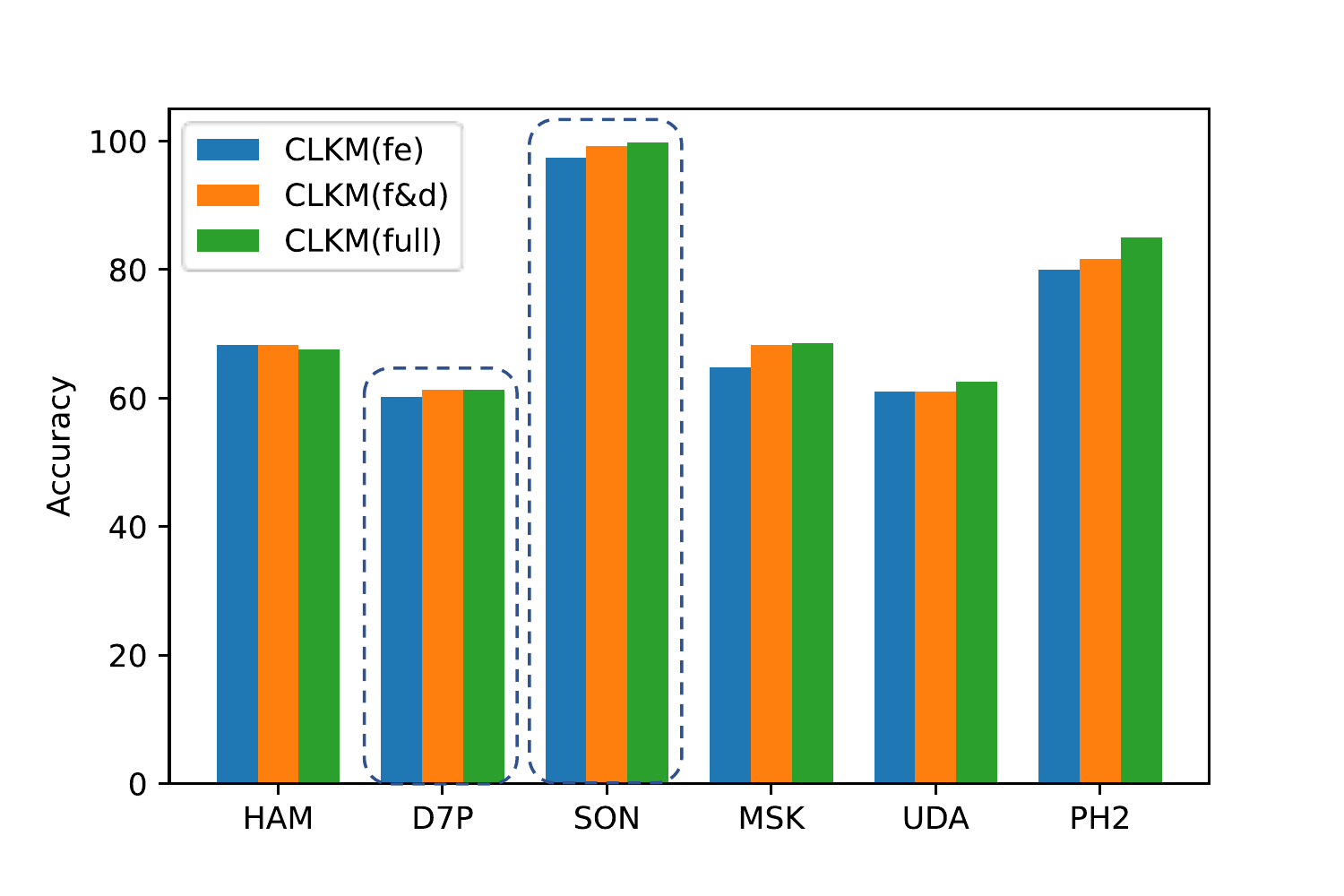}
\end{minipage}
}%
\subfigure[meta-testing: D7P, UDA]{
\begin{minipage}[t]{0.32\linewidth}
\centering
\includegraphics[width=2.35in]{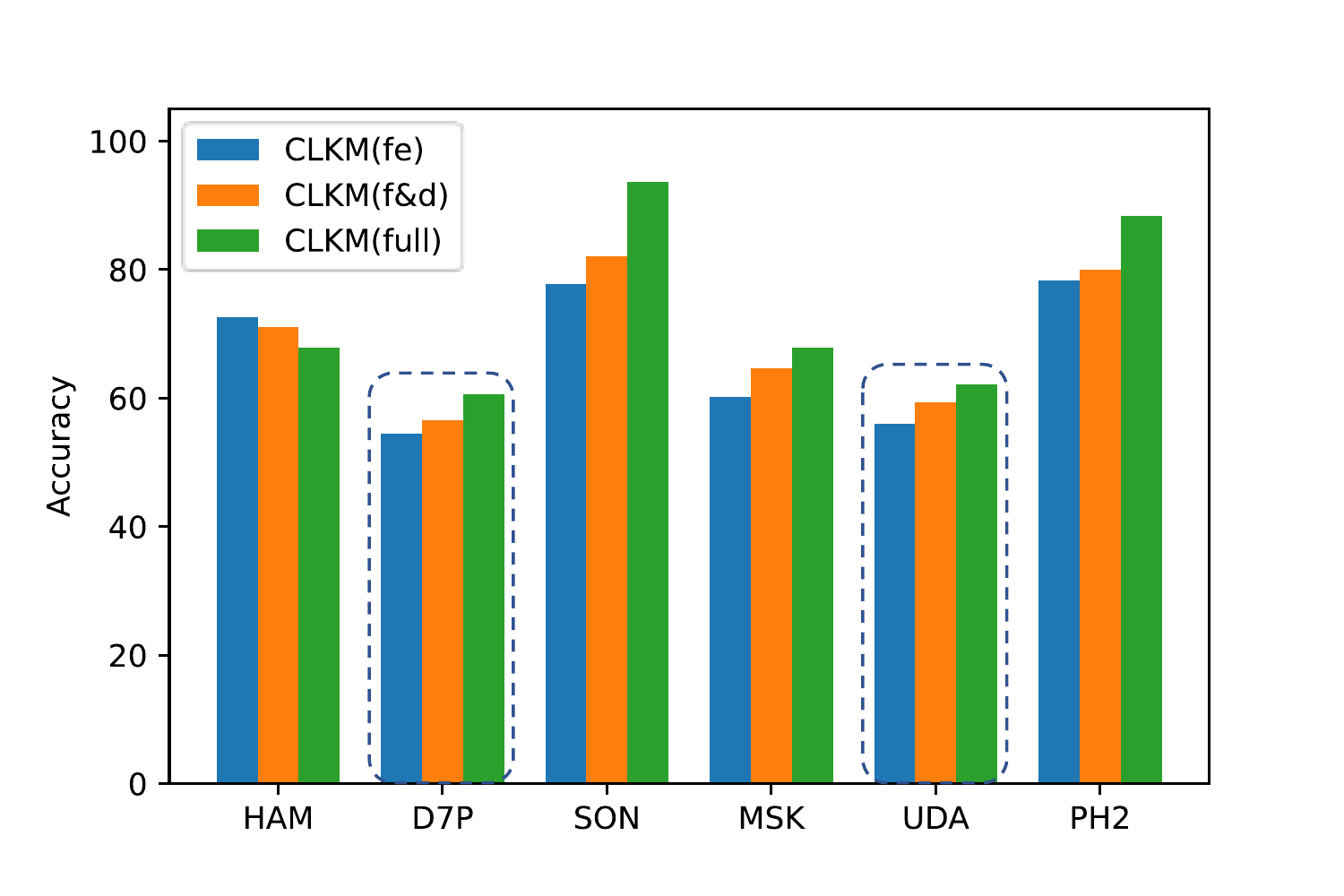}
\end{minipage}
}%
\subfigure[meta-testing: D7P, MSK]{
\begin{minipage}[t]{0.32\linewidth}
\centering
\includegraphics[width=2.35in]{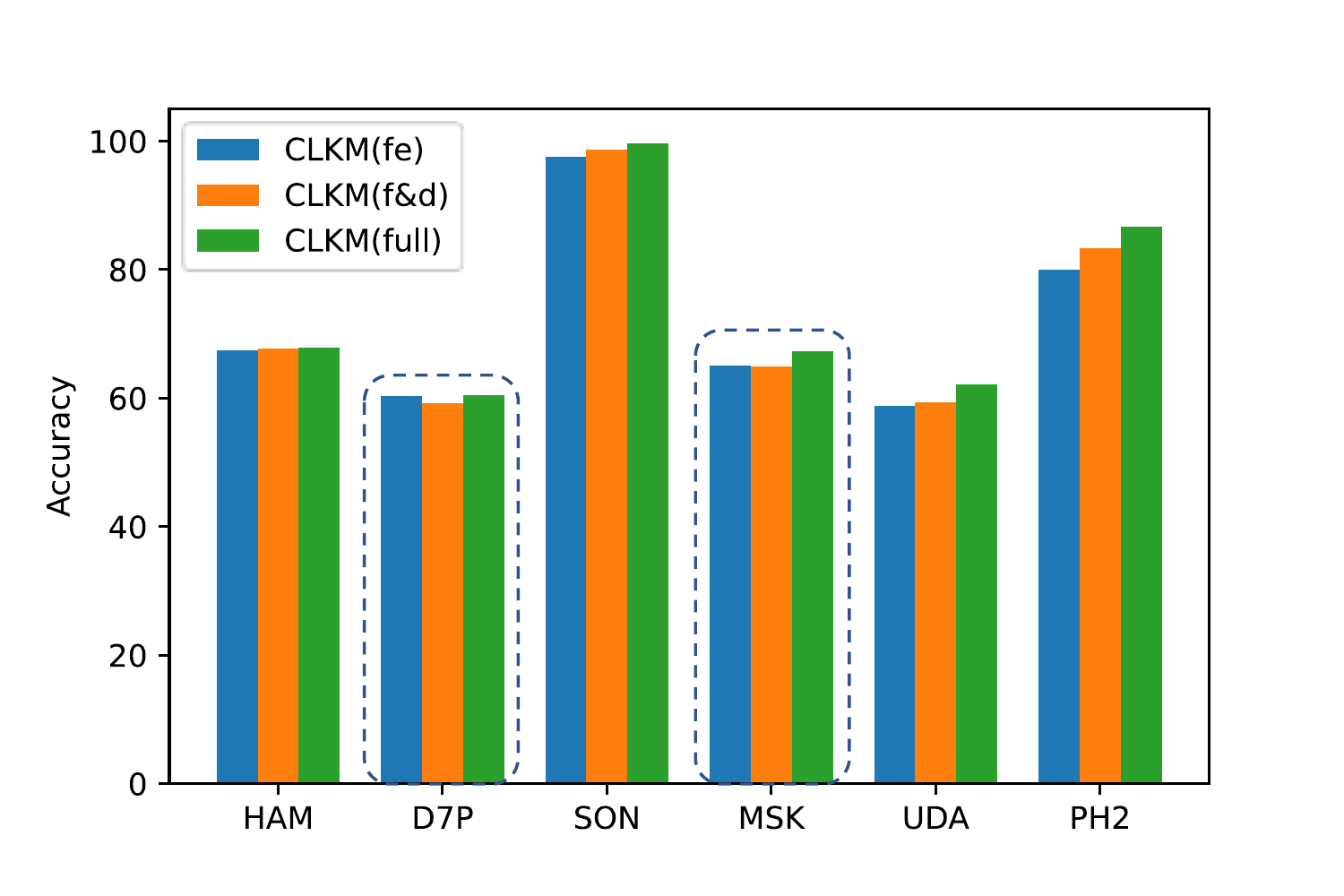}
\end{minipage}
}%
\setlength{\abovecaptionskip}{-2pt}
\setlength{\belowcaptionskip}{-1cm}
\caption{The accuracy of CLKM(fe), CLKM(f\&d), and CLKM(full) when different target lesion domains are selected as meta-testing datasets.}
\label{fig10}
\end{figure}

\begin{figure}[H]
\centering
\subfigcapskip = -8pt
\subfigure[meta-testing:UDA, PH2]{
\begin{minipage}[t]{0.32\linewidth}
\centering
\includegraphics[width=2.35in]{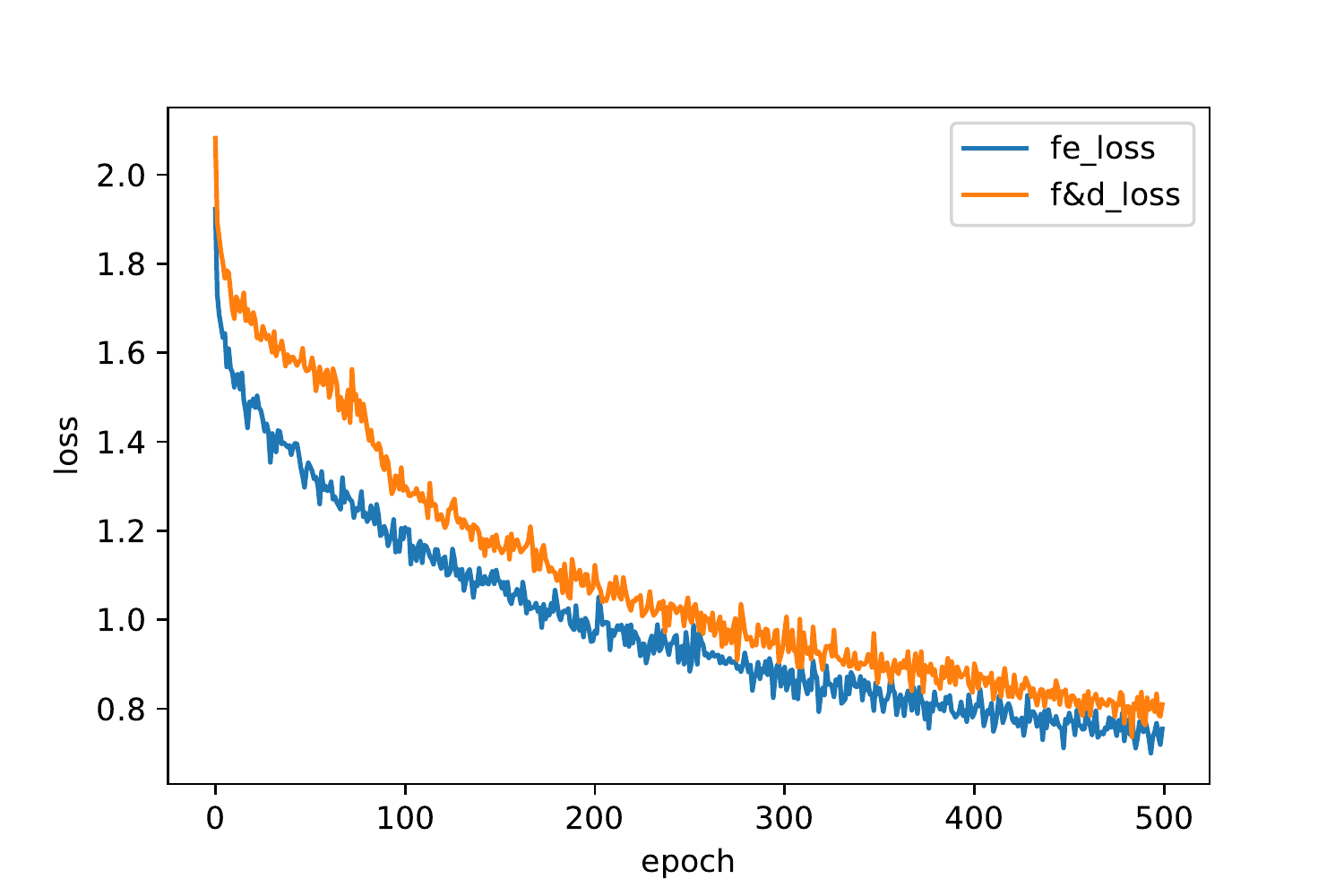}
\end{minipage}%
}
\subfigure[meta-testing:MSK, PH2]{
\begin{minipage}[t]{0.32\linewidth}
\centering
\includegraphics[width=2.35in]{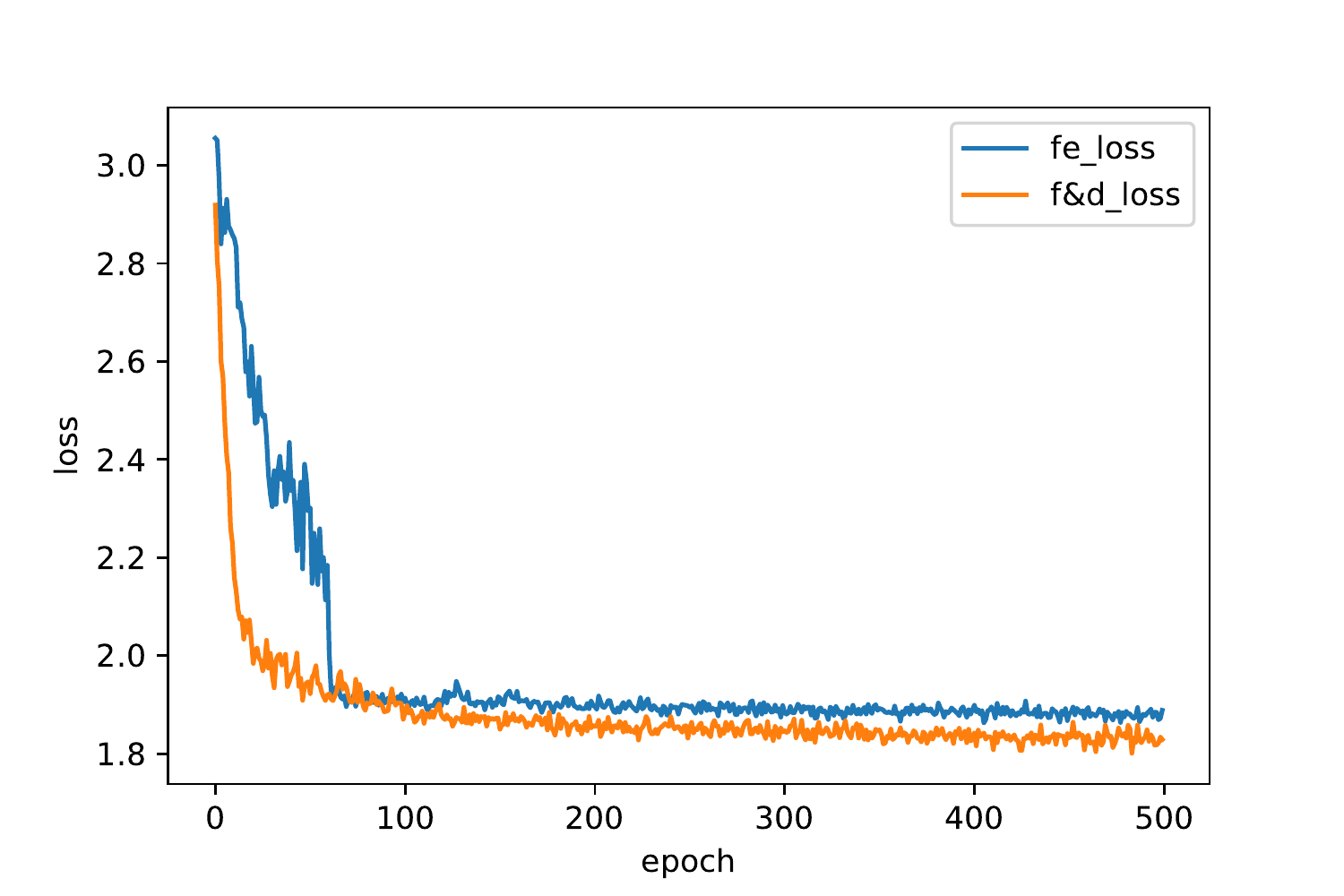}
\end{minipage}%
}
\subfigure[meta-testing:SON, PH2]{
\begin{minipage}[t]{0.32\linewidth}
\centering
\includegraphics[width=2.35in]{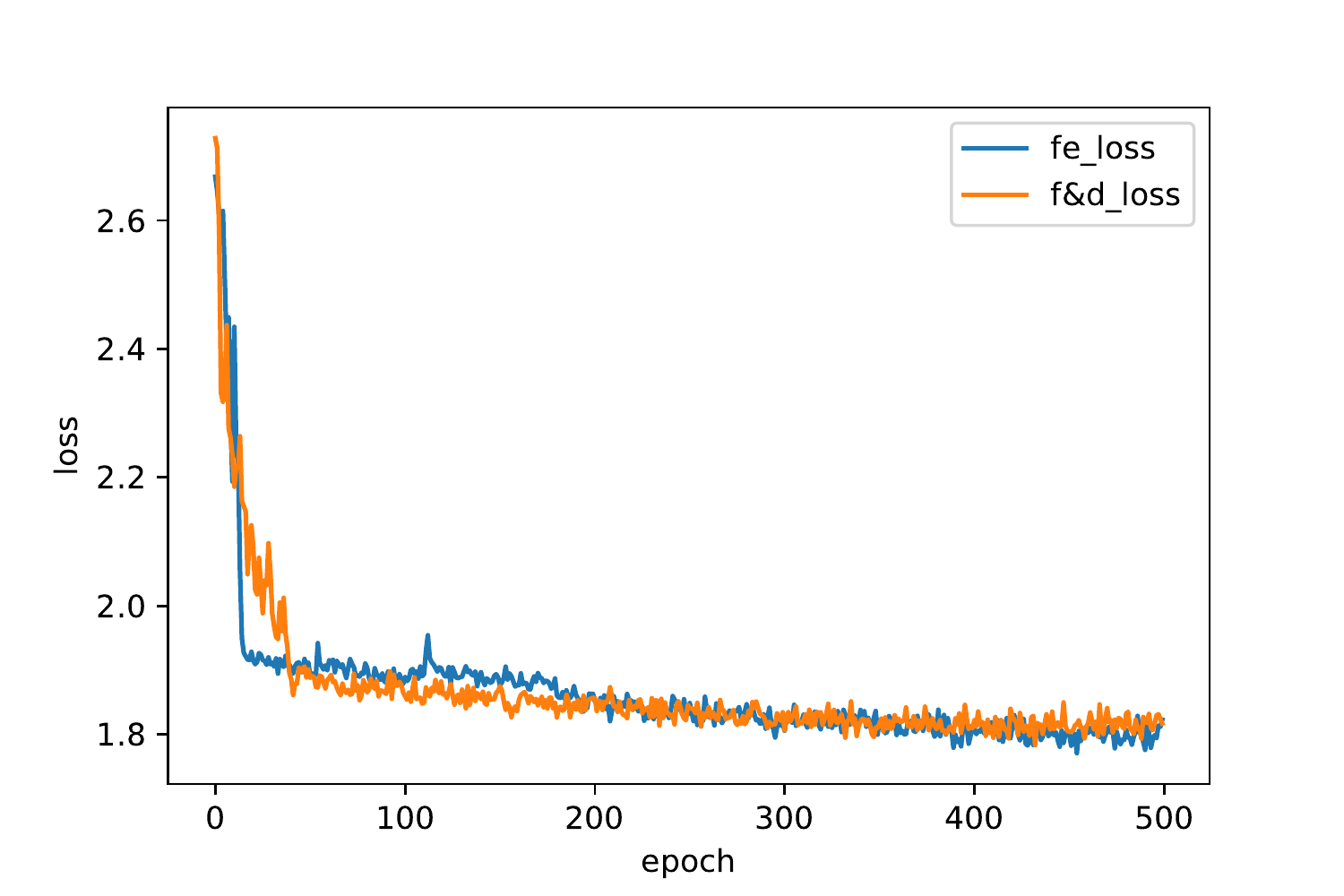}
\end{minipage}%
}%
\vspace{-0.5cm}
\subfigure[meta-testing:D7P, SON]{
\begin{minipage}[t]{0.32\linewidth}
\centering
\includegraphics[width=2.35in]{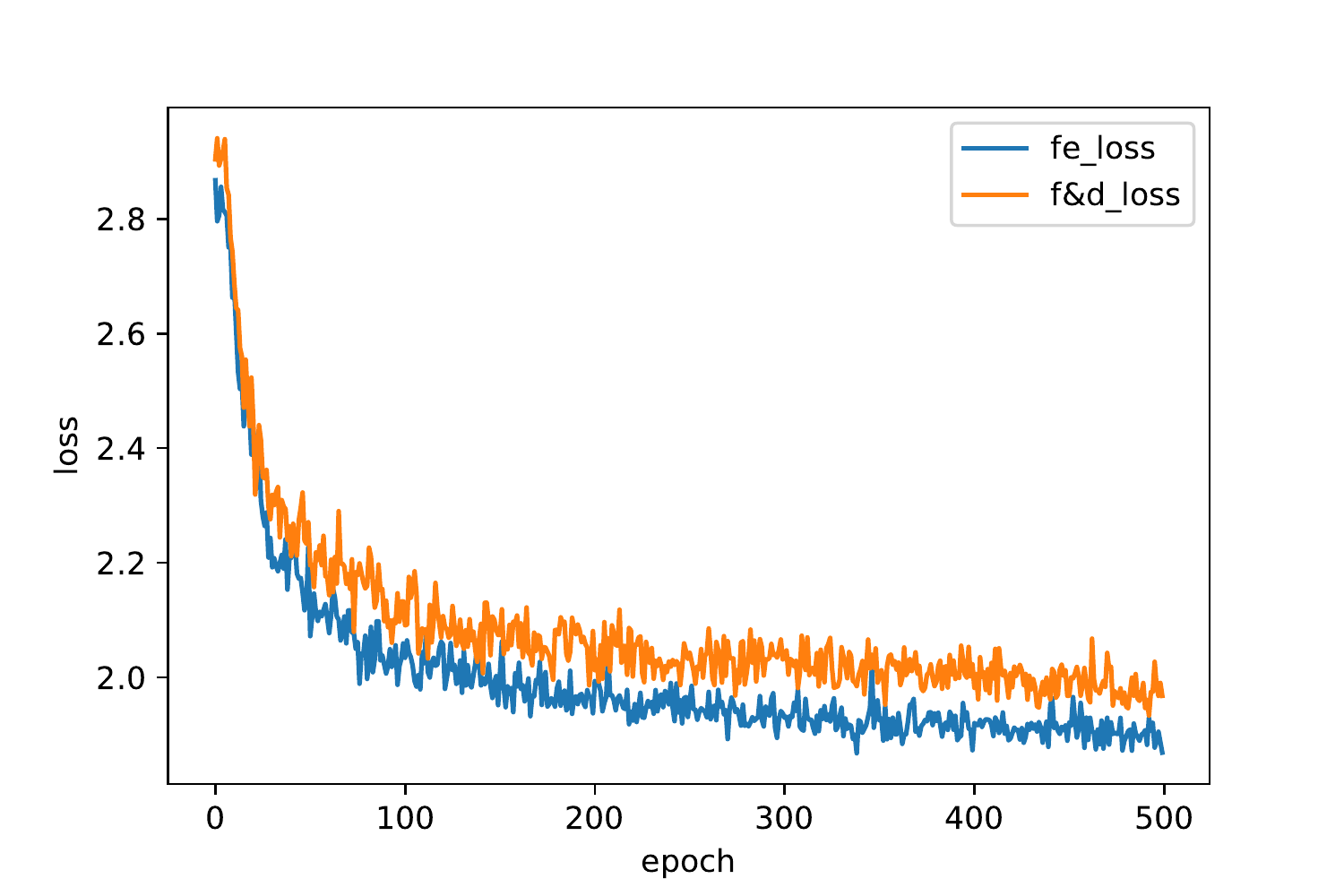}
\end{minipage}%
}
\subfigure[meta-testing:D7P, UDA]{
\begin{minipage}[t]{0.32\linewidth}
\centering
\includegraphics[width=2.35in]{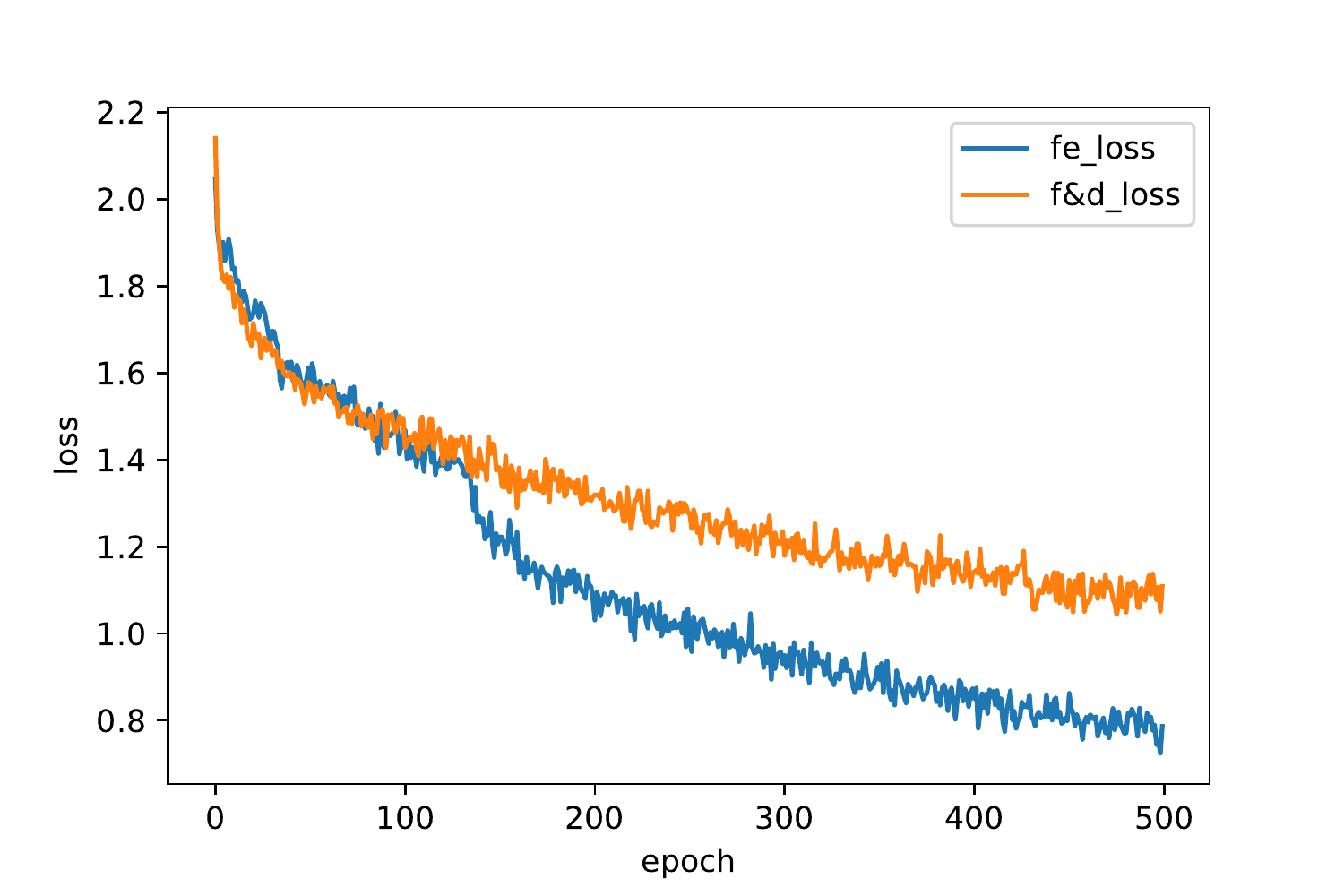}
\end{minipage}%
}
\subfigure[meta-testing:D7P, MSK]{
\begin{minipage}[t]{0.32\linewidth}
\centering
\includegraphics[width=2.35in]{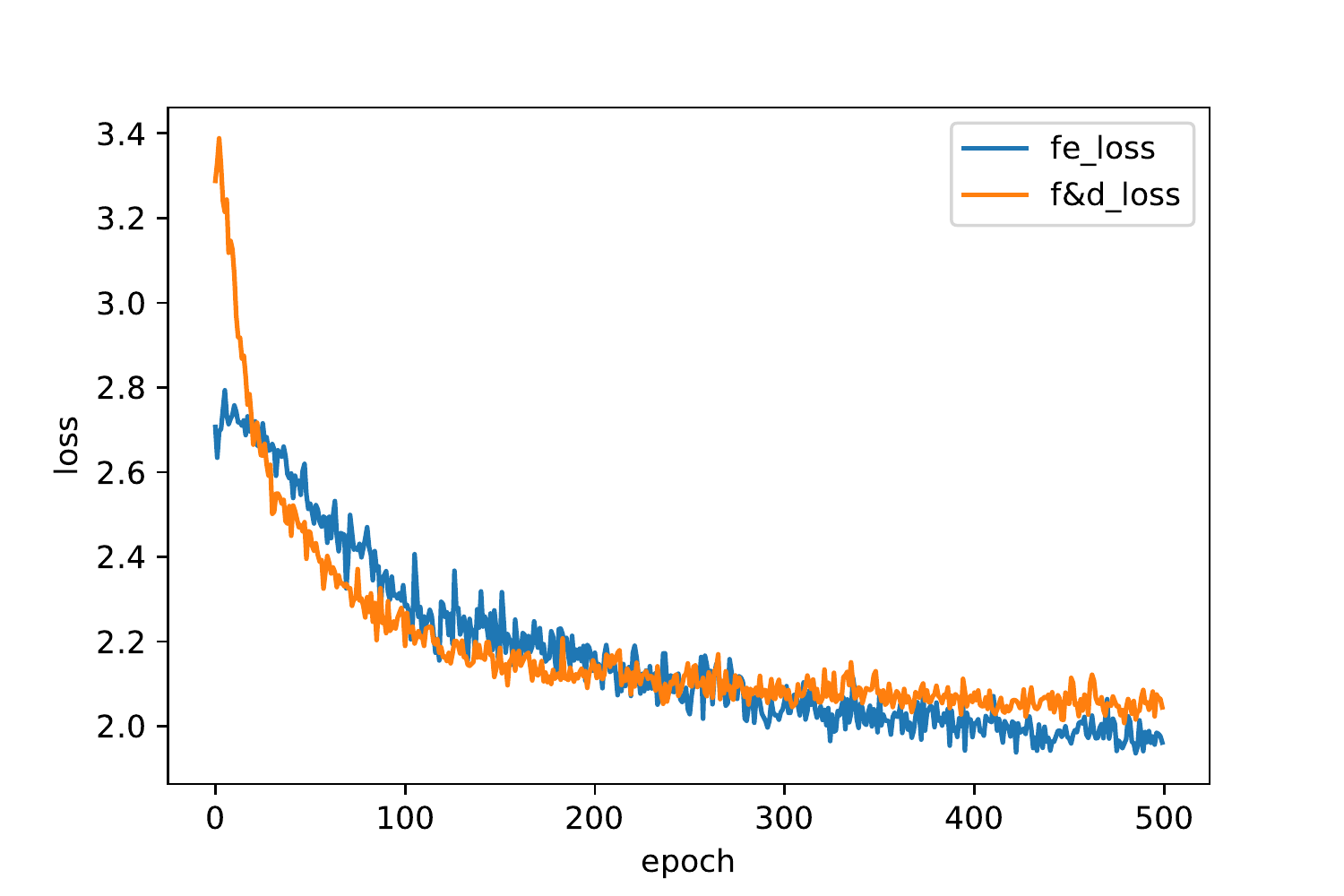}
\end{minipage}%
}
\setlength{\abovecaptionskip}{0pt}
\caption{\textbf{The training loss curves on different meta-testing datasets.} On the whole, when the penalty term is added to the loss function, the training loss is slightly increased.}
\label{fig:fig9}
\end{figure}

\textbf{(c) Meta-testing datasets: D7P, MSK.}
The results in Table~\ref{tab4} show that the discrete adaptive methods can slightly improve model performance compared with no adaptation, and the CLKM is suitable for this situation.
CLKM captures the differences across lesion domains, meantime, while maintaining the accuracy of the previously learned domains.

\textbf{(d) Meta-testing datasets: MSK, PH2.} In this setting, the distribution of samples in the meta-training datasets (D7P, UDA, and SONIC) is smooth to some extent.
From D7P to SON, the number of major categories decreases gradually.
In such cases, JAN performed better than DAN. 
Similarly result also appear in Table~\ref{tab4} and Table~\ref{tab5}.

\begin{figure}[H]
\centering
\subfigcapskip = -15pt
\subfigure[CLKM(full)]{
\begin{minipage}[t]{0.32\linewidth}
\centering
\includegraphics[width=2.3in]{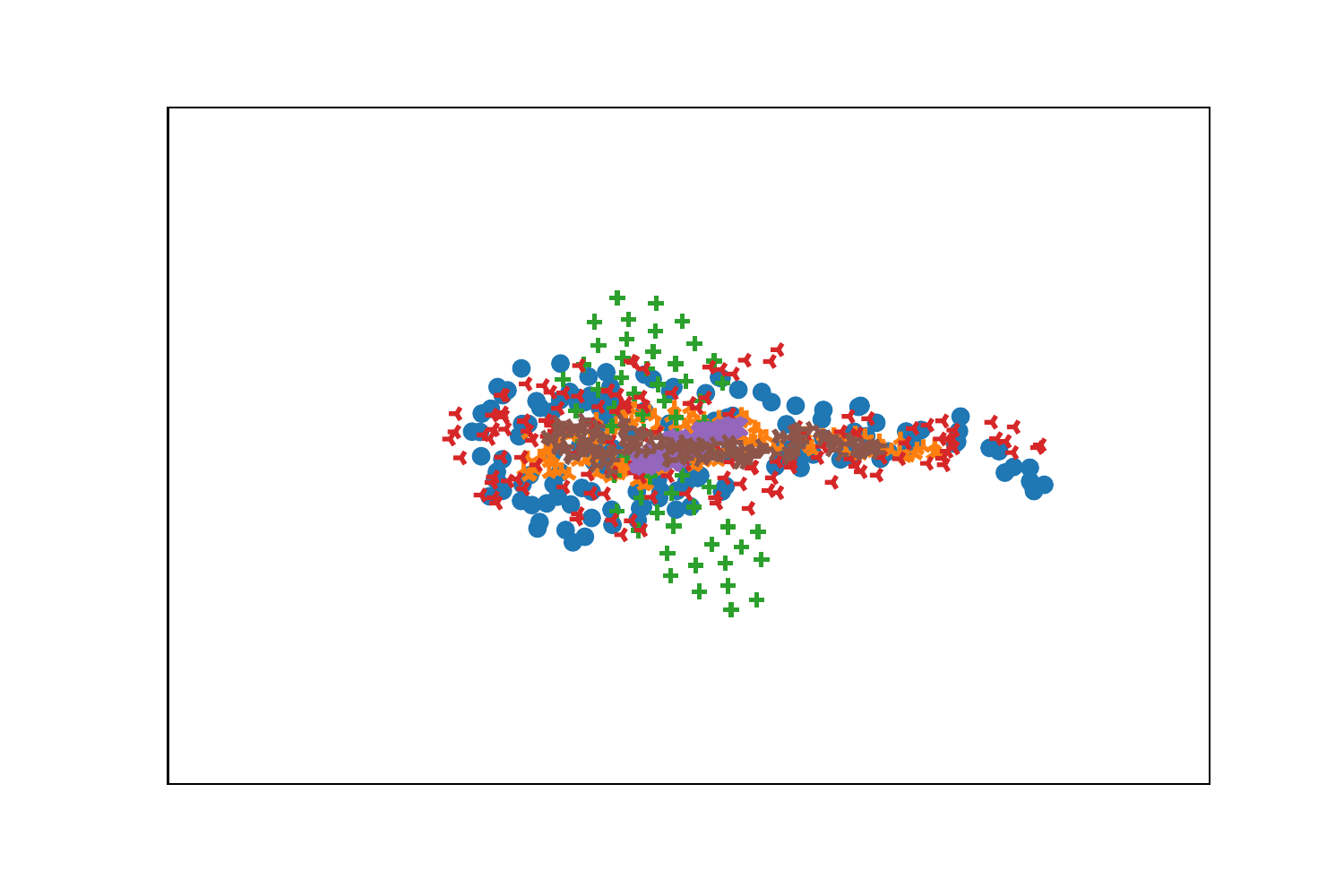}
\end{minipage}%
}
\subfigure[DAN]{
\begin{minipage}[t]{0.32\linewidth}
\centering
\includegraphics[width=2.3in]{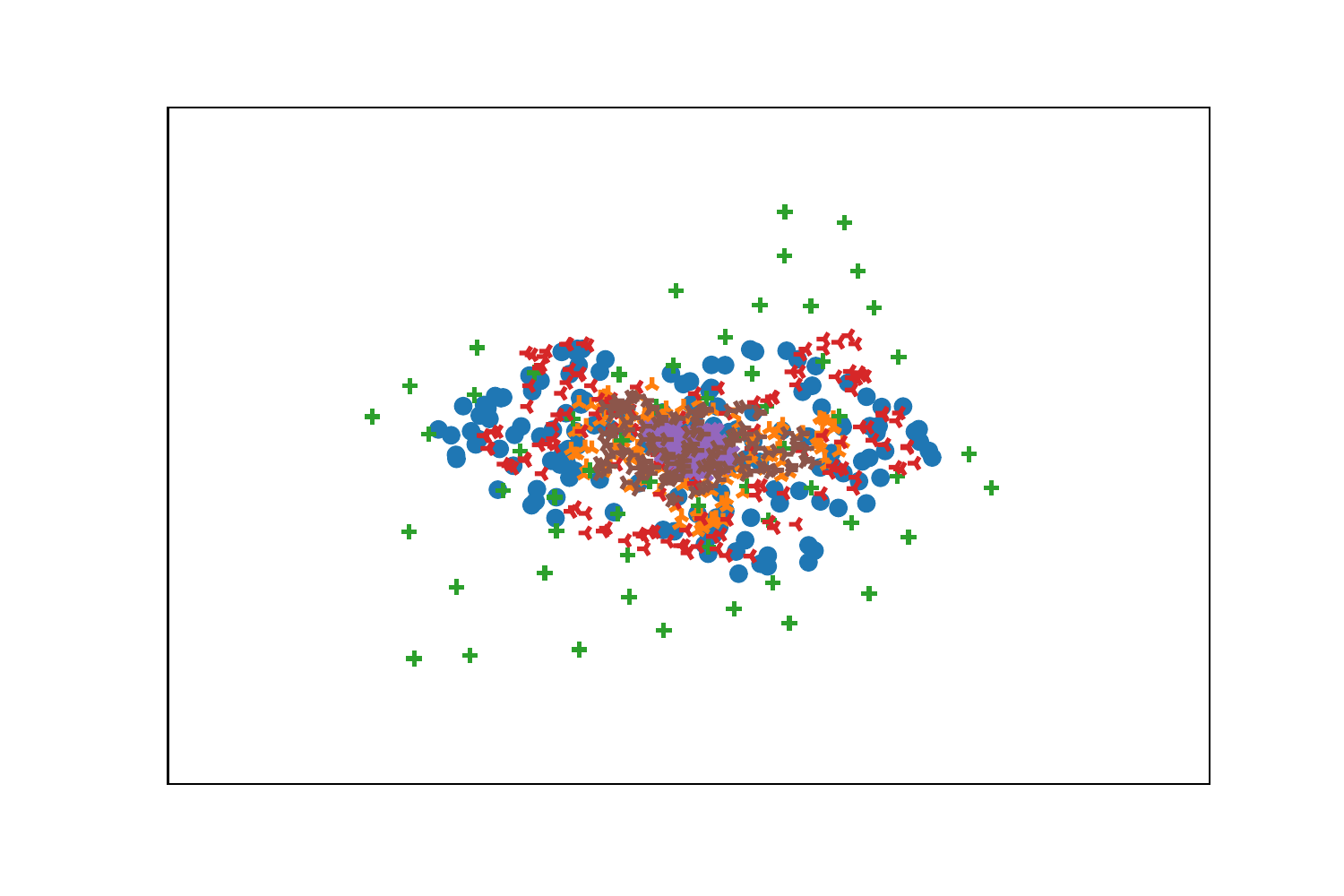}
\end{minipage}%
}
\subfigure[JAN]{
\begin{minipage}[t]{0.32\linewidth}
\centering
\includegraphics[width=2.3in]{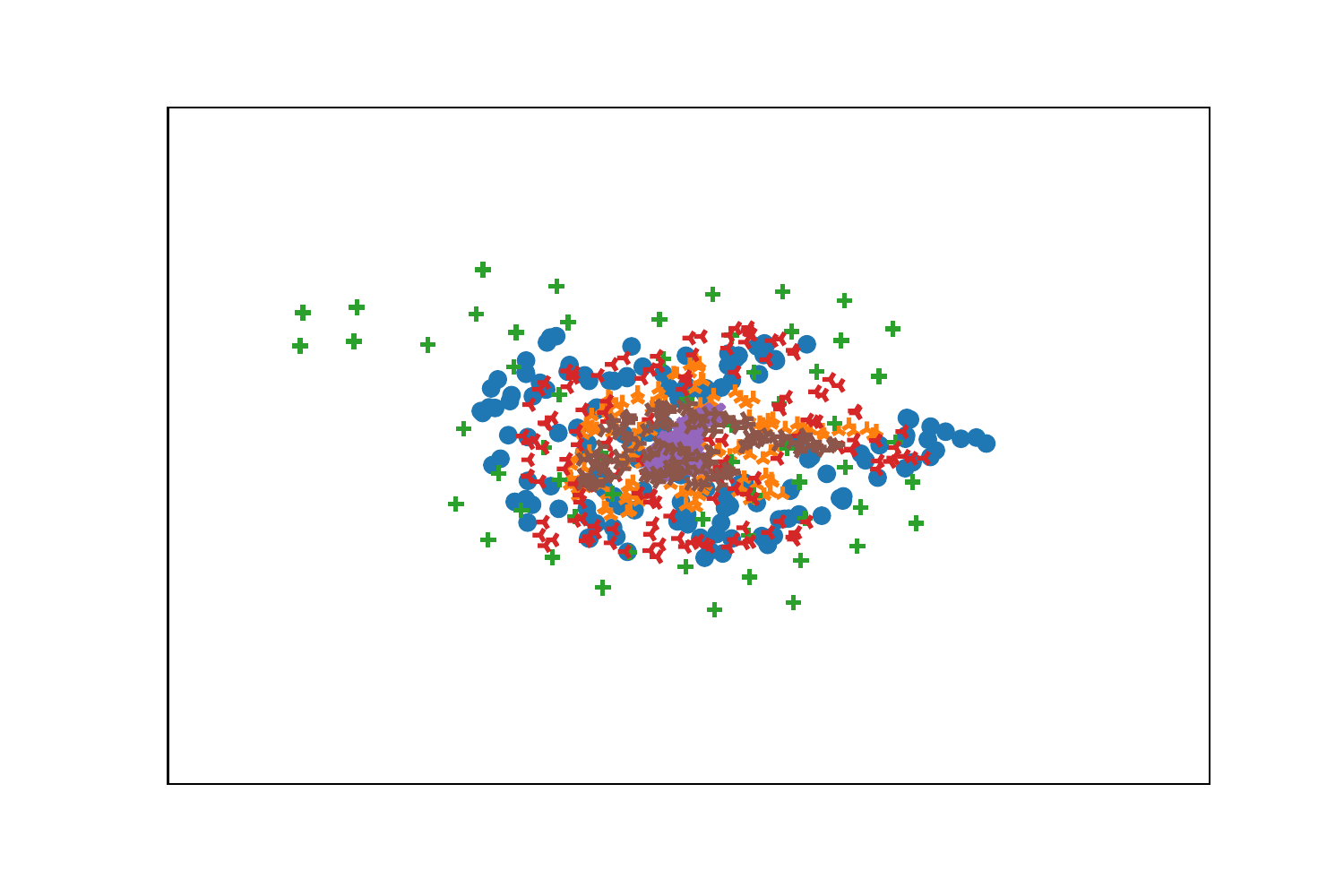}
\end{minipage}%
}%
\setlength{\abovecaptionskip}{-1pt}
\caption{The $t$-SNE feature visualizations when meta-training datasets are \textbf{D7P}, \textbf{UDA}, and \textbf{SONIC}.}
\label{fig:fig6}
\end{figure}
\vspace{-1.0cm}

\begin{figure}[H]
\centering
\subfigcapskip = -15pt
\subfigure[CLKM(full)]{
\begin{minipage}[t]{0.32\linewidth}
\centering
\includegraphics[width=2.35in]{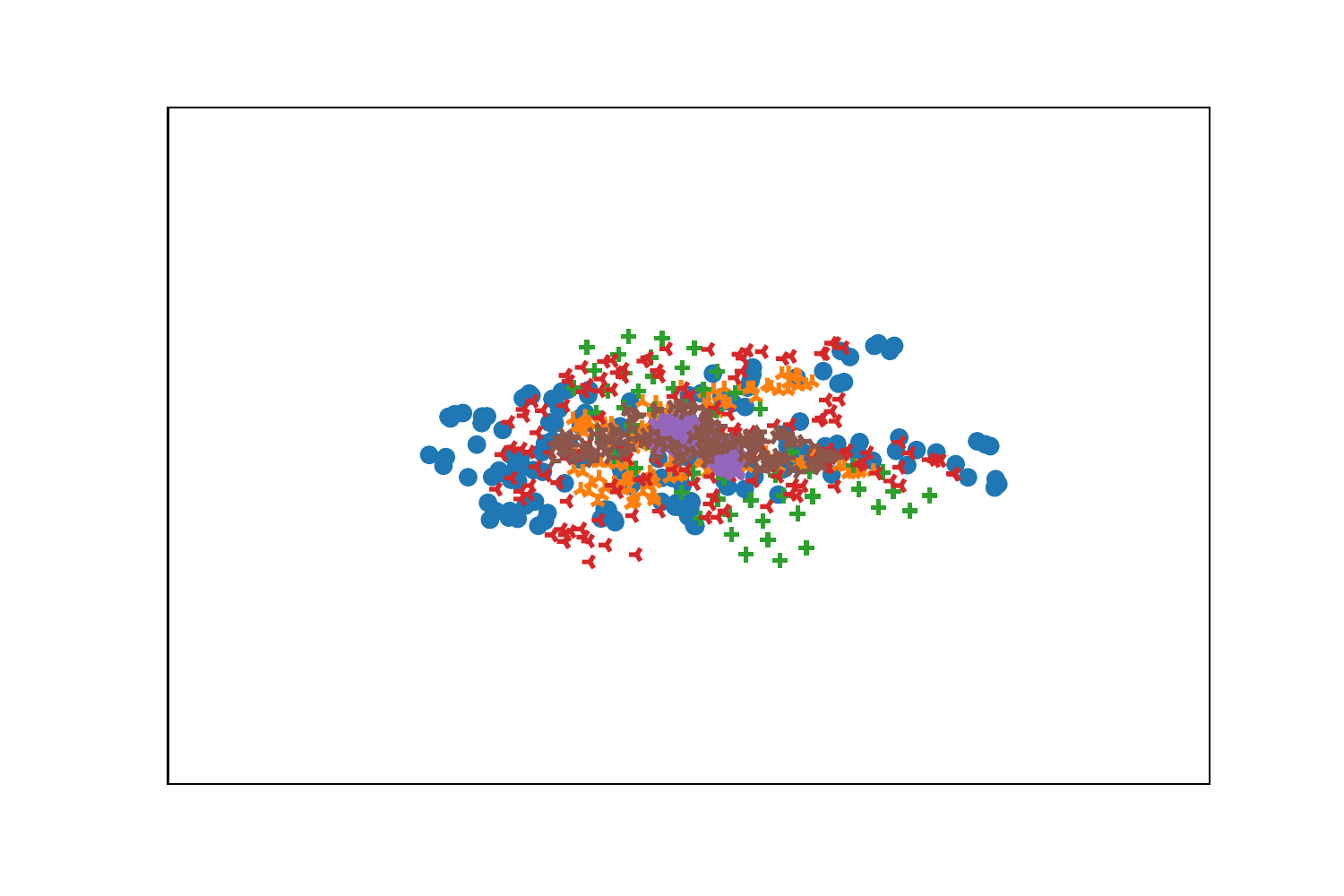}
\end{minipage}%
}
\subfigure[DAN]{
\begin{minipage}[t]{0.32\linewidth}
\centering
\includegraphics[width=2.35in]{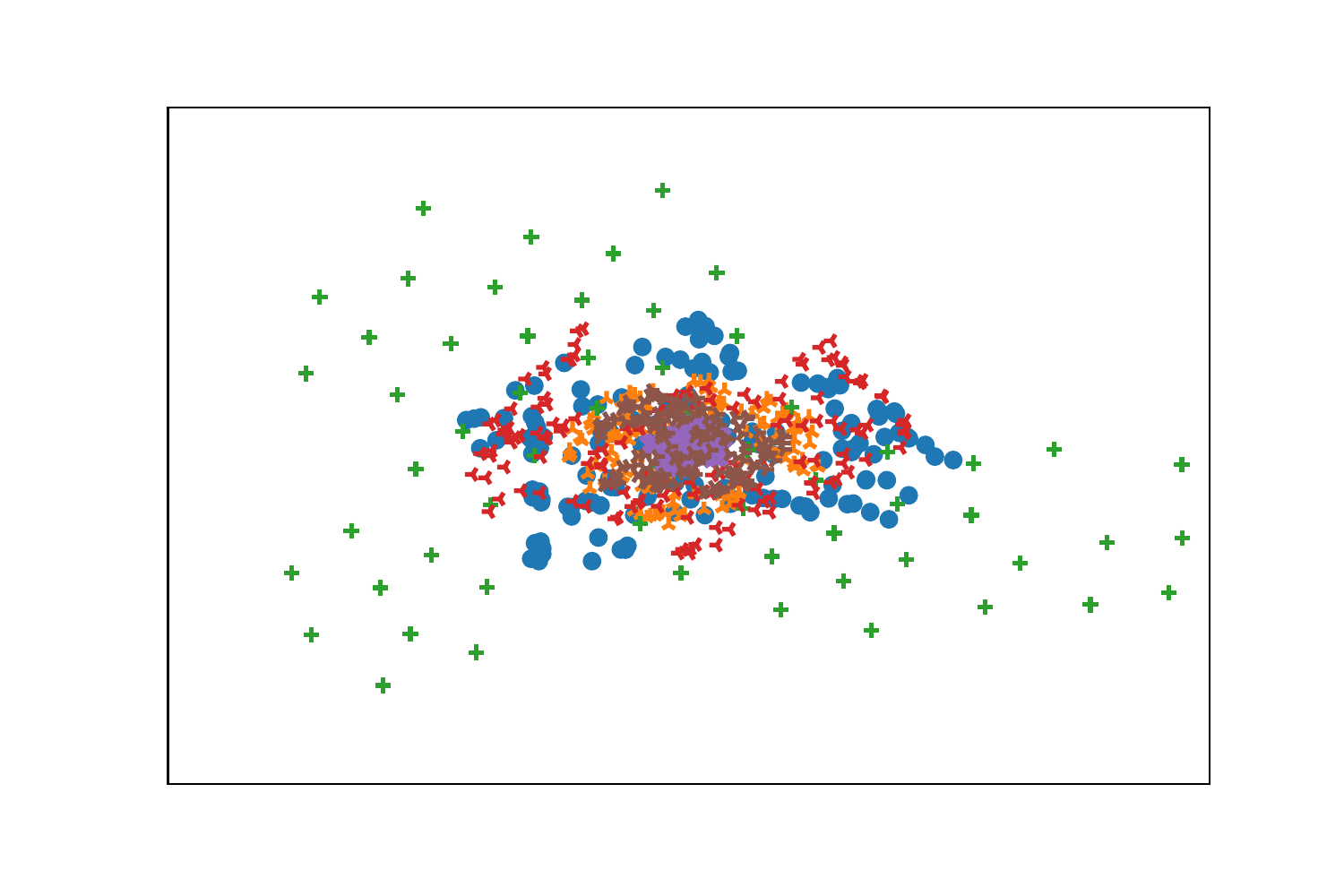}
\end{minipage}%
}
\subfigure[JAN]{
\begin{minipage}[t]{0.32\linewidth}
\centering
\includegraphics[width=2.35in]{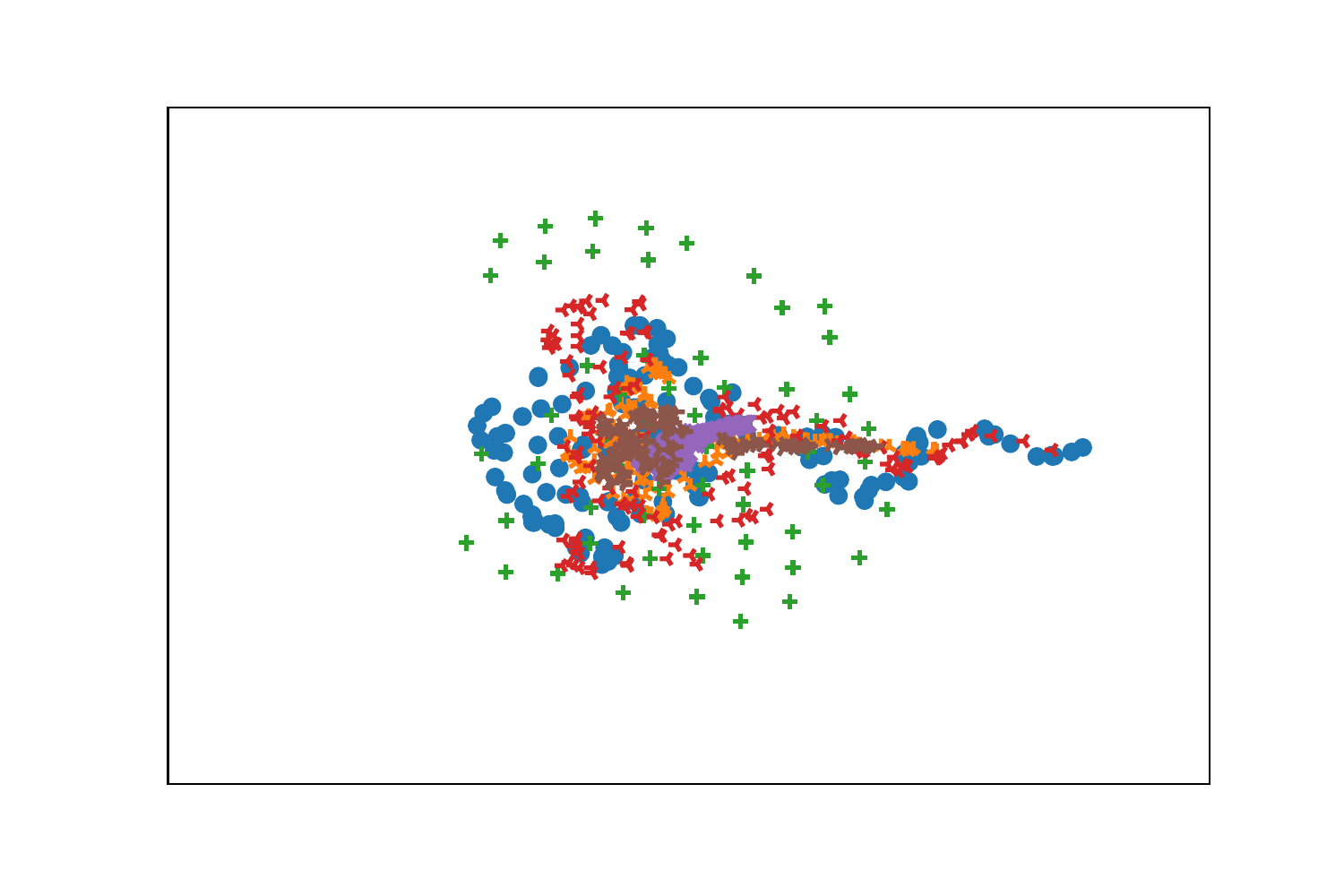}
\end{minipage}%
}%
\setlength{\abovecaptionskip}{-1pt}
\caption{The $t$-SNE feature visualizations when meta-training datasets are \textbf{PH2}, \textbf{UDA}, and \textbf{MSK}.}
\label{fig:fig7}
\end{figure}
\vspace{-1.0cm}

\begin{figure}[H]
\centering
\subfigcapskip = -15pt
\subfigure[CLKM(full)]{
\begin{minipage}[t]{0.32\linewidth}
\centering
\includegraphics[width=2.35in]{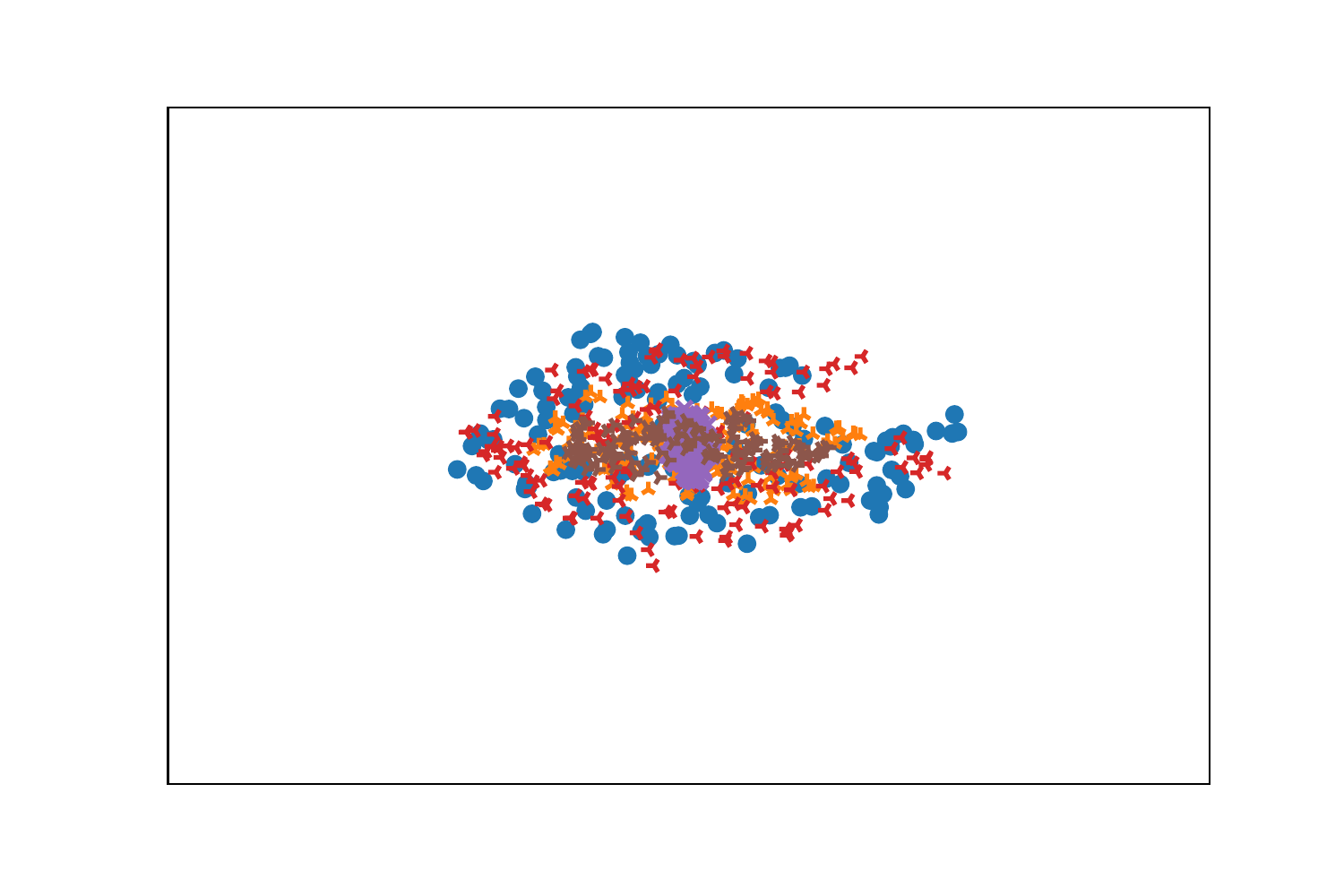}
\end{minipage}%
}
\subfigure[DAN]{
\begin{minipage}[t]{0.32\linewidth}
\centering
\includegraphics[width=2.35in]{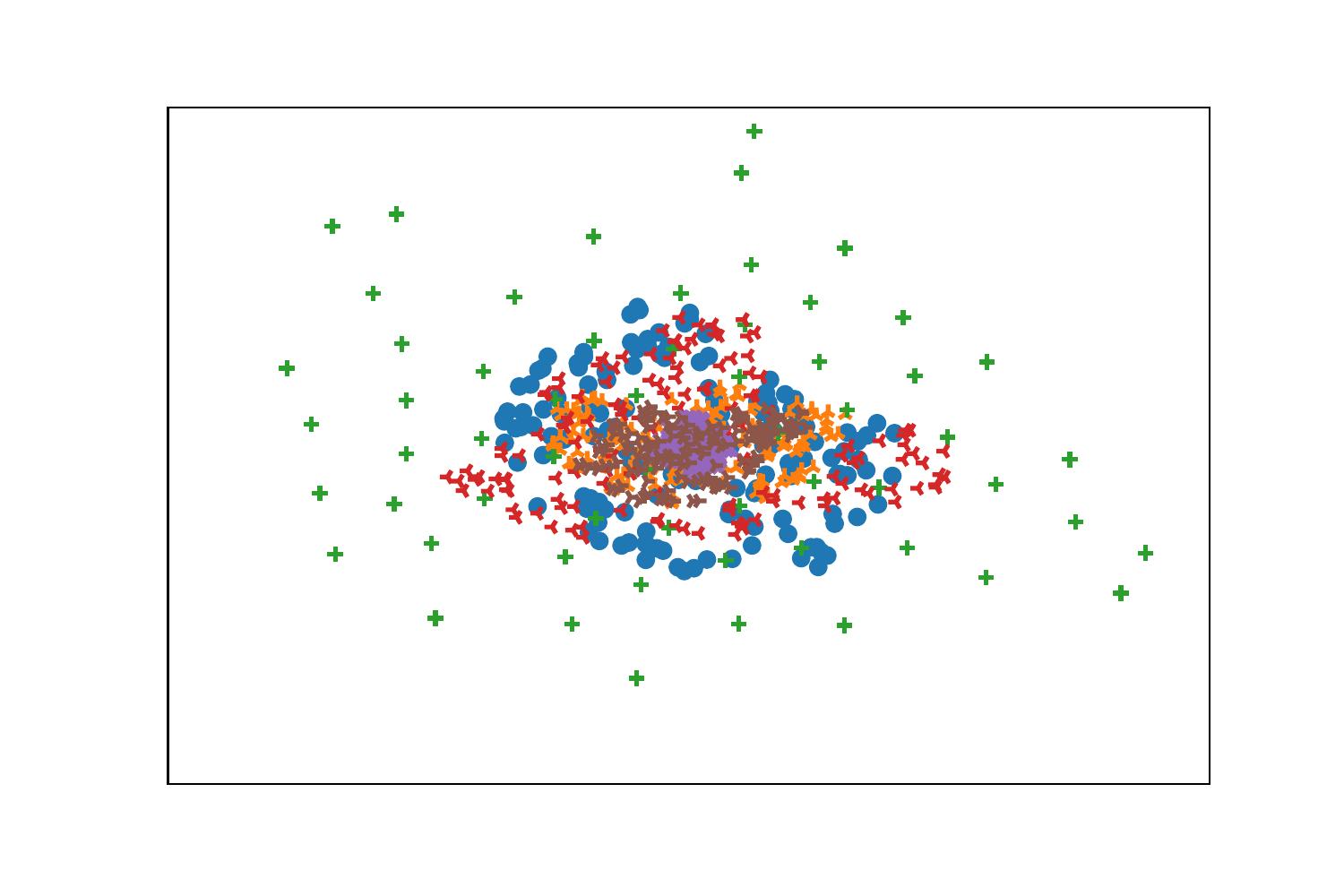}
\end{minipage}%
}
\subfigure[JAN]{
\begin{minipage}[t]{0.32\linewidth}
\centering
\includegraphics[width=2.35in]{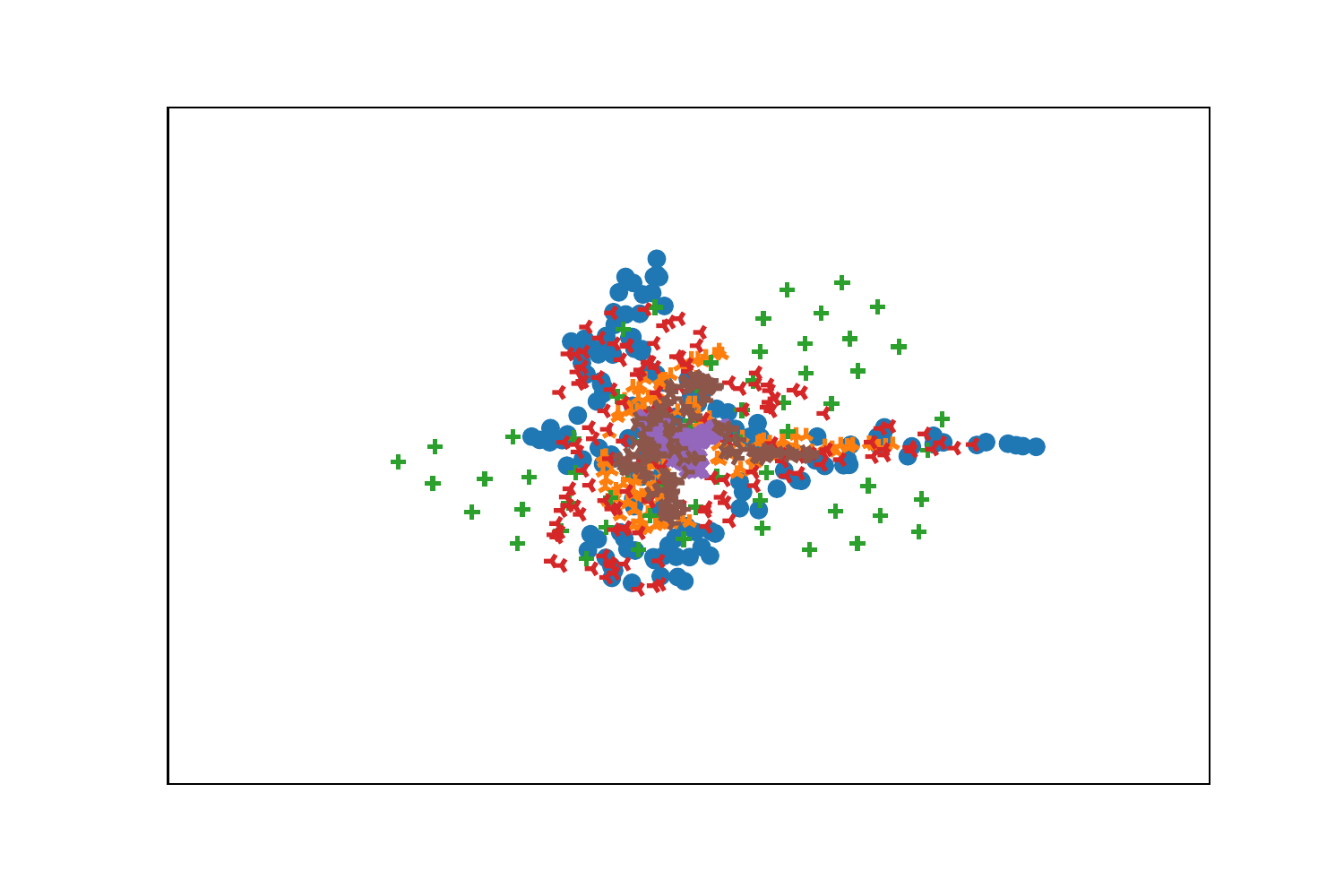}
\end{minipage}%
}%
\setlength{\abovecaptionskip}{-1pt}
\caption{The $t$-SNE feature visualizations when meta-training datasets are \textbf{SONIC}, \textbf{UDA} ,and \textbf{PH2}.}
\label{fig:fig8}
\end{figure}

\begin{figure}[H]
    \centering
    \includegraphics[width = 9cm]{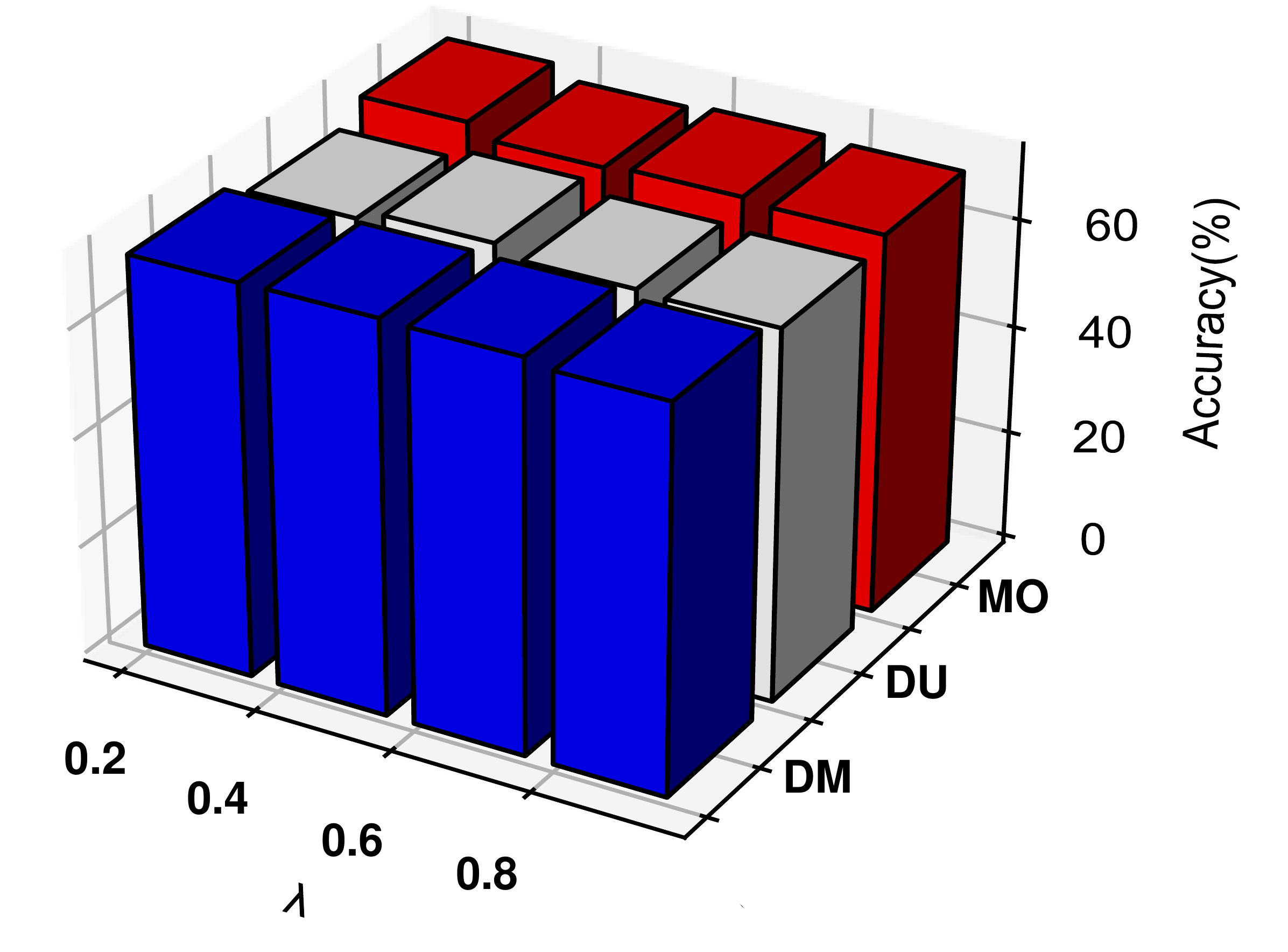}
    \caption{The model performance with different levels $\lambda$ on different meta-testing datasets.}
    \label{fig11}
\end{figure}

\textbf{(e) Meta-testing datasets: D7P, UDA.}
When we select MSK, PH2, and SON as the meta-training datasets, note that the performance of DAN is worse than SourceOnly.
DAN force-distinguishes the source and target domains, which has the opposite effect.
Although the performance is improved after adding the regularization constraint, still inferior to SourceOnly.

\textbf{(f) Meta-testing datasets: UDA, PH2.} 
In this setting, we found that the domain adaptation method JAN performed even worse than the model without adaptation.
The data distributions of MSK, D7P, and SON are significantly different, however, JAN intends to narrow the discrepancy across multiple domains at the same time, resulting in a degradation of the model performance.

\begin{center}
    Table 8: Classification accuracy(\%) for two meta-training datasets
\end{center}

\begin{minipage}{\textwidth}
    \begin{minipage}[l]{0.31\textwidth}
        \centering
        \setlength{\tabcolsep}{1.5mm}
        \setlength{\abovecaptionskip}{0cm}
        \setlength{\belowcaptionskip}{0cm}
        \footnotesize
            \makeatletter\def\@captype{table}\makeatother

            \caption*{(a) Meta-training: \textbf{UDA, PH2}}
            
            \begin{tabular}{c c c c c}
                \hline
                
                \hline
                 & fe & dq & f\&d & full \\
                 \hline
                 
                 \hline
                 HAM& 67.6 & 68.1 & 68.0 & 67.6\\
                 D7P& 53.0 & 52.5 & 54.0 & 53.5\\
                 SON& 90.1 & 85.4 & 92.1 & 93.8\\
                 MSK& 62.1 & 58.0 & 61.2 & 63.5\\
                 \textbf{UDA}& 58.0 & 60.0 & 61.1 & 61.0\\
                 \textbf{PH2}& 75.0 & 80.0 & 83.3 & 85.0\\
                 \hline
                 Avg& 67.6 & 67.3 & 69.9 & 70.7\\

                \hline
                
                \hline
                \end{tabular}
                
    \end{minipage}%
    \hspace{0mm}
    \begin{minipage}[l]{0.31\textwidth}
        \centering
        \setlength{\tabcolsep}{1.5mm}
        \footnotesize
            \makeatletter\def\@captype{table}\makeatother \setlength{\abovecaptionskip}{0cm}    \setlength{\belowcaptionskip}{0cm}
            \caption*{(b) Meta-training: \textbf{MSK, PH2}}
            
                \begin{tabular}{c c c c c}
                \hline
                
                \hline
                 & fe & dq & f\&d & full \\
                 \hline
                 
                 \hline
                 HAM& 71.2 & 70.5 & 71.5 & 69.5\\
                 D7P& 52.4 & 53.5 & 54.3 & 56.4\\
                 SON& 75.0 & 80.0 & 83.3 & 92.5\\
                 \textbf{MSK}& 60.5 & 58.3 & 63.7 & 63.3\\
                 UDA& 56.5 & 54.9 & 58.0 & 58.0\\
                 \textbf{PH2}& 78.3 & 81.7 & 80.0 & 85.0\\
                 \hline
                 Avg& 65.7 & 66.5 & 68.5 & 70.8\\
                
                \hline
                
                \hline
                \end{tabular}
                
    \end{minipage}%
    \hspace{0mm}
    \begin{minipage}[l]{0.31\textwidth}
        \centering
        \setlength{\tabcolsep}{1.5mm}
        \footnotesize
            \makeatletter\def\@captype{table}\makeatother \setlength{\abovecaptionskip}{0cm}    \setlength{\belowcaptionskip}{0cm}
            \caption*{(c) Meta-training: \textbf{D7P, MSK}}
            
                \begin{tabular}{c c c c c}
                \hline
                
                \hline
                 & fe & dq & f\&d & full \\
                 \hline
                 
                 \hline
                 HAM& 67.3 & 68.5 & 68.5 & 68.0\\
                 \textbf{D7P}& 53.5 & 58.3 & 57.6 & 61.3\\
                 SON& 82.3 & 79.5 & 93.2 & 96.0\\
                 \textbf{MSK}& 59.5 & 63.2 & 60.1 & 67.3\\
                 UDA& 57.7 & 53.2 & 59.3 & 62.1\\
                 PH2& 73.3 & 78.3 & 81.7 & 83.3\\
                 \hline
                 Avg& 65.6 & 66.8 & 70.1 & 73.0\\
                
                \hline
                
                \hline

                \end{tabular}
                
    \end{minipage}
\end{minipage}

\begin{minipage}{\textwidth}
    \begin{minipage}[l]{0.3\textwidth}
        \centering
        \setlength{\tabcolsep}{1.5mm}
        \setlength{\abovecaptionskip}{0.2cm}
        \setlength{\belowcaptionskip}{0.2cm}
        \footnotesize
            \makeatletter\def\@captype{table}\makeatother

            \caption*{(d) Meta-training: \textbf{MSK, UDA}}
            
                \begin{tabular}{c c c c c}
                \hline
                
                \hline
                 & fe & dq & f\&d & full \\
                 \hline
                 
                 \hline
                 HAM& 68.2 & 68.4 & 67.7 & 67.1\\
                 D7P& 56.1 & 53.2 & 56.6 & 55.7\\
                 SON& 86.1 & 88.7 & 93.7 & 93.2\\
                 \textbf{MSK}& 58.3 & 60.1 & 63.2 & 61.9\\
                 \textbf{UDA}& 57.7 & 57.1 & 59.3 & 61.0\\
                 PH2& 80.0 & 81.7 & 83.3 & 85.0\\
                 \hline
                 Avg& 67.7 & 68.2 & 70.6 & 70.7\\

                \hline
                
                \hline
                \end{tabular}
                
    \end{minipage}
    \hspace{1mm}
    \begin{minipage}[l]{0.3\textwidth}
        \centering
        \setlength{\tabcolsep}{1.5mm}
        \setlength{\abovecaptionskip}{0.2cm}
        \setlength{\belowcaptionskip}{0.2cm}
        \footnotesize
            \makeatletter\def\@captype{table}\makeatother

            \caption*{(e) Meta-training: \textbf{SON, PH2}}
            
                \begin{tabular}{c c c c c}
                \hline
                
                \hline
                 & fe & dq & f\&d & full \\
                 \hline
                 
                 \hline
                 HAM& 67.3 & 68.6 & 67.5 & 67.5\\
                 D7P& 55.5 & 54.4 & 55.5 & 56.1\\
                 \textbf{SON}& 93.1 & 86.1 & 96.0 & 97.5\\
                 MSK& 60.1 & 65.1 & 61.9 & 62.1\\
                 UDA& 54.9 & 52.2 & 58.8 & 59.3\\
                 \textbf{PH2}& 76.7 & 78.3 & 80.0 & 86.7\\
                 \hline
                 Avg& 67.9 & 67.5 & 70.0 & 71.5\\
                
                \hline
                
                \hline
                \end{tabular}
                
    \end{minipage}
    \hspace{1mm}
    \begin{minipage}[l]{0.3\textwidth}
        \centering
        \setlength{\tabcolsep}{1.5mm}
        \setlength{\abovecaptionskip}{0.2cm}
        \setlength{\belowcaptionskip}{0.2cm}
        \footnotesize
            \makeatletter\def\@captype{table}\makeatother

            \caption*{(f) Meta-training: \textbf{D7P, SON}}
            
                \begin{tabular}{c c c c c}
                \hline
                
                \hline
                 & fe & dq & f\&d & full \\
                 \hline
                 
                 \hline
                 HAM& 67.5 & 68.6 & 68.4 & 68.1\\
                 \textbf{D7P}& 55.4 & 55.9 & 58.6 & 60.1\\
                 \textbf{SON}& 81.5 & 86.9 & 91.3 & 94.5\\
                 MSK& 61.2 & 59.5 & 65.8 & 66.2\\
                 UDA& 56.6 & 61.0 & 61.0 & 59.3\\
                 PH2& 81.7 & 85.0 & 85.0 & 86.7\\
                 \hline
                 Avg& 67.3 & 69.5 & 71.7 & 72.5\\
                
                \hline
                
                \hline
                \end{tabular}
                
    \end{minipage}
\end{minipage}

\begin{minipage}{\textwidth}
    \begin{minipage}[l]{0.3\textwidth}
        \centering
        \setlength{\tabcolsep}{1.5mm}
        \setlength{\abovecaptionskip}{0.2cm}
        \setlength{\belowcaptionskip}{0.2cm}
        \footnotesize
            \makeatletter\def\@captype{table}\makeatother

            \caption*{(g) Meta-training: \textbf{D7P, PH2}}
            
                \begin{tabular}{c c c c c}
                \hline
                
                \hline
                 & fe & dq & f\&d & full \\
                 \hline
                 
                 \hline
                 HAM& 65.7 & 66.3 & 67.0 & 67.5\\
                 \textbf{D7P}& 52.3 & 53.4 & 55.9 & 58.6\\
                 SON& 81.5 & 76.8 & 86.3 & 84.0\\
                 MSK& 52.0 & 48.8 & 56.6 & 55.8\\
                 UDA& 53.9 & 49.5 & 56.0 & 56.0\\
                 \textbf{PH2}& 76.7 & 78.3 & 81.7 & 85.0\\
                 \hline
                 Avg& 63.7 & 62.2 & 67.3 & 67.8\\

                \hline
                
                \hline
                \end{tabular}
                
    \end{minipage}
    \hspace{1mm}
    \begin{minipage}[l]{0.3\textwidth}
        \centering
        \setlength{\tabcolsep}{1.5mm}
        \setlength{\abovecaptionskip}{0.2cm}
        \setlength{\belowcaptionskip}{0.2cm}
        \footnotesize
            \makeatletter\def\@captype{table}\makeatother

            \caption*{(h) Meta-training: \textbf{SON, MSK}}
            
                \begin{tabular}{c c c c c}
                \hline
                
                \hline
                 & fe & dq & f\&d & full \\
                 \hline
                 
                 \hline
                 HAM& 68.1 & 68.3 & 67.9 & 67.7\\
                 D7P& 53.2 & 50.5 & 55.5 & 56.1\\
                \textbf{SON}& 86.1 & 88.8 & 86.9 & 88.8\\
                 \textbf{MSK}& 55.2 & 56.6 & 59.5 & 60.1\\
                 UDA& 54.9 & 56.0 & 56.6 & 59.3\\
                 PH2& 75.0 & 78.3 & 80.0 & 85.0\\
                 \hline
                 Avg& 65.4 & 66.4 & 67.7 & 69.5\\
                
                \hline
                
                \hline
                \end{tabular}
                            
    \end{minipage}
    \hspace{1mm}
    \begin{minipage}[l]{0.3\textwidth}
        \centering
        \setlength{\tabcolsep}{1.5mm}
        \setlength{\abovecaptionskip}{0.2cm}
        \setlength{\belowcaptionskip}{0.2cm}
        \footnotesize
            \makeatletter\def\@captype{table}\makeatother
            \caption*{(i) Meta-training: \textbf{SON, UDA}}
            \begin{tabular}{c c c c c}
                \hline
                
                \hline
                 & fe & dq & f\&d & full \\
                 \hline
                 
                 \hline
                 HAM& 68.4 & 68.5 & 67.4 & 67.7\\
                 \textbf{D7P}& 55.9 & 54.6 & 54.8 & 59.2\\
                 \textbf{SON}& 89.3 & 88.7 & 91.3 & 96.0\\
                 MSK& 56.6 & 55.8 & 55.2 & 59.5\\
                 UDA& 52.2 & 57.7 & 55.5 & 56.6\\
                 PH2& 81.7 & 83.3 & 91.7 & 85.0\\
                 \hline
                 Avg& 67.4 & 68.1 & 69.3 & 70.7\\
                
                \hline
                
                \hline
                \end{tabular}
                
    \end{minipage}

\end{minipage}
\vspace{5pt}

The bar plot in Figure~\ref{fig10} shows the accuracy of CLKM(fe), CLKM(f\&d), and CLKM(full) on different meta-testing datasets.
As the performance of the model on target domains improves, the performance on the source domain experiences slight degradation.
However, we think this degeneration on source is insignificant compared to the promotion on targets.
Compared with CLKM(fe), when the domain-quantizer is added, the performance of the model on meta-testing datasets almost improves, meanwhile, the performance on meta-training datasets is maintained.
Finally, the self-adaptive kernel helps to measure the gap across the source and target lesion domains accurately, making the learned knowledge transfer better, and thus the overall accuracy increases.
For the SON domain, the accuracy of the model can easily reach a high level on it due to the domain containing only one diagnostic category.
However, the learned self-adaptive kernel still improves the performance, as illustrated in Figure~\ref{fig10}(c) and Figure~\ref{fig10}(d).

We briefly analyzed the impact of data distribution on model performance, 
and the results show that discrete adaptation methods are not suitable for CLDA.
In general, these methods slightly improve performance compared with no adaptation.
However, the traditional discrete adaptive method sometimes will degrade the performance of the model, as shown in Table 2.
Besides, EWC, the traditional way of overcoming catastrophic forgetting, sometimes also hurts the performance of the model.
In contrast to the above methods extremely, the method CLKM our proposed harness the consecutive lesion domains accurately. 
Moreover, for alleviating catastrophic forgetting, it is a more flexible and general approach to let the model itself learn what to retain during adaptation.

\textbf{Meta-representation visualization.}
To further analyze the transfer of knowledge in representation space, we visualize it learned by different methods via $t$-SNE in Figure~\ref{fig:fig6}, Figure~\ref{fig:fig7}, and Figure~\ref{fig:fig8}.
Different colors and shapes scattered dots in the figure represent the data sampled from diverse lesion domains. 
As shown in these $t$-SNE figures, CLKM can align the source and multiple target domain in feature space properly.
Although CLKM can't align the representation feature well when we select D7P, UDA, and SONIC as meta-training datasets, its performance still surpasses the other two methods.
Figure~\ref{fig:fig7} and Figure~\ref{fig:fig8} validate the effectiveness of CLKM.
\begin{figure}[!t]
\centering
\subfigcapskip = -8pt
\subfigure[meta-training: UDA, PH2]{
\begin{minipage}[t]{0.33\linewidth}
\centering
\includegraphics[width=2.35in]{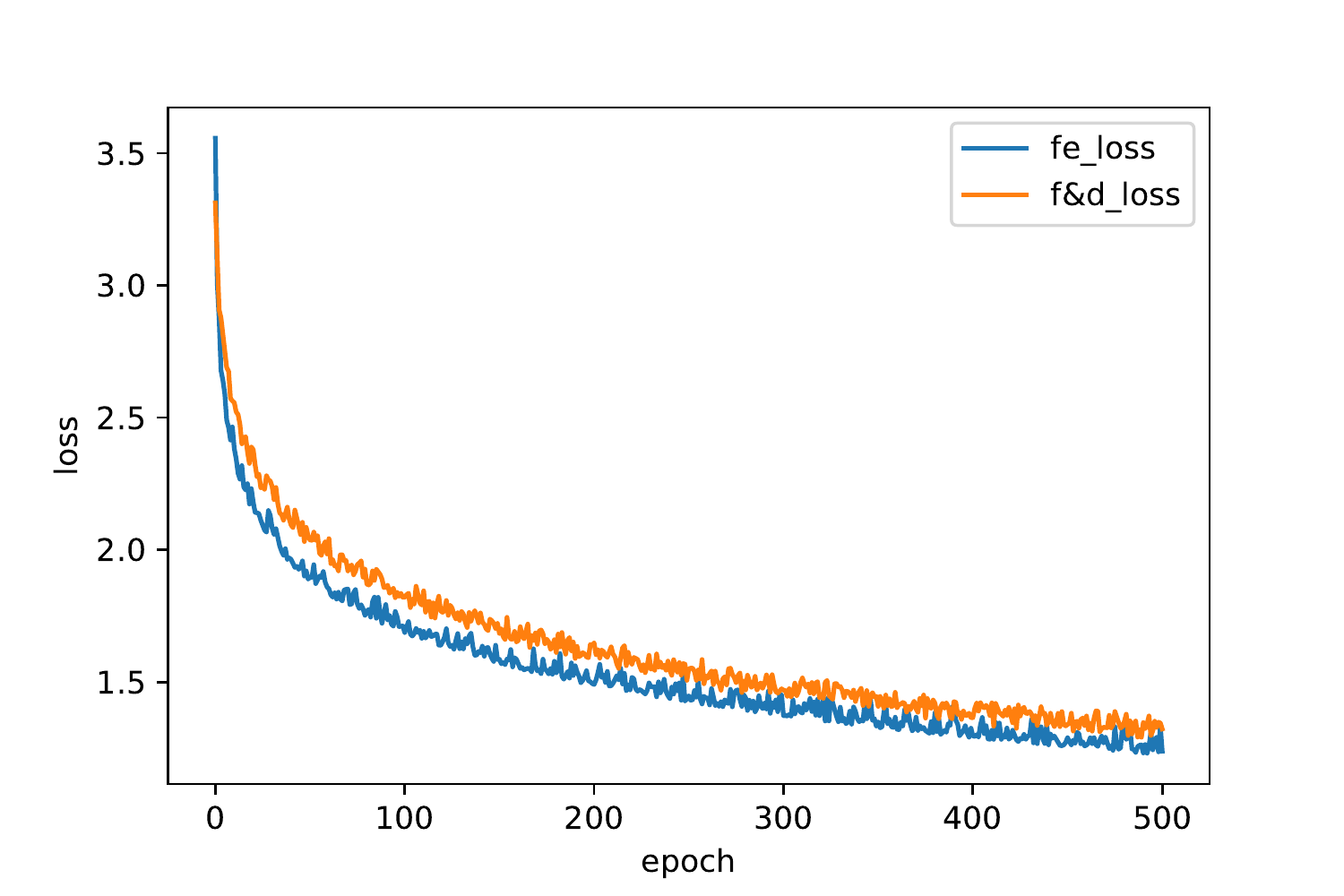}
\end{minipage}%
}%
\subfigure[meta-training: MSK, PH2]{
\begin{minipage}[t]{0.33\linewidth}
\centering
\includegraphics[width=2.35in]{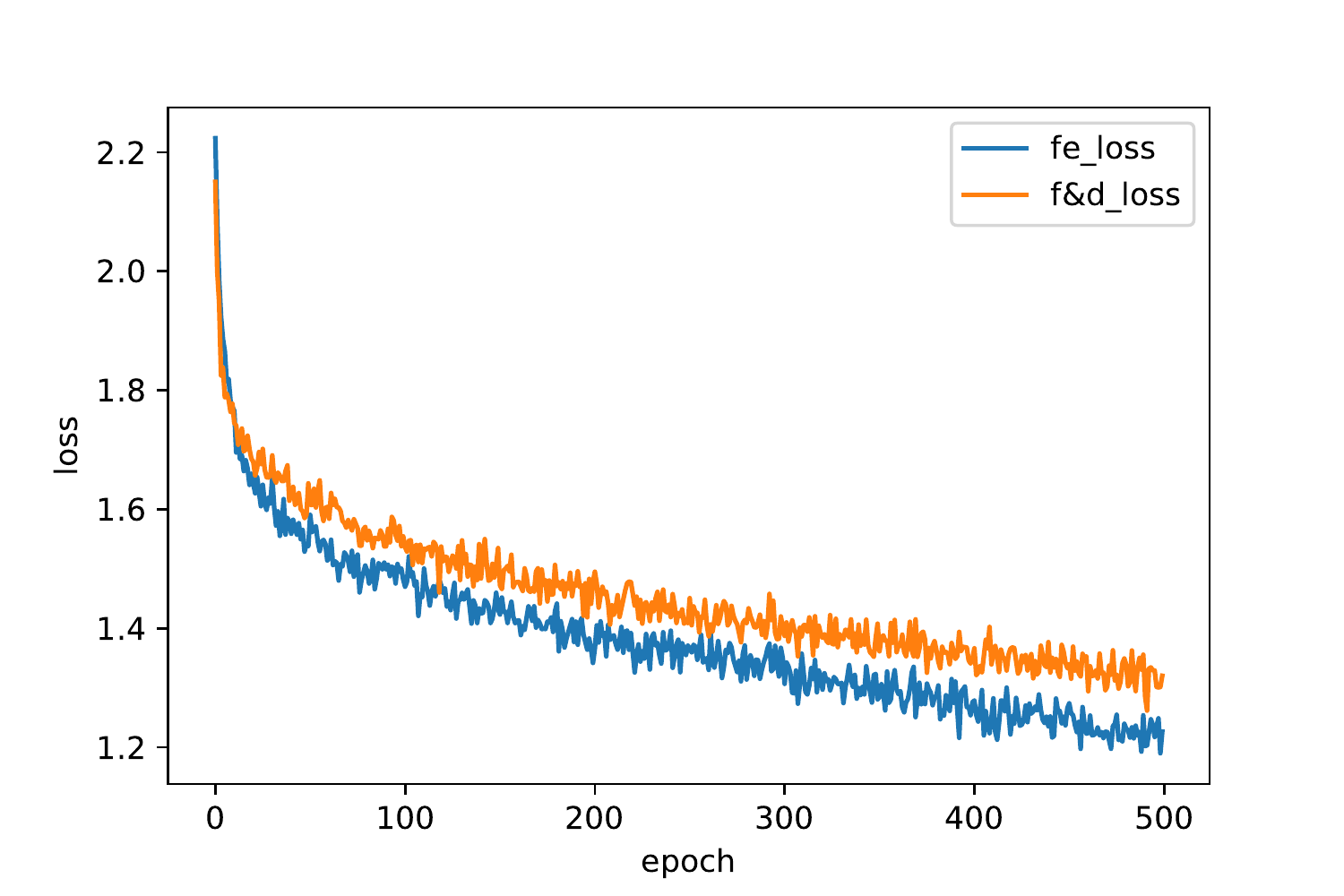}
\end{minipage}%
}%
\subfigure[meta-training: D7P, MSK]{
\begin{minipage}[t]{0.33\linewidth}
\centering
\includegraphics[width=2.35in]{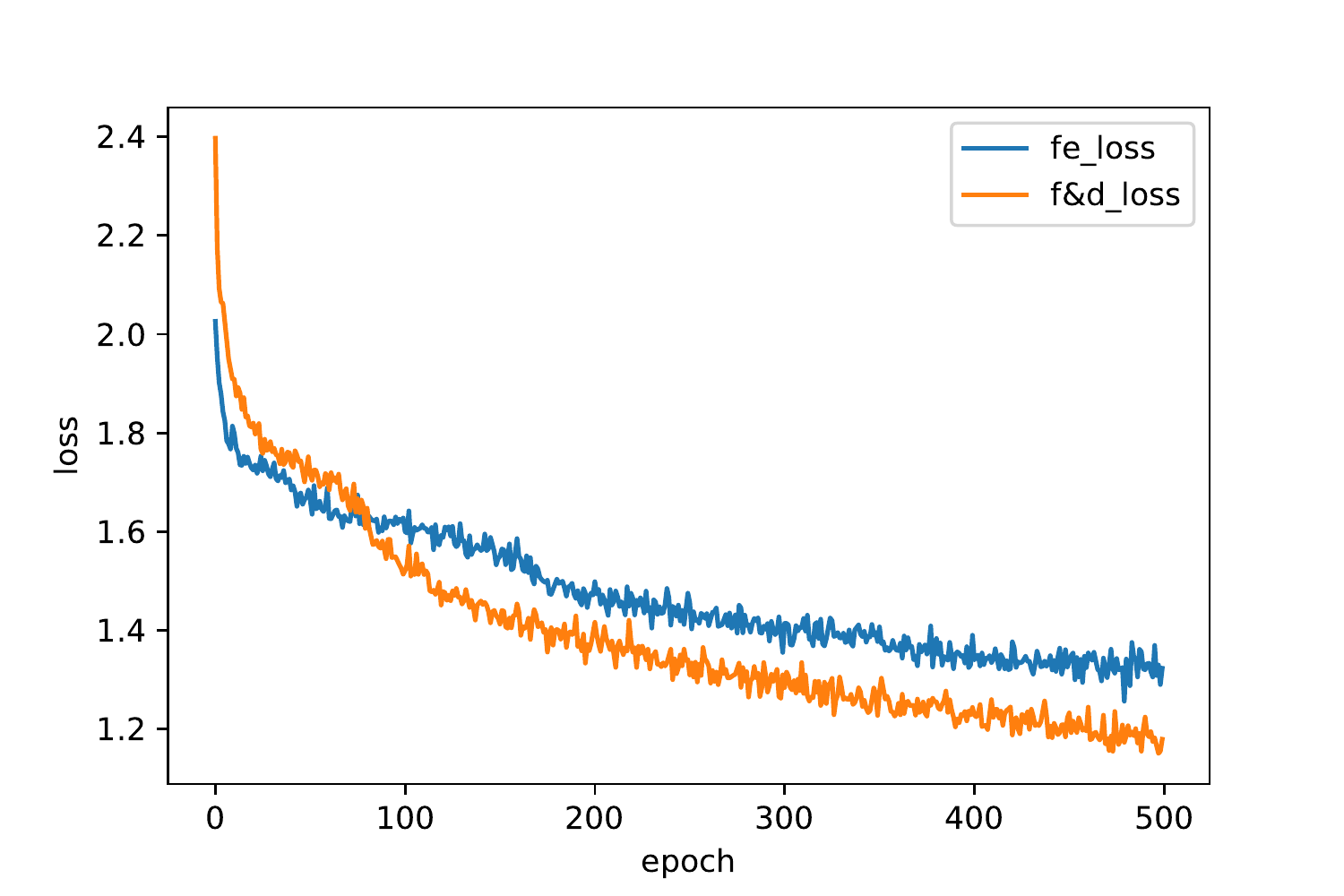}
\end{minipage}
}%
\vspace{-0.65cm}
\subfigure[meta-training: MSK, UDA]{
\begin{minipage}[t]{0.33\linewidth}
\centering
\includegraphics[width=2.35in]{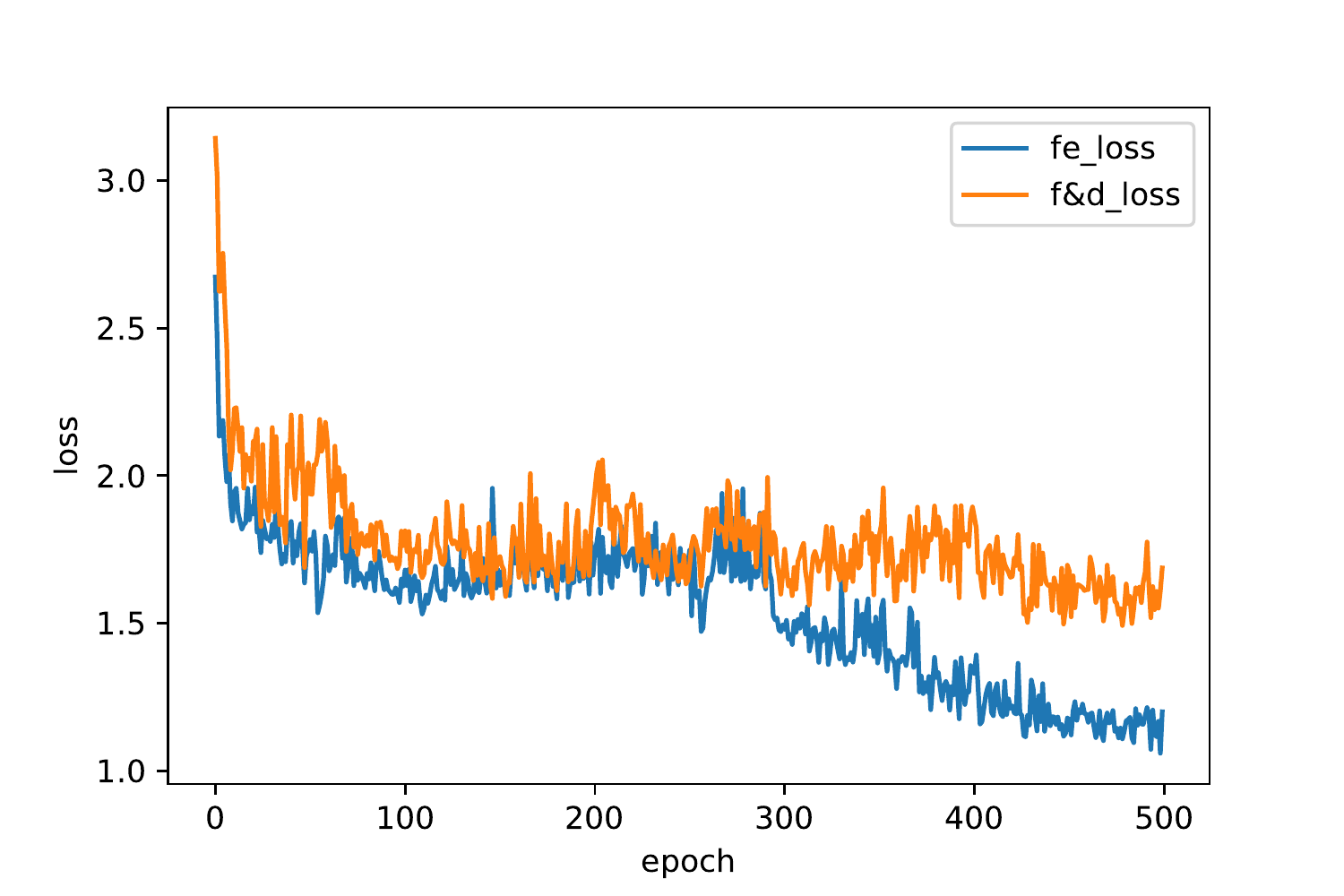}
\end{minipage}%
}%
\subfigure[meta-training: SON, PH2]{
\begin{minipage}[t]{0.33\linewidth}
\centering
\includegraphics[width=2.35in]{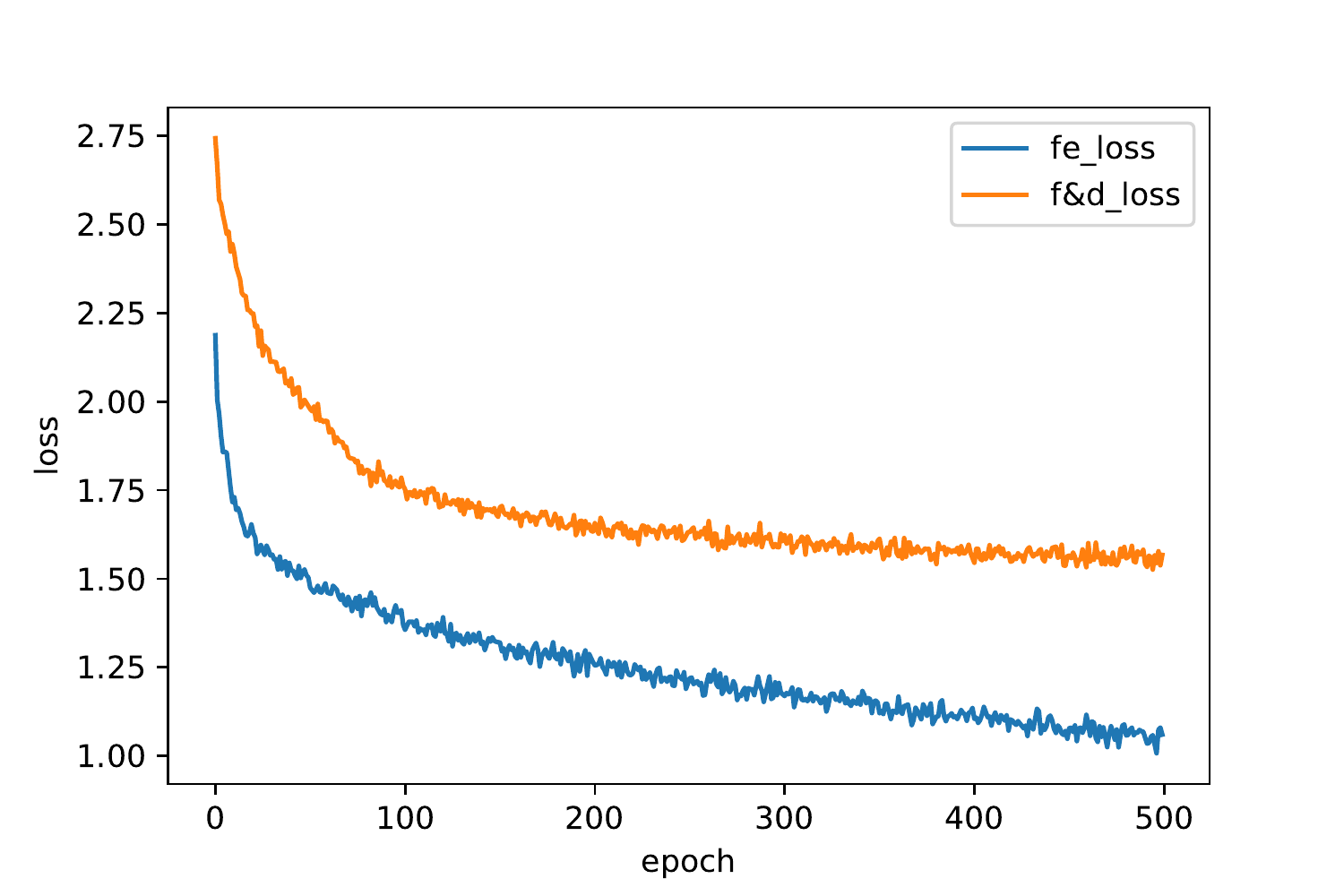}
\end{minipage}%
}%
\subfigure[meta-training: D7P, SON]{
\begin{minipage}[t]{0.33\linewidth}
\centering
\includegraphics[width=2.35in]{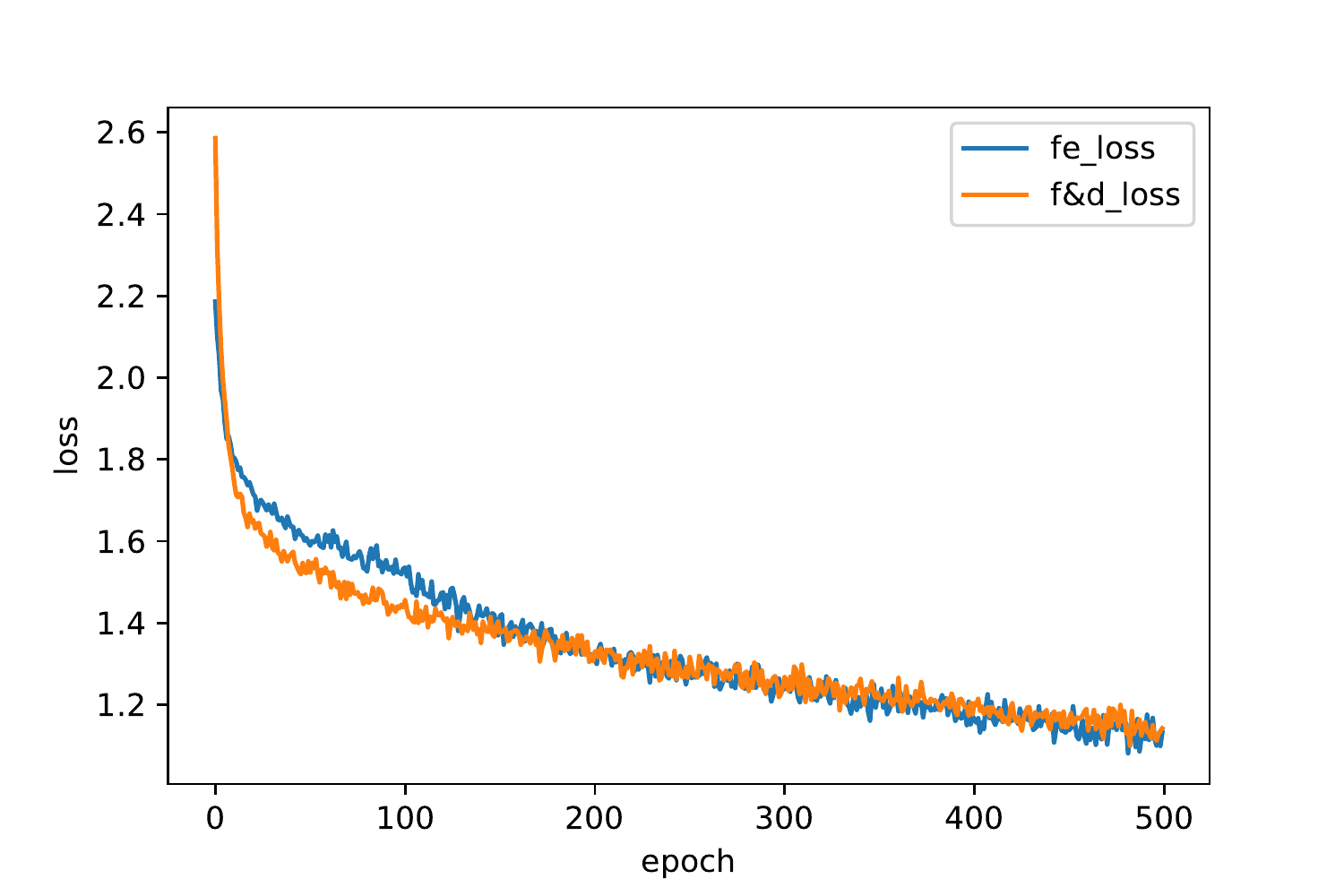}
\end{minipage}
}%
\vspace{-0.65cm}
\subfigure[meta-training: D7P, PH2]{
\begin{minipage}[t]{0.33\linewidth}
\centering
\includegraphics[width=2.35in]{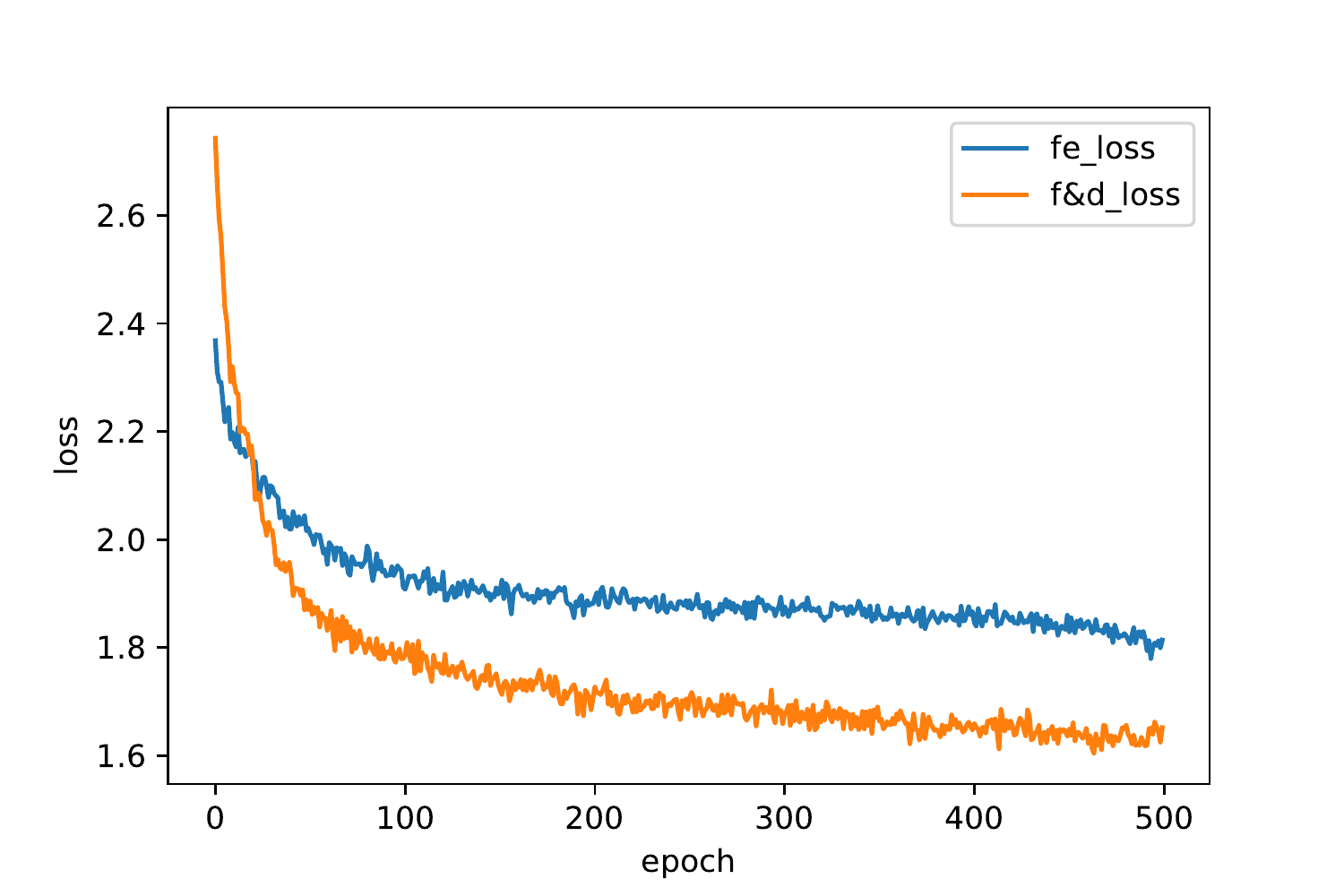}
\end{minipage}%
}%
\subfigure[meta-training: SON, MSK]{
\begin{minipage}[t]{0.33\linewidth}
\centering
\includegraphics[width=2.35in]{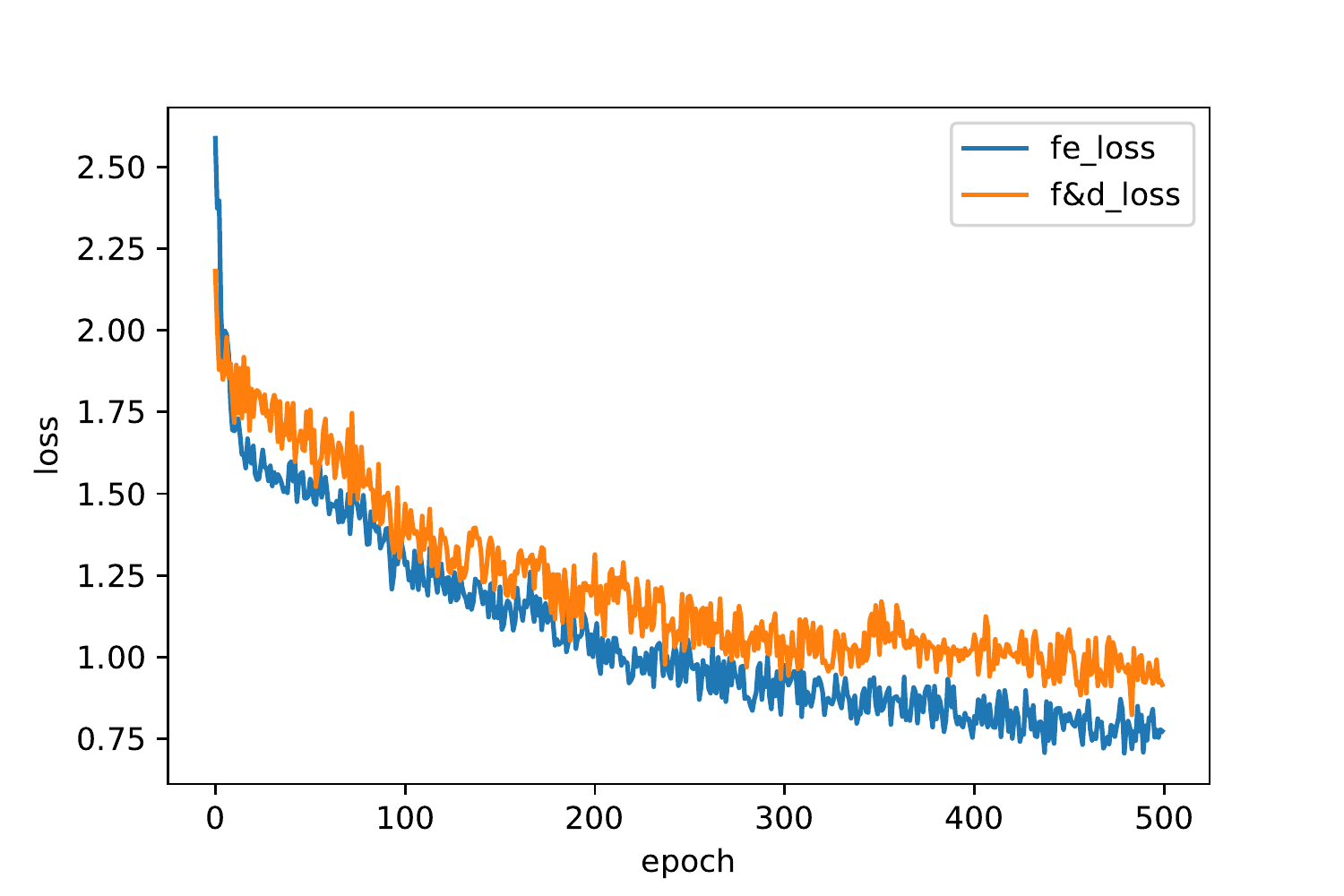}
\end{minipage}%
}%
\subfigure[meta-training: SON, UDA]{
\begin{minipage}[t]{0.33\linewidth}
\centering
\includegraphics[width=2.35in]{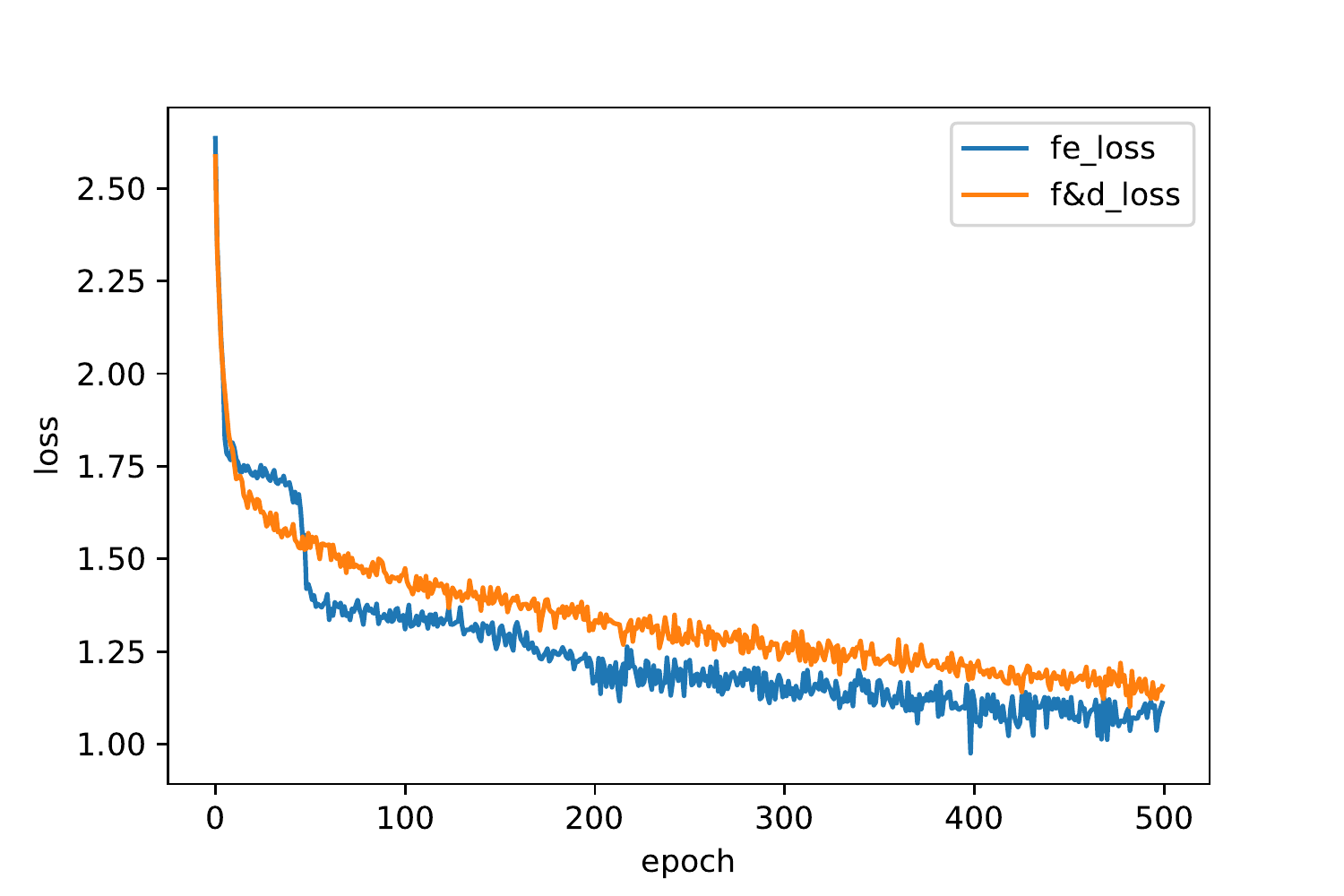}
\end{minipage}
}%
\setlength{\abovecaptionskip}{-0.5pt}
\caption{The training loss curves on different meta-training datasets.}
\label{fig:fig12}
\end{figure}

\textbf{Ablation Study.} We first performed CLKM with only a feature-extractor (\emph{i.e.}, CLKM(fe)), showing the performance without overcoming catastrophic forgetting. 
Then we added a domain-quantizer (\emph{i.e.}, CLKM(f\&d)) into our model, and the result tells us that the model accuracy on perviously learned datasets is maintained.
In Figure~\ref{fig:fig9}, we plotted the training loss curves for six meta-testing datasets scenarios. 
As shown in the figure, the loss increased slightly when the anti-forgetting term added.
Finally, we add the learned self-adaptive kernel (\emph{i.e.}, CLKM(full)), which further help mitigate the gap across the source and target lesion domains in the feature space. 
Detailed results are shown in the Table~\ref{tab1} and Table~\ref{tab2}-Table~\ref{tab7}.

\textbf{Balance adaptation and anti-forgetting.}
Review the algorithm~\ref{alg1}, the loss function Eq.(\ref{eq16}) in SAP, including a penalty term to against catastrophic forgetting, $\lambda \mathcal{L}_m$.
The hyper-parameter $\lambda$ plays a moderating role between anti-forgetting and domain adaptation.
Intuitively, the larger $\lambda$, the greater penalty for the forgetting. 
To explore the effects of different levels of $\lambda$, we set four levels from 0.2 to 0.8 and visualized the result in Figure~\ref{fig11}.
Without loss of generality, we choose D7P and MSK, D7P and UDA, MSK and SON as meta-testing datasets to build three cases respectively.
As shown in Figure~\ref{fig11}, the CLKM is not sensitive to the coefficient of the penalty term. 
It means that, for CLKM, the balance between alleviating forgetting and domain adaptation isn't hard to achieve.

\begin{figure}[!t]
\centering
\subfigcapskip = -10pt
\subfigure[meta-training: UDA, PH2]{
\begin{minipage}[t]{0.32\linewidth}
\centering
\includegraphics[width=2.35in]{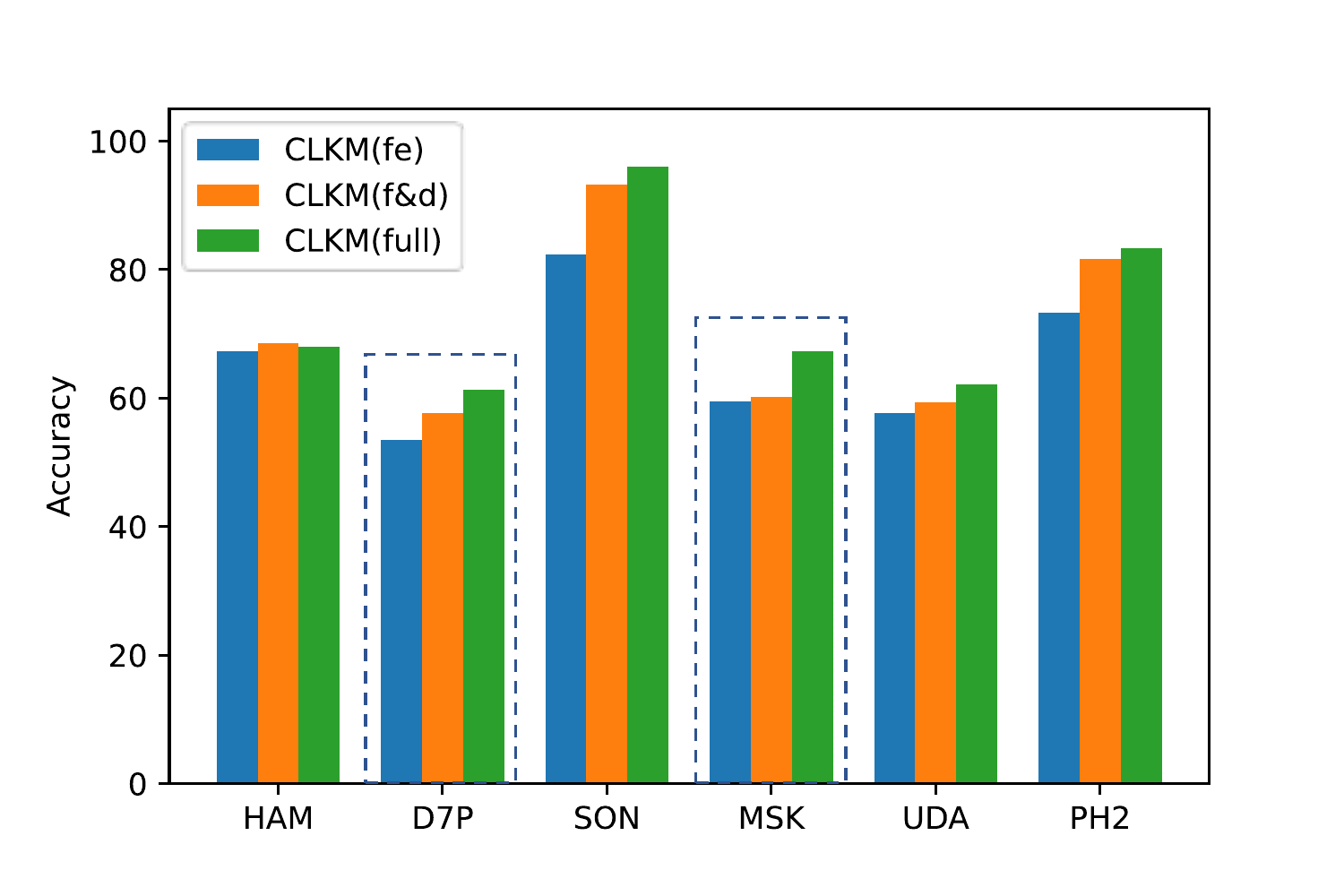}
\end{minipage}%
}
\subfigure[meta-training: MSK, PH2]{
\begin{minipage}[t]{0.32\linewidth}
\centering
\includegraphics[width=2.35in]{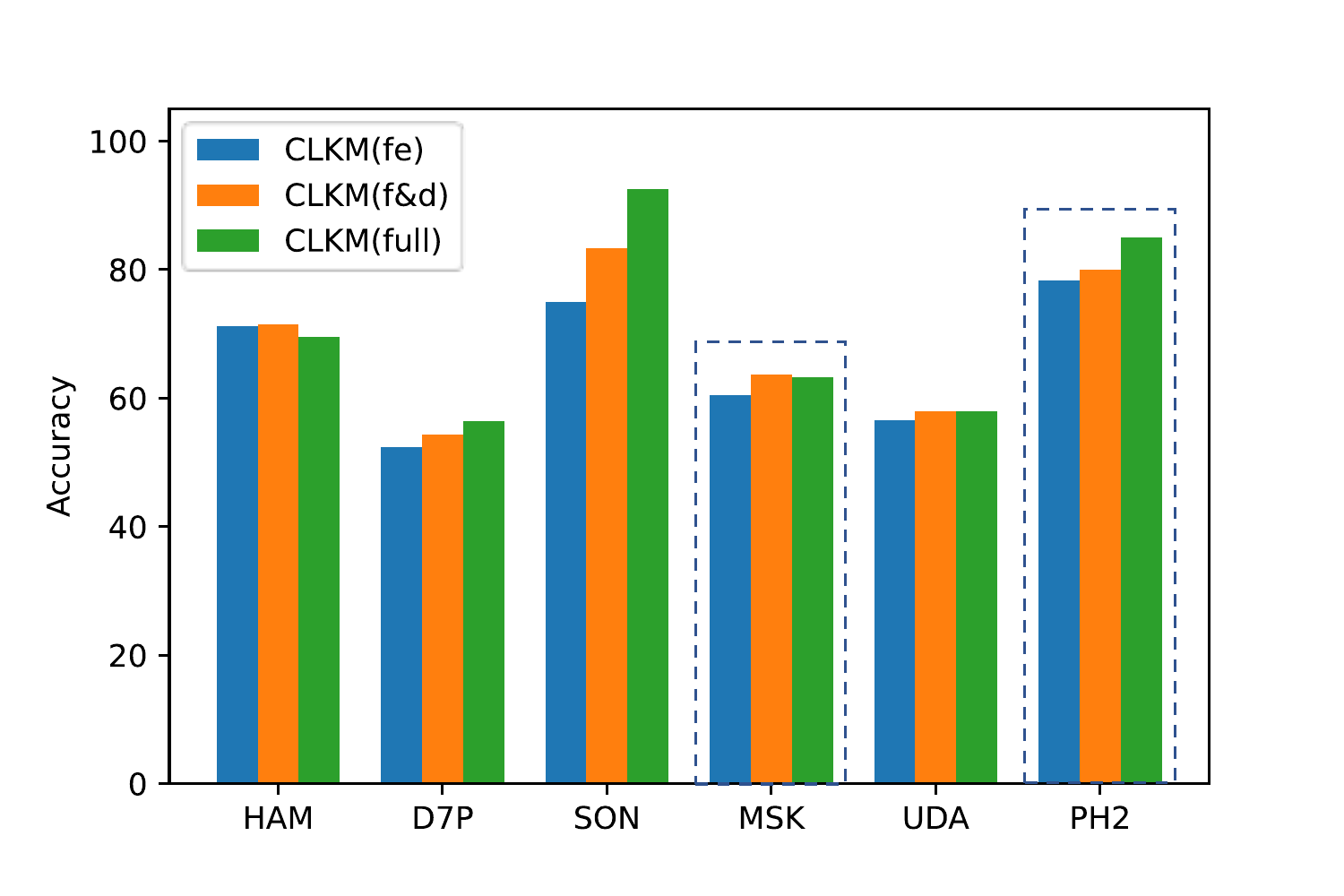}
\end{minipage}%
}
\subfigure[meta-training: D7P, MSK]{
\begin{minipage}[t]{0.32\linewidth}
\centering
\includegraphics[width=2.35in]{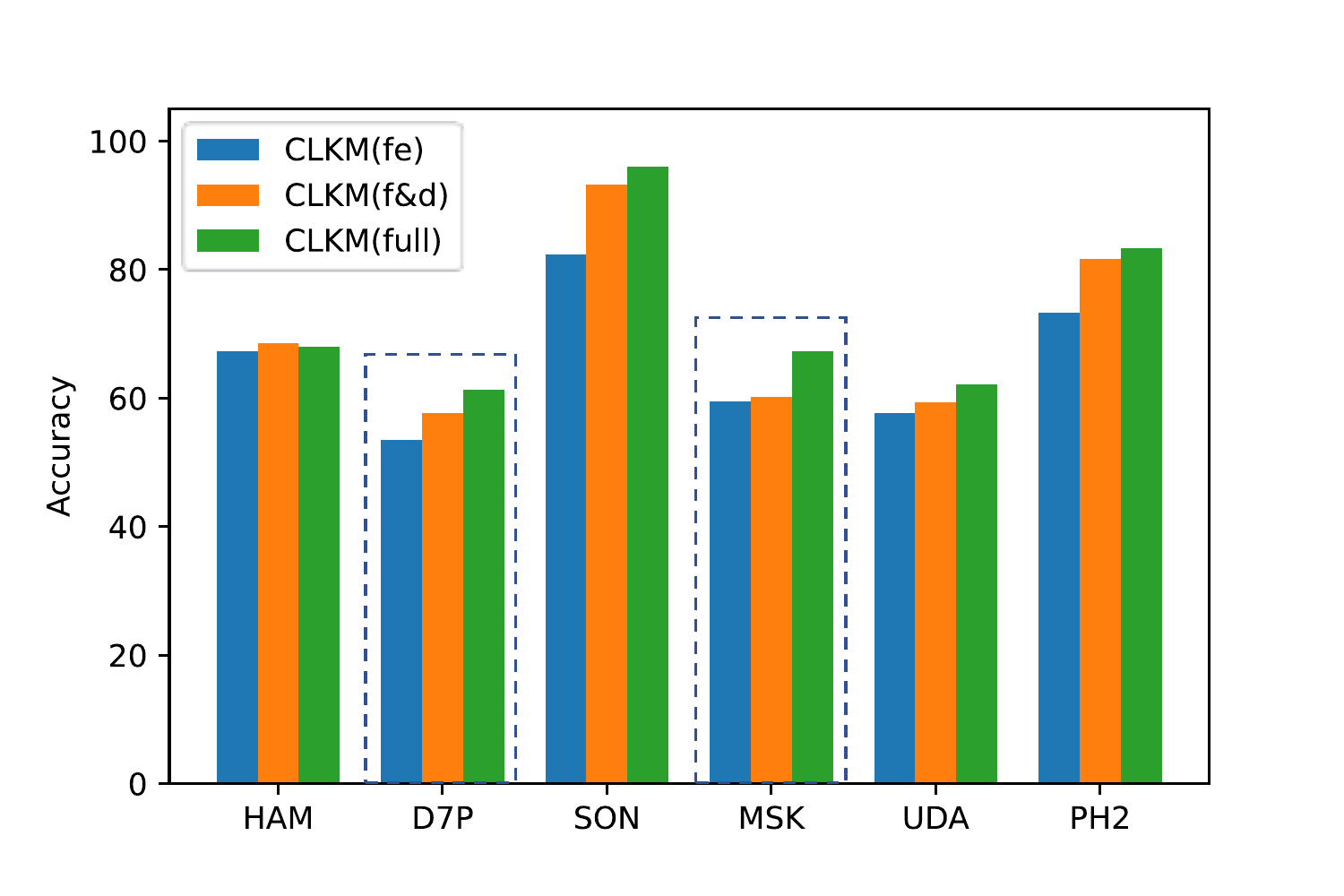}
\end{minipage}%
}%
\vspace{-0.2cm}
\subfigure[meta-training: MSK, UDA]{
\begin{minipage}[t]{0.32\linewidth}
\centering
\includegraphics[width=2.35in]{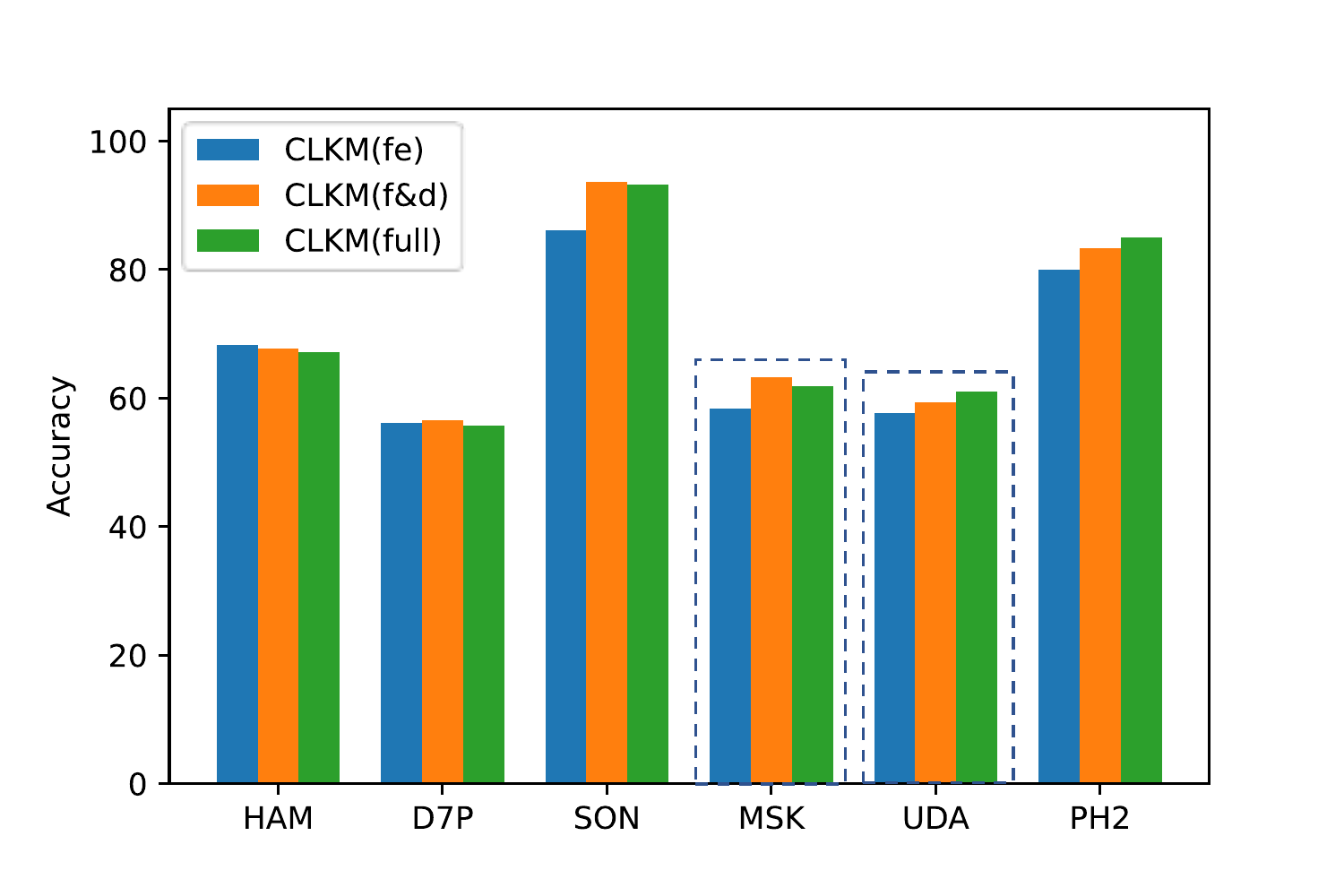}
\end{minipage}%
}
\subfigure[meta-testing: SON, PH2]{
\begin{minipage}[t]{0.32\linewidth}
\centering
\includegraphics[width=2.35in]{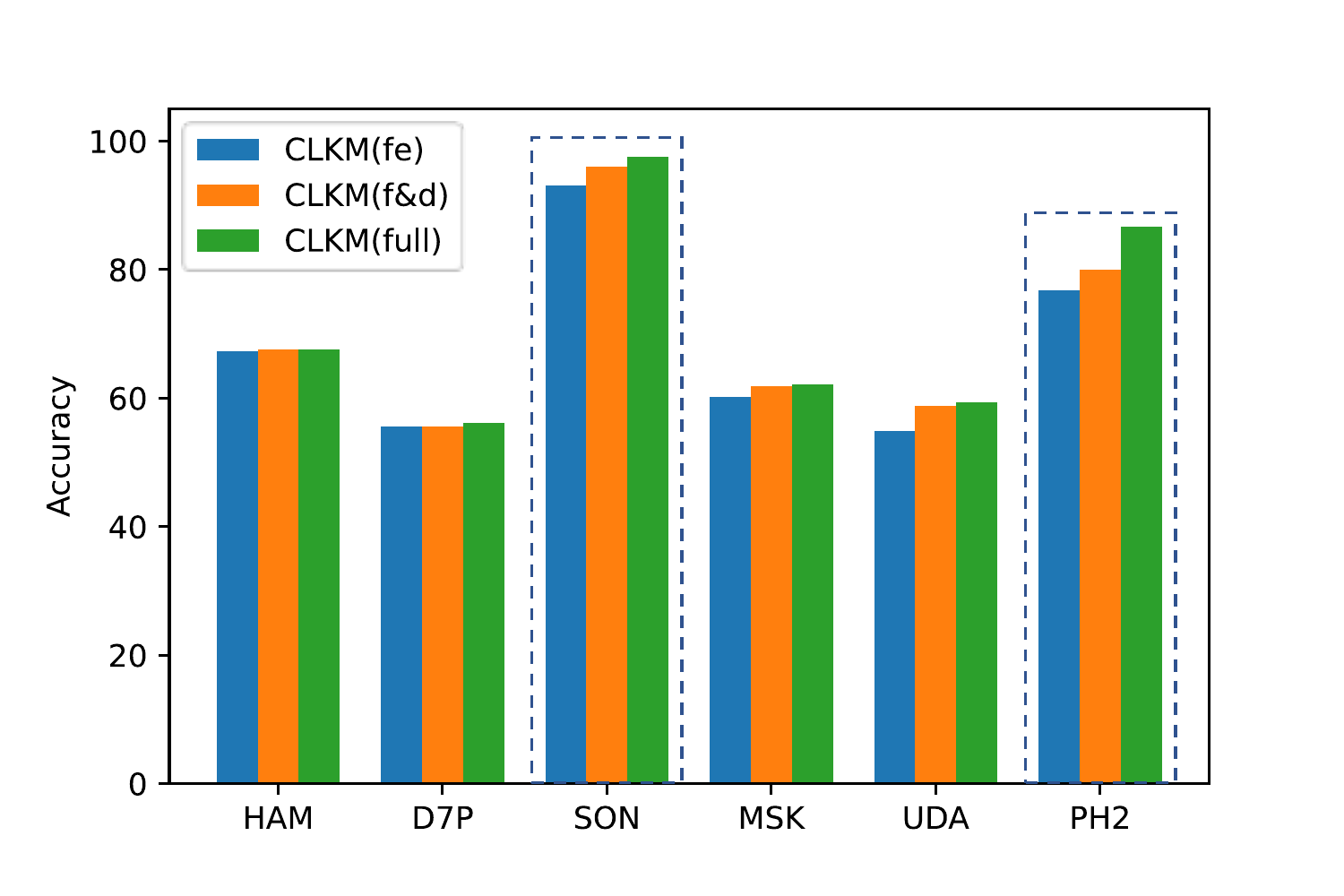}
\end{minipage}%
}
\subfigure[meta-training: D7P, SON]{
\begin{minipage}[t]{0.32\linewidth}
\centering
\includegraphics[width=2.35in]{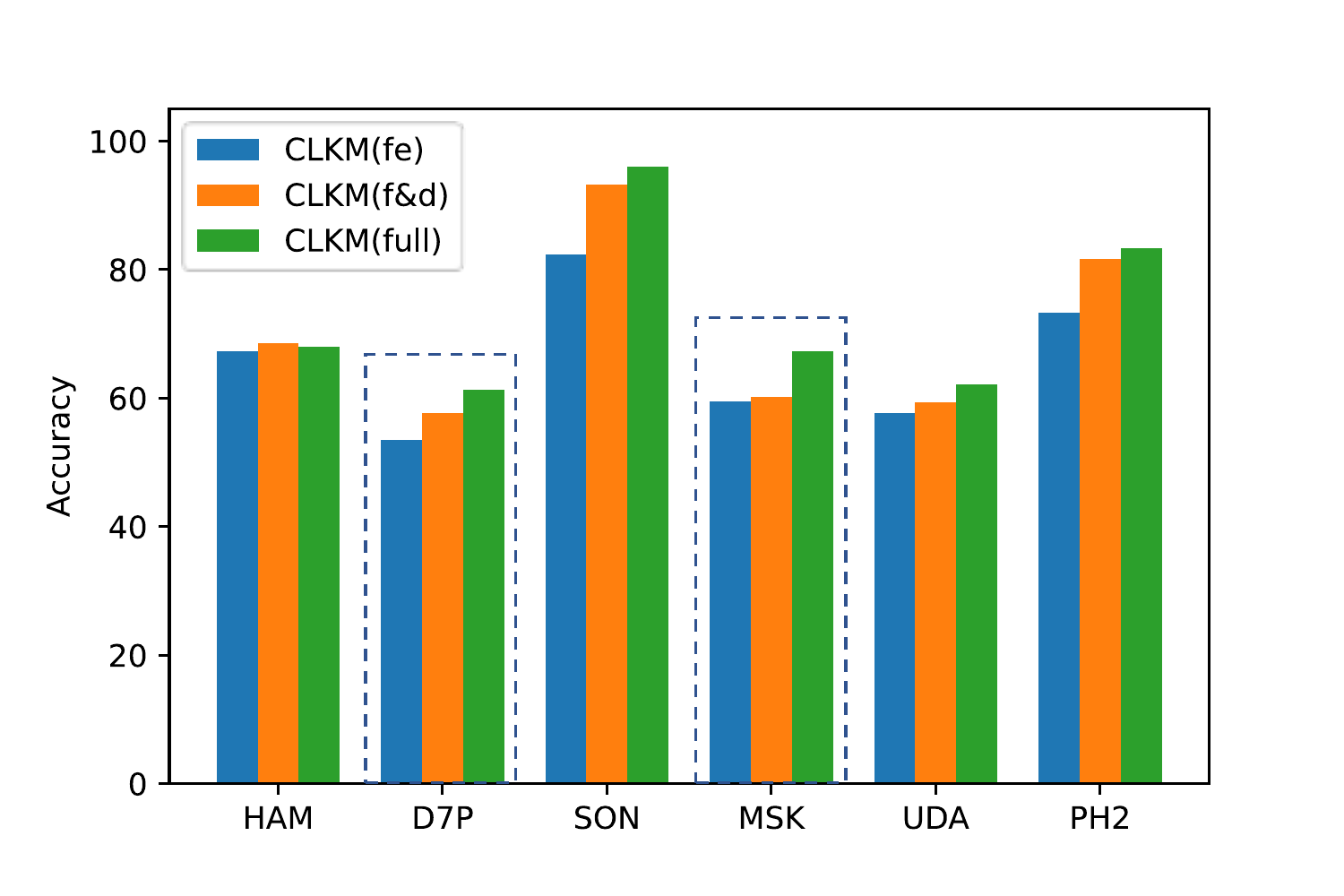}
\end{minipage}%
}%
\vspace{-0.2cm}
\subfigure[meta-training: D7P, PH2]{
\begin{minipage}[t]{0.32\linewidth}
\centering
\includegraphics[width=2.35in]{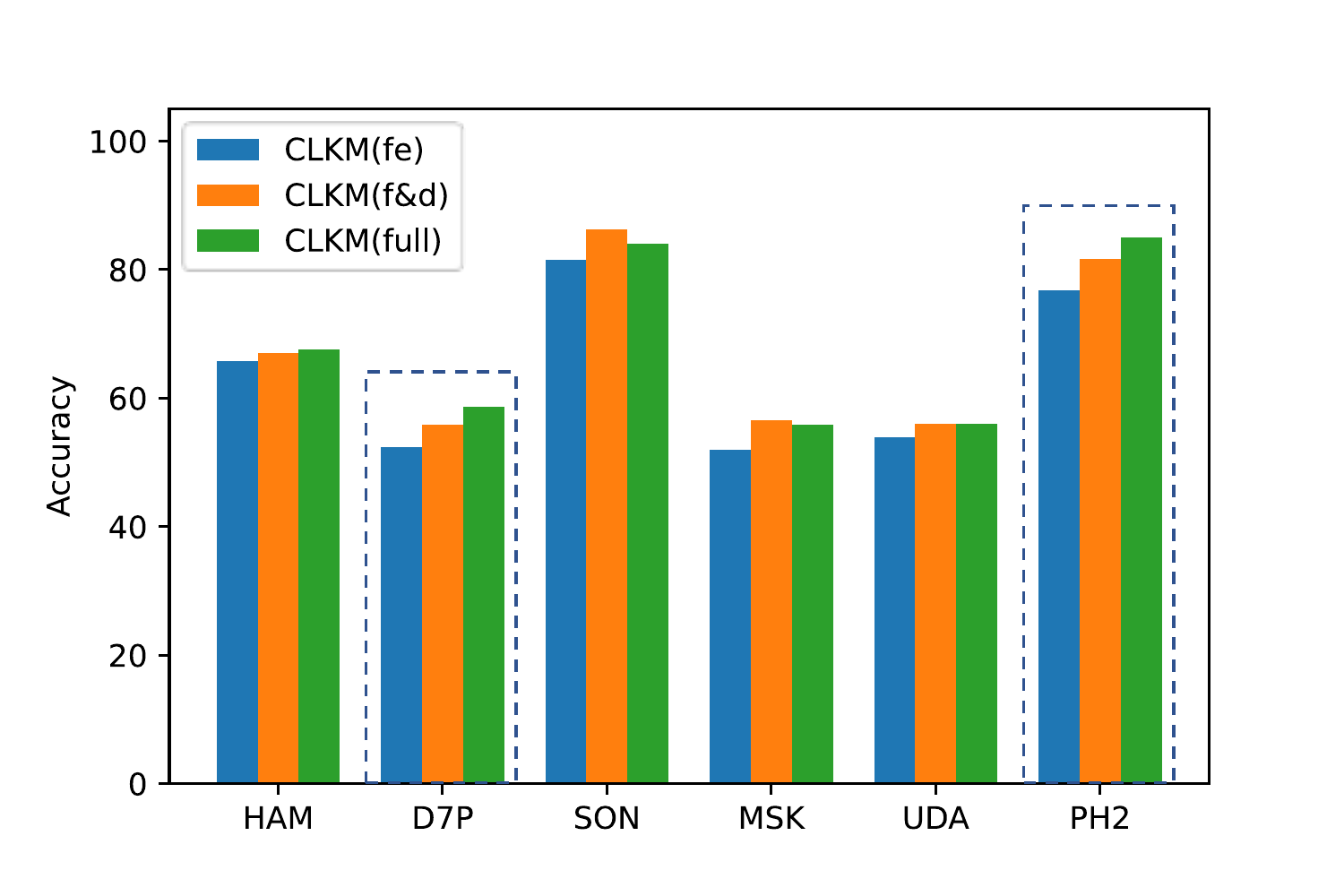}
\end{minipage}%
}
\subfigure[meta-testing: SON, MSK]{
\begin{minipage}[t]{0.32\linewidth}
\centering
\includegraphics[width=2.35in]{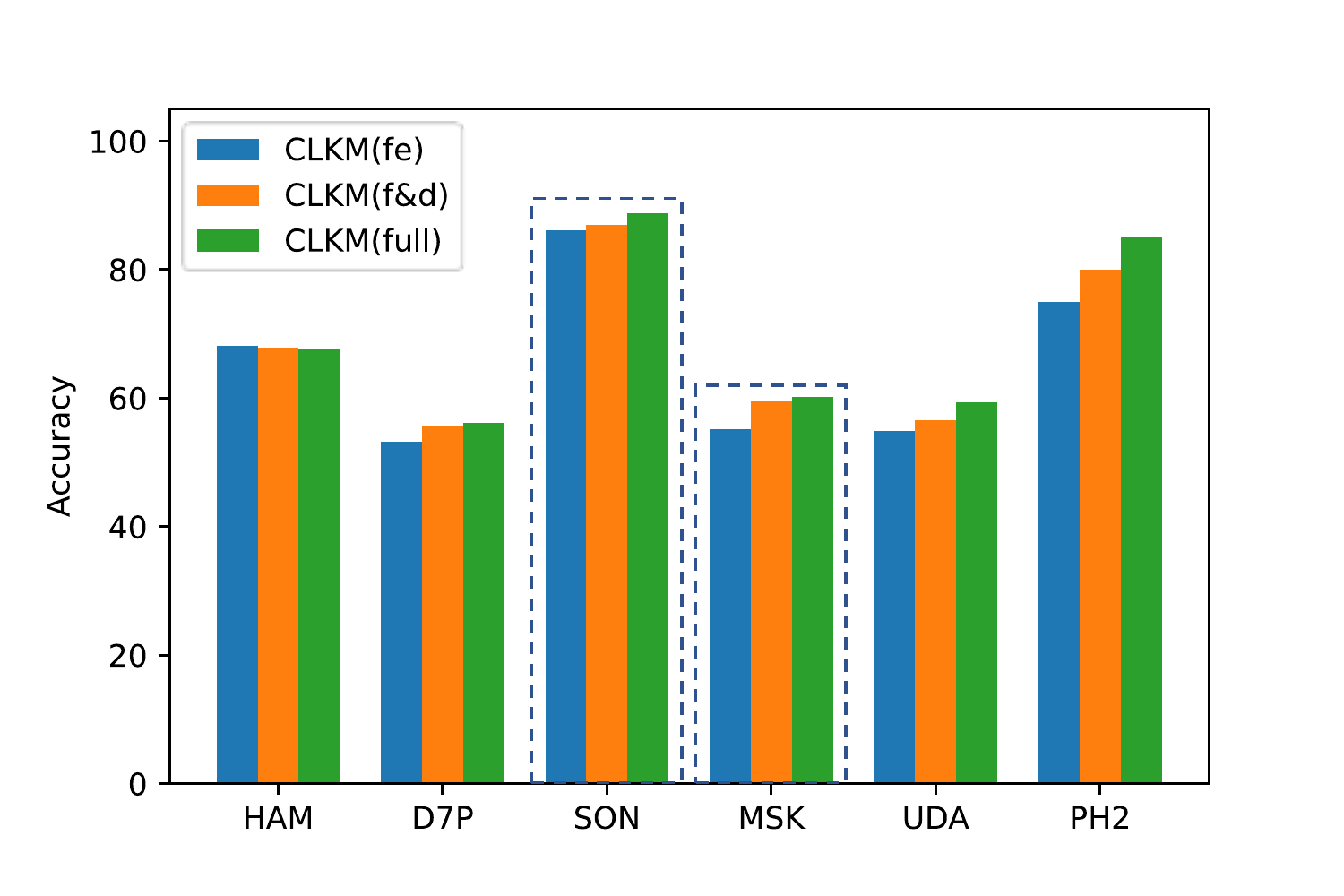}
\end{minipage}%
}
\subfigure[meta-training: SON, UDA]{
\begin{minipage}[t]{0.32\linewidth}
\centering
\includegraphics[width=2.35in]{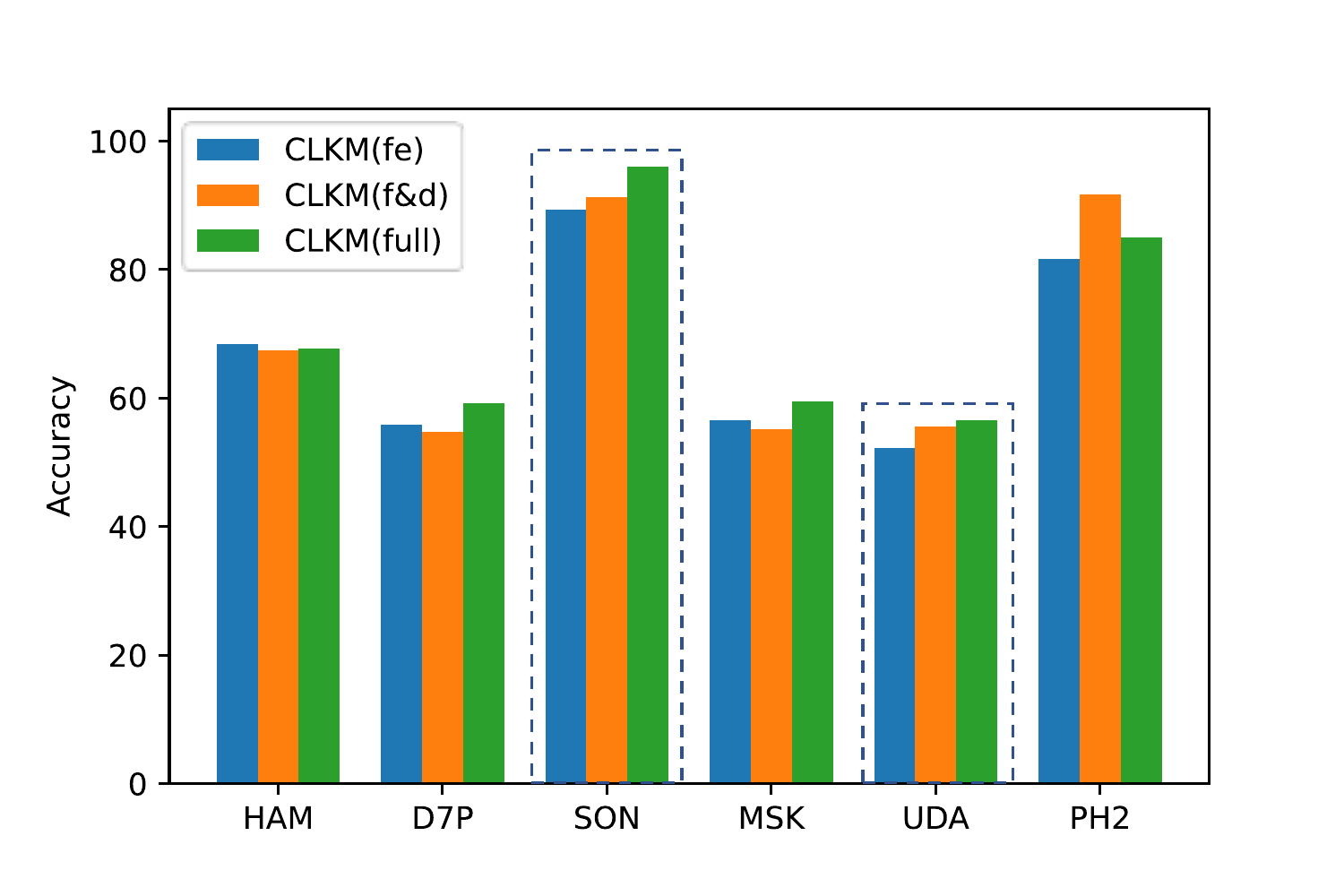}
\end{minipage}%
}%
\setlength{\abovecaptionskip}{-0.5pt}
\caption{The accuracy of CLKM(fe), CLKM(f\&e), and CLKM(full) on different meta-training datasets.}
\label{fig13}
\end{figure}

\textbf{Robustness analysis.}
In the above experiments, we adapt the diagnosis model from the source lesion domain to three target lesion domains in the meta-training phase, and we set two novel target lesion domains as meta-testing datasets to test the model performance.
In the clinical environment, the number of target lesion domains that can be accessed isn't static, and to further analyze CLKM robustness, we reduce the target lesion domain in the meta-training phase by one, and set three target lesion domains as meta-testing datasets.
We summarized the results in Table 8, and compared with the previous, the model performance slightly decrease.
Analogously, when we add the domain-quantizer (\emph{i.e.} f\&d), the model maintained the accuracy on previously learned datasets, and thus improve the overall performance.
Further, the learned self-adaptive kernel help to align the semantic feature across source and target lesion domains. 
We visualized the result in Figure~\ref{fig13}.
Besides, we plotted the training loss curves in Figure~\ref{fig:fig12}.
Not all the loss curves are raised by adding the penalty term, for example, when the D7P is selected as meta-training dataset, the loss curve of only feature-extractor (\emph{i.e.}, fe) always lightly higher than f\&d.
Compared to other lesion domains, the distribution of D7P is closest to the source lesion domain, which perhaps helps the model further harness the characteristics of the source lesion domain while preserving the utility knowledge of D7P.

\section{Conclusion}
We have proposed a meta-adaptation framework named Consecutive lesion Adaptive Meta-Learning (CLKM) to continually adapt to new unlabeled target lesion domains without catastrophic forgetting.
Specifically, CLKM mainly consists of two phases: Semantic Adaptation Phase (SAP) and Representation Adaptation Phase (RAP).
We introduce self-adaptive kernel in SAP to flexibly measure discrepancy across the source lesion distribution and consecutive target lesion distributions.
Besides, the penalty term from the domain-quantizer is added to the loss function for alleviating catastrophic forgetting during semantic adaptation.
In the RAP, the feature-extractor is optimized to align the transferable learned representation features across the source lesion domain and multiple target lesion domains.
Meanwhile, the domain-quantizer is also optimized to learn reserve more important learned knowledge from the previously learned target lesion domains.
The effectiveness of our CLKM model is justified well via the experiments.
In the future, we plan to extend our method into other unsupervised domain adaptive visual tasks,
such as objection detection and semantic segmentation.

\appendix 
\section{Theoretical analysis}\label{appendix}

We begin by reviewing the maximum mean discrepancy and its use in two-sample tests.
Let $\delta_{ak}$ be the self-adaptive kernel, and assume that kernel is measurable and bounded.
Then the MMD between two distribution $\mathcal{P}^s$ and $\mathcal{P}^t$ is:
\begin{equation}
    d^2(\mathcal{P}^s, \mathcal{P}^t, \delta_{ak}) = \mathbb{E}_{x^s, x^{s'}}[\delta_{ak}(x^s, x^{s'};\Theta_F)] + \mathbb{E}_{x^t,x^{t'}}[\delta_{ak}(x^t, x^{t'};\Theta_F)]-2\mathbb{E}_{x^s,x^t}[\delta_{ak}(x^s,x^t;\Theta_F)],
    \nonumber
\end{equation}
where $x^s, x^{s'}\stackrel{i.i.d}{\sim}\mathcal{P}^s$, and $x^t, x^{t'}\stackrel{i.i.d}{\sim}\mathcal{P}^t$.

Given the sample sets $\mathcal{D}^s = \{\textbf{x}^s_i, \textbf{y}^s_i\}_{i=1}^{N^s} \subset \mathcal{P}^s$ and $\mathcal{D}^t = \{\textbf{x}^t_j\}_{j=1}^{N^t} \subset \mathcal{P}^t$, and suppose $n = N^s = N^t$. The unbiased estimator of $d^2$ is:
\begin{equation}
\begin{split}
    \hat{d}_u^2(\mathcal{D}^s, \mathcal{D}^t, \delta_{ak}) = & \frac{1}{N^s (N^s - 1)} \sum_{i = 1}^{N^s} \sum_{j \neq i}^{N^s}\delta_{ak}(\textbf{x}^s_i, \textbf{x}^s_j;\Theta_F) + \frac{1}{N^t (N^t - 1)} \sum_{i = 1}^{N^t}\sum_{j \neq i}^{N^t}\delta_{ak}(\textbf{x}^t_i, \textbf{x}^t_j;\Theta_F) \\
    & - \frac{2}{N^s N^t}\sum_{i=1}^{N^s}\sum_{j=1}^{N^t}\delta_{ak}(\textbf{x}^s_i, \textbf{x}^t_j;\Theta_F).
\end{split}
    \nonumber
\end{equation}
Suppose $n = N^s = N^t$, we can get a slightly simpler empirical estimator
\begin{equation}
    \hat{d}_u^2(\mathcal{D}^s, \mathcal{D}^t, \delta_{ak}) = \frac{1}{n(n-1)}\sum_{i\neq j}^n \mathcal{M}(u_i,u_j;\Theta_F),
    \nonumber
\end{equation}
where $u_i$ denote the pair $(\textbf{x}_i^s, \textbf{x}_i^t)$, and
\begin{equation}
    \mathcal{M}_{ak}(u_i, u_j;\Theta_F) = \delta_{ak}(\textbf{x}^s_i, \textbf{x}^s_j;\Theta_F) + \delta_{ak}(\textbf{x}^t_i, \textbf{x}^t_j;\Theta_F)-\delta_{ak}(\textbf{x}^s_i, \textbf{x}^t_j;\Theta_F) - \delta_{ak}(\textbf{x}^s_j, \textbf{x}^t_i;\Theta_F).
    \nonumber
\end{equation}
Obviously, the $d_u^2 = \mathbb{E}[\mathcal{M}_{ak}(u_1,u_2;\Theta_F)]$.
The variance of $\hat{d}_u^2$ is 
\begin{equation}
    \mathrm{Var}[\hat{d}_u^2] = \frac{4(n-2)}{n(n-1)}\zeta_1 + \frac{2}{n(n-2)}\zeta_2 = \frac{4}{n}\zeta_1 + \frac{2\zeta_2 - 4\zeta_1}{n(n-1)},
\end{equation}
where 
$$\zeta_1 = \mathbb{E}[\mathcal{M}_{ak}(u_1, u_2;\Theta_F) \mathcal{M}_{ak}(u_1, u_3;\Theta_F)] - \mathbb{E}[\mathcal{M}_{ak}(u_1,u_2;\Theta_F)]^2,$$
$$ \zeta_2 = \mathbb{E}[(\mathcal{M}_{ak}(u_1, u_2;\Theta_F))^2]-\mathbb{E}[\mathcal{M}_{ak}(u_1, u_2;\Theta_F)]^2.$$
Thus as $n \rightarrow \infty$, 
$$n \mathrm{Var}[\hat{d}_u^2] \rightarrow 4\zeta_1 = \sigma^2_{\mathrm{H}_1}$$
\begin{proof}
According to Lemma A in Section 5.2.1 of \cite{serfling2009approximation}, the following equation is hold
\begin{equation}
\begin{split}
\zeta_1 & = \mathrm{Var}_u[\mathbb{E}_{u^{'}}[\delta_{ak}(u,u^{'};\Theta_F)]] \\
& = \mathbb{E}_u[\mathbb{E}_{u^{'}}[\delta_{ak}(u, u^{'};\Theta_F)]]\mathbb{E}_{u^{''}}[\delta_{ak}(u,u^{''};\Theta_F)]-\mathbb{E}_u[\mathbb{E}_{u^{'}}[\delta_{ak}(u, u^{'};\Theta_F)]]^2 \\
& = \mathbb{E}[\mathcal{M}(u_1, u_2;\Theta_F) \mathcal{M}(u_1, u_3;\Theta_F)] - \mathbb{E}[\mathcal{M}(u_1, u_2;\Theta_F)]^2
\end{split}\nonumber
\end{equation}
and
\begin{equation}
    \zeta_2 = \mathrm{Var}_{u,u^{'}}[\mathcal{M}(u,u^{'};\Theta_F)] = \mathbb{E}[(\mathcal{M}(u_1, u_2;\Theta_F))^2] - \mathbb{E}[\mathcal{M}(u_1, u_2;\Theta_F)]^2.
\nonumber
\end{equation} 
\end{proof}

We conduct a hypothesis H$_0:\mathcal{P}^s = \mathcal{P}^t$ and alternative H$_1:\mathcal{P}^s \neq \mathcal{P}^t$.
Then, the $\hat{d}_u^2$ converges in distribution to a Gaussian according to 
\begin{equation}
    n^{\frac{1}{2}}(d^2-\hat{d}^2_u)\stackrel{\mathbf{D}}{\rightarrow} \mathcal{N}(0, \sigma_{\rm{H_1}}^2).
    \nonumber
\end{equation}
Although \cite{sutherland2019unbiased} give a quadratic-time estimator unbiased for $\sigma_{\mathrm{H}_1}$, it is much more complicated to implement and analysis, and can be negative, especially when the kernel is inappropriate.
Thus, we adopt a $V$-statistic $\hat{\sigma}_{\mathrm{H}_1}^2$ as the estimator of $\sigma^2_{\mathrm{H}_1}$:

\begin{equation}
    \hat{\sigma}^2_{\mathrm{H}_1} = 4\left[\frac{1}{n}\left(\frac{1}{n}\sum_{j=1}^n \mathcal{M}(u_i, u_j;\Theta_F)\right)^2 - \left(\frac{1}{n^2}\sum_{i=1}^{n}\sum_{j=1}^n \mathcal{M}(u_i, u_j;\Theta_F)\right)^2\right].
\end{equation}
For a given significance level $\alpha$, we choose a test threshold $c_\alpha$ and reject H$_0$ if $n \hat{d}_u^2 > c_\alpha$.
The power of our test is,
\begin{equation}
    {\rm{Pr}}_{1}\left(n\hat{d}_u^2 > c_\alpha\right)={\rm{Pr}}_{1}\left(\frac{{\hat{d}}_u^2-{\hat{d}}^2}{\hat{\sigma}_{\rm{H_1}}} > \frac{c_\alpha / n - {d}^2}{\hat{\sigma}_{\rm{H}_1}}\right)\rightarrow \Phi\left(\frac{\sqrt{n}{d}^2}{\hat{\sigma}_{\rm{H}_1}} - \frac{c_\alpha}{\sqrt{n}\hat{\sigma}_{\rm{H}_1}}\right).
\end{equation}
The $\rm{Pr}_1$ denote the probability under $\rm{H}_1$, and $\Phi$ is the cumulative distribution function of the standard normal distribution.
When $n \rightarrow \infty$, the test power is dominated by the first term.

Note that, in our implementation, we use $\hat{\sigma}_{\mathrm{H}_1,\lambda}$ as the estimator of $\sigma_{\mathrm{H}_1}$ instead of $\hat{\sigma}_{\mathrm{H}_1}$ directly.
Further, we prove that the $\hat{\sigma}_{\mathrm{H}_1}$ is replaced by its regularization $\hat{\sigma}_{\mathrm{H}_1, \lambda}$ is feasible.  
\begin{theorem}
We assume the kernels $\delta_{ak}$ are uniformly bounded in a Banach space of dimension $\mathcal{B}_d$ with the set of possible kernel parameters $\mathcal{F}$, which is bounded by $||\Theta_F|| \leq \mathcal{R}_F$, i.e., $$\sup_{\Theta_F \in \mathcal{F}} \sup_{x}\delta_{ak}(x,x;\Theta_F) \leq v,$$
and the kernel parameterization is Lipschitz, that for all $x^s, x^t $ and $\Theta_F, \Theta_F^{'} \in \mathcal{F}$, 
$$|\delta_{ak}(x^s, x^t;\Theta_F) - \delta_{ak}(x^s, x^t; \Theta_F^{'})| \leq L_{\delta} ||\Theta_F - \Theta_F^{'} ||.$$
Let $\Bar{\mathcal{F}}_s$ be the set of kernel parameters for which $\mathrm{\sigma_{H_1}^2}\geq s^2 > 0$
and let $\lambda = n^{-1/3}$.
Then, with probability $1 - \beta$,
\begin{equation}
\sup_{\Theta_F \in \Bar{\mathcal{F}}_s}|\frac{\hat{d}^2_u}{\hat{\sigma}_\mathrm{{H_1,\lambda}}}-\frac{d^2}{\sigma_\mathrm{H_1}}| = \mathcal{O}\left(\frac{1}{s^2 n^{1/3}}\left[\frac{1}{s} + L_\delta +  \sqrt{\mathcal{B}_d}\right]\right)
\label{ap4}
\end{equation}
\end{theorem}

\begin{proof}
The $\hat{\sigma}^2_{\mathrm{H_1},\lambda} = \sigma^2_{\mathrm{H_1}} + \lambda$, and $|\hat{d}^2_u| \leq 4v$. 
Thus, the Eq.(\ref{ap4}) can be decomposed as follows, 
\begin{equation}
\begin{split}
\sup_{\Theta_F \in \Bar{\mathcal{F}}_{s}} |\frac{\hat{d}_u^2}{\hat{\sigma}_{\mathrm{H_1},\lambda}} - \frac{d^2}{\sigma_\mathrm{H_1}}| &\leq 
\sup_{\Theta_F \in \Bar{\mathcal{F}}_{s}}|\frac{\hat{d}^2_u}{\hat{\sigma}_{\mathrm{H_1},\lambda}}-\frac{\hat{d}^2_u}{\sigma_{\mathrm{H_1},\lambda}}|+
\sup_{\Theta_F \in \Bar{\mathcal{F}}_{s}}|\frac{\hat{d}^2_u}{\sigma_{\mathrm{H_1},\lambda}}-\frac{\hat{d}^2_u}{\sigma_{\mathrm{H_1}}}|+
\sup_{\Theta_F \in \Bar{\mathcal{F}}_{s}}|\frac{\hat{d}^2_u}{\sigma_{\mathrm{H_1}}}-\frac{{d}^2}{\sigma_{\mathrm{H_1}}}| \\
&= \sup_{\Theta_F \in \Bar{\mathcal{F}}_s}|\hat{d}_u^2|\frac{1}{\hat{\sigma}_{\mathrm{H_1},\lambda}}\frac{1}{\sigma_{\mathrm{H}_1,\lambda}}\frac{|\hat{\sigma}^2_{\mathrm{H}_1,\lambda}-\sigma^2_{\mathrm{H}_1,\lambda}|}{\hat{\sigma}_{\mathrm{H}_1, \lambda} + \sigma_{\mathrm{H}_1,\lambda}} + 
\sup_{\Theta_F \in \Bar{\mathcal{F}}_s}|\hat{d}_u^2|\frac{1}{\hat{\sigma}_{\mathrm{H_1},\lambda}}\frac{1}{\sigma_{\mathrm{H}_1}}\frac{|\hat{\sigma}^2_{\mathrm{H}_1,\lambda}-\sigma^2_{\mathrm{H}_1}|}{\hat{\sigma}_{\mathrm{H}_1, \lambda} + \sigma_{\mathrm{H}_1}} \\ 
&+ \sup_{\Theta_F \in \Bar{\mathcal{F}}_s}\frac{1}{\sigma_{\mathrm{H}_1}}|\hat{d}^2_u - d^2| \\
&\leq \sup_{\Theta_F \in \Bar{\mathcal{F}}_s} \frac{4v}{\sqrt{\lambda}s(s + \sqrt{\lambda})}|\hat{\sigma}^2_{\mathrm{H}_1} - \sigma^2_{\mathrm{H}_1}| + \frac{4 v \lambda}{\sqrt{s^2 + \lambda}s(\sqrt{s^2 + \lambda} + s)} + \sup_{\Theta_F \in \Bar{\mathcal{F}}_s}\frac{1}{s}|\hat{d}^2_u - d^2| \\
& \leq \frac{4v}{s^2 \sqrt{\lambda}}\sup_{\Theta_F \in \mathcal{F}}|\hat{\sigma}^2_{\mathrm{H_1}} - \sigma^2_{\mathrm{H_1}}| + \frac{2v}{s^3}\lambda + \frac{1}{s}\sup_{\Theta_F \in \mathcal{F}}|\hat{d}^2_u - d^2|
\end{split}
\nonumber
\end{equation}
According to the Supplementary of \cite{liu2020learning1},
the $\hat{\sigma}_{\mathrm{H_1}}$ and $\hat{d}_u^2$ are uniform convergence. 
Then, with least probability $1 - \beta$, the error bound is, 
\begin{equation}
    \frac{2v}{s^3}\lambda + \left[\frac{8v}{s\sqrt{n}} + \frac{1792v}{\sqrt{n}s^2\sqrt{\lambda}}\right]\sqrt{2\log\frac{2}{\beta} + 2\mathcal{B}_d\log(4\mathcal{R}_F\sqrt{n})}+\left[\frac{8}{s\sqrt{n} }+ \frac{2048 v^2}{\sqrt{n}s^2\sqrt{\lambda}}\right]L_{\delta} + \frac{4608v^3}{s^2 n \sqrt{\lambda}}.
\nonumber
\end{equation}
Then, taking $\lambda = n^{-\frac{1}{3}}$,
\begin{equation}
\begin{split}
    \frac{2v}{s^3 n^{1/3}} + \left[\frac{8v}{s\sqrt{n}} + \frac{1792v}{s^2 n^{1/3}}\right]\sqrt{2 \log \frac{2}{\beta} + 2 \mathcal{B}_d \log(4\mathcal{R}_F\sqrt{n})} + \left[\frac{8}{s\sqrt{n}} + \frac{2048v^2}{s^2 n^{1/3}}\right]L_{\delta} + \frac{4608v^3}{s^2 n^{5/6}} \\
    \leq 
    \frac{2v}{s^3 n^{1/3}} + \left[\frac{8v}{s\sqrt{n}} + \frac{2048v^2}{s^2n^{1/3}}\right]
    \left[L_{\delta} + \sqrt{2\log\frac{2}{\beta}+2\mathcal{B}_d\log(4\mathcal{R}_F \sqrt{n})}\right] + \frac{4608 v^3}{s^2 n^{5/6}}.
\end{split}
\nonumber
\end{equation}
Treating $v$ as a constant, we can get the Eq.(\ref{ap4}).
\end{proof}

 \bibliographystyle{elsarticle-num} 
 \bibliography{ref}




\end{document}